\documentclass[12pt]{article}
\usepackage{lMac}
\usepackage{becMac}

\def\showintremarks{n}   %%%    y displays internal remarks
\def\showissues{n}       %%%    y displays "issue" internal remarks

\newcommand{\Su}{S}

\begin{document}
\title{The Small Field Parabolic Flow for Bosonic Many--body Models:\\
        \Large Part 1 --- Main Results and Algebra}

\author{Tadeusz Balaban}
\affil{\small Department of Mathematics \authorcr
       % Hill Center--Busch Campus \authorcr
       Rutgers, The State University of New Jersey \authorcr
     %  110 Frelinghuysen Rd \authorcr
      % Piscataway, NJ 08854-8019 \authorcr
       tbalaban@math.rutgers.edu\authorcr
       \  }

\author{Joel Feldman\thanks{Research supported in part by the Natural 
                Sciences and Engineering Research Council 
                of Canada and the Forschungsinstitut f\"ur 
                Mathematik, ETH Z\"urich.}}
\affil{Department of Mathematics \authorcr
       University of British Columbia \authorcr
     %  Vancouver, B.C.,   %\authorcr
      %  Canada\ \   V6T 1Z2 \authorcr
       feldman@math.ubc.ca \authorcr
       http:/\hskip-3pt/www.math.ubc.ca/\squig feldman/\authorcr
       \  }

\author{Horst Kn\"orrer}
\author{Eugene Trubowitz}
\affil{Mathematik \authorcr
       ETH-Z\"urich \authorcr
      %  ETH-Zentrum \authorcr
      % CH-8092 Z\"urich, %\authorcr
      %  Switzerland \authorcr
       knoerrer@math.ethz.ch, trub@math.ethz.ch \authorcr
       http:/\hskip-3pt/www.math.ethz.ch/\squig knoerrer/}

%\date{\editdate}

\maketitle

\begin{abstract}
\noindent
This paper is a contribution to a program to see symmetry breaking in a
weakly interacting many Boson system on a three dimensional lattice at 
low temperature.  It is part of an analysis of the ``small field''  approximation 
to the ``parabolic flow'' which exhibits the formation of a ``Mexican hat''
potential well. Here we state the main result of this analysis, outline the 
strategy of the proof, which uses a renormalization group flow, and 
perform the first, algebraic, part of a renormalization group step.

\end{abstract}

\newpage
\tableofcontents

\newpage
%%%%%%%%%%%%%%%%%%%%%%%%%%
\section{Introduction}
%%%%%%%%%%%%%%%%%%%%%%%%%%

%%%%%%%%%%%%%%%%%%%%%%%
\subsection{The model and dominant contributions to the effective action}\label{sectINTmodel}
%%%%%%%%%%%%%%%%%%%%%%%%

An interacting many Boson system on a three--dimensional lattice in 
thermodynamic equilibrium is characterized by
a single particle ``kinetic energy'' operator $\bh$ on this lattice\footnote{The most commonly used $\bh=-\sfrac{1}{2m}\Delta$,  where $\Delta$ is the lattice Laplacian.}, a translation invariant two--body potential, 
$\bv$ that describes the particle--particle interaction, the temperature $T>0$ and the chemical potential $\mu$. 
This paper is a contribution to a program that is to provide a 
mathematically rigorous investigation of the partition and correlation functions of such a gas of bosons. For simplicity, we assume that the underlying lattice has been scaled to be the unit lattice $\bbbz^3$.  
We also assume that both the two--body potential, $\bv$ and the kernel of 
$\bh$ decay exponentially in the distance between their arguments, 
and that $\bv$ is the kernel of a strictly positive operator.

It is a standard strategy for the investigation of such a system to 
consider it as a limit of the correponding systems with a periodic cutoff 
$L_\sp$, as this infrared cutoff tends to infinity.
The system with  periodic cutoff $L_\sp$ is defined on the finite lattice
$$
X=\bbbz^3/L_\sp\bbbz^3
$$ 
and is characterized by the periodizations\footnote
{The periodization of a translation invariant function 
$ f(\xi_1,\cdots,\xi_n)$ on ${({\sst\bbbz}^3)}^n$ is the function 
on $X^n$ that maps $(\bx_1,\cdots,\bx_n) \in X^n$ to 
$\sum\limits_{\atop{\xi_2,\cdots \xi_n}{[\xi_i]=\bx_i}} f(\xi_1,\cdots,\xi_n) $,
where $\xi_1$ is an arbitrary point in ${\sst\bbbz}^3$ whose class $[\xi_1]$ 
in $X$ equals $\bx_1$. The periodization of an operator is the operator whose kernel is the periodization of the given operator.} 
$h$ of $\bh$ and $v$ of $\bv$. 

In previous papers we started an investigation of the partition function 
of the periodized system
$$
\Tr e^{-\sfrac{1}{kT} (H-\mu N)}
$$
where $H$ is the second quantized Hamiltonian and $N$ is the number operator.
In \cite{fnlint1,fnlint2}, we represented this partition function in terms of  
coherent state functional integrals (see also \cite{NO}) and then, 
in \cite{UV}, using 
``decimation'', controlled the ``temporal ultraviolet limit'' to obtain 
the following representation for the partition function. (The precise 
hypotheses are specified in \cite[\S2]{UV}.)

There exists a constant $\th>0$ and a function
$\,I_\theta(\al_* ,\be)\,$ of two complex valued fields $\,\al_*\,$ and
$\,\be$ on $X$ such that
\begin{align}\label{eqnTHuvoutput}
&\hskip-7pt\Tr\, e^{-\sfrac{1}{kT}\,(H-\mu N)} =
\int \hskip-28pt\prod_{\hskip22pt\lower5pt\hbox{$\sst\tau\in\th\bbbz
\cap(0,\nicefrac{1}{kT}]$}}\hskip-22pt 
\Big[ \smprod_{\bx\in X} 
\sfrac{ d\al_\tau(\bx)^\ast\wedge d\al_\tau(\bx)}{2\pi \imath} \,
                      e^{-\al_\tau(\bx)^\ast \al_\tau(\bx)}\Big]
I_\theta(\al_{\tau-\th}^\ast ,\al_\tau) 
\end{align}
One can write $\,I_\theta\,$ as the sum of a dominant part 
$\,I_\theta^{(SF)}\,$, called the \emph{pure small field contribution},  
and terms, indexed by proper subsets of $X$,  which are nonperturbatively 
small\footnote{We call a function nonperturbatively small, if it is
of order $O(e^{-1/\|v\|^\veps})$ for some norm on $v$ and some $\veps>0$.
A precise bound is given in \cite[Theorem 2.18]{UV}},
exponentially in the size of the subsets. The properties of the 
function $\,I_\theta^{(SF)}\,$
are reviewed in \S\ref{sectSZrewriteModel} and \S\ref{sectSZrewriteOutput}.

We want to control the integrals in the representation \eqref{eqnTHuvoutput} 
of the partition function
uniformly in small temperature $T$ and lattice size $L_\sp$ to rigorously
establish the phase transition in the many particle system of bosons, when
the chemical potential $\,\mu\,$ lies sufficiently above a certain 
critical value. This phase transition is intimately related to the 
formation of a ``mexican hat'' shaped potential well in the effective 
action. See, for example, \cite{Col} and \cite[\S19]{Wein} for an introduction
to symmetry breaking in general, and \cite{AGD,BOG,FW,PS} as 
general references to Bose-Einstein condensation. See \cite{Ben,CG,LSSY,Sei} 
for other mathematically rigorous work on the subject.

In this paper, we replace the function $\,I_\th\,$ in \eqref{eqnTHuvoutput} 
by $\,I_\theta^{(SF)}\,$, that is, we study
\begin{equation}\label{eqnTHsmallfieldoutput}
\int \hskip-28pt\prod_{\hskip22pt\lower6pt\hbox{$\sst\tau\in\th\bbbz
\cap(0,\nicefrac{1}{kT}]$}}\hskip-22pt 
\Big[ \smprod_{\bx\in X} 
\sfrac{ d\al_\tau(\bx)^\ast\wedge d\al_\tau(\bx)}{2\pi \imath} \,
                      e^{-\al_\tau(\bx)^\ast \al_\tau(\bx)}\Big]
I_\theta^{(SF)}(\al_{\tau-\th}^\ast ,\al_\tau) 
\end{equation}
Using this model, we exhibit the mechanism that leads to the onset of
the potential well. A full fledged large field/small field analysis of \eqref{eqnTHuvoutput} will be  performed later.

To simplify the discussion, we assume that $\,L_\sp$ and 
$\,L_\tp=\sfrac{1}{\th k T}\,$ are powers of some  odd natural number 
$\,L>2\,$. 
After rescaling the ``temporal lattice'' $\th\bbbz/\sfrac{1}{kT}\bbbz$
to $\bbbz/\sfrac{1}{\th kT}\bbbz$, \eqref{eqnTHsmallfieldoutput} can be viewed as 
a functional integral over fields on the lattice
\begin{equation*}
\cX_0 
%= \big( \bbbz \times \bbbz^3 \big) / 
% \big(L_\tp\bbbz \times L_\sp\bbbz^3 \big)
= \big( \bbbz / L_\tp\bbbz\big) \times 
  \big(\bbbz^3 / L_\sp\bbbz^3 \big)
\end{equation*}
For a point $\,x=(x_0,\bx) \in \cX_0\,$, we call $\,x_0\,$ its time and
$\,\bx\,$ its spatial component. The ``real inner product'' for functions $\,f,g\,$ on $\cX_0$ is
$\,\<f,g\>_0 = \smsum_{x\in\cX_0} f(x)g(x)\,$.

In Proposition  \ref{propSZprepforblockspin}, we show that, up to 
a normalization constant,   the integral  \eqref{eqnTHsmallfieldoutput} can be written in the form
\begin{equation}\label{eqnHTstartingpoint}
   \int \Big[ \smprod_{x\in \cX_0} 
        \sfrac{ d\psi(x)^\ast\wedge d\psi(x)}{2\pi \imath}\Big] \,
  e^{\cA_0(\psi^*,\psi) }\chi_0(\psi)
\end{equation}
where
\begin{equation}\label{eqnHTaIn}
\cA_0(\psi_*,\psi) = -\<\psi_*,\,D_0\psi\>_0
       -\cV_0(\psi_*,\psi)
       +\mu_0 \<\psi_*,\,\psi\>_0
       +\cR_0(\psi_*,\psi) 
       +\cE_0(\psi_*,\psi)
\end{equation}
and
\begin{itemize}[leftmargin=*, topsep=2pt, itemsep=2pt, parsep=0pt]
%%%%%%%%%%%%%%%%%%%%%%%%%
\item
%%%%%%%%%%%%%%%%%%%%%%%%%
$
D_0=\bbbone - e^{-\oh_0} -e^{-\oh_0} \partial_0
\,$ with $\oh_0 = \th\oh$ and  $\,\partial_0$  the forward time derivative 
$\,
(\partial_0f)(x_0,\bx) = f(x_0+1,\bx) - f(x_0,\bx)
,$
%%%%%%%%%%%%%%%%%%%%%%%%%%%%%
\item
%%%%%%%%%%%%%%%%%%%%%%%%%%%%%
$
\cV_0(\psi_*,\psi)=\half\smsum\limits_{x_1,\cdots,x_4\in\cX_0}\hskip-10pt   
                  V_0(x_1,x_2,x_3,x_4)\,  
                  \psi_*(x_1) \psi(x_2)\psi_*(x_3) \psi(x_4)
$ 
is a quartic monomial whose translation invariant kernel $V_0$ is 
determined by $\,v\,$ and $\,\oh\,$. 
It is invariant under $x_1\leftrightarrow x_3$ and under 
$x_2\leftrightarrow x_4$.
Its average 
 $\rv_0=\smsum\limits_{x_2,x_3,x_4\in\cX_0}  V_0(0,x_2,x_3,x_4)$ is positive.
The kernel $V_0(x_1,x_2,x_3,x_4)$ is the spatial periodization of a 
translation invariant, exponentially decaying kernel $\bV_0$
on $\big((\bbbz/L_\tp\bbbz) \times \bbbz^3 \big)^4$.
%%%%%%%%%%%%%%%%%%%%%%%%%%%%%%%
\item $\mu_0 $ is close to $\th \mu$.
%%%%%%%%%%%%%%%%%%%%%%%%%%%%%%
\item $\cR_0(\psi_*,\psi)$ and $\cE_0(\psi_*,\psi)$ are 
perturbatively small, particle--number preserving functions.  For 
the different characteristics and roles of $\cR_0$ and $\cE_0$, 
see Proposition  \ref{propSZprepforblockspin} and Theorem  
\ref{thmTHmaintheorem}.
%%%%%%%%%%%%%%%%%%%%%%%%%%%%%%%%%%%%%%%%
\item $\,\chi_0(\psi)\,$ is a ``small field cut off function''.
\end{itemize}
$\,\oh_0\,$ acts only on the spatial variables of a function of 
$\,x=(x_0,\bx) \in \cX_0\,$. Observe that $\,\psi^*\,$ denotes the complex conjugate of the field $\psi$,
while $\psi_*$ and $\psi$ are treated as two independent complex valued
fields on $\cX_0$.

\noindent
More details, including precise estimates, are given in 
Proposition  \ref{propSZprepforblockspin}.

For a constant field $\,\psi(x) =\psi\,$, the dominant part, 
$-\cA_0(\psi^*,\psi)\,$, of (minus) the action  in \eqref{eqnHTstartingpoint},
reduces to
\begin{equation*}
\cV_0(\psi^*,\psi) -\mu_0 \<\psi^*,\,\psi\>_0
= |\cX_0| \big[ \half \rv_0 |\psi|^4 -\mu_0 |\psi|^2 \big]
= |\cX_0| \big[ \half \rv_0 \big(|\psi|^2 -\sfrac{\mu_0}{\rv_0} \big)^2 - \sfrac{\mu_0^2}{2\rv_0}\big]
\end{equation*}
and has a potential well. If $\mu_0$ is of the order of $\rv_0$, this well
is quite shallow. Using block spin transformations, as in 
\cite{KAD,BalLausane,GK}, (see Definition \ref{defHTblockspintr}
below) we will successively perform parts of the integral and show that the effective action after these block spin
transformations has a much better developed potential well. 
We expect that the result of this paper will be the starting point for an analysis 
that is adapted to the symmetry breaking caused by the degenerate ground state 
in this potential well.

We believe (see the discussion before  \cite[Lemma \lemOSImustar]{PAR2})
that the scenario described above holds whenever the chemical potential 
$\mu$ is bigger than some critical value that, to leading order in $\bv$, should be
\begin{equation}\label{eqnHTdefcriticalmu}
 2  \int_{\bbbr^3 /2\pi\bbbz^3}  \sfrac{d^3\bk}{(2\pi)^3} 
\sfrac{\hat\bv(\bZ)+\hat \bv(\bk) }{ e^{\hat\bh(\bk)/kT}-1}
\end{equation}
where $\hat\bh(\bk)$ and $\hat \bv(\bk)$ are the Fourier transforms 
of $\bh(\bx,\bZ)$ and $\bv(\bx,\bZ)$. See \eqref{eqnINTmustardef}, 
Lemma \ref{lemSZexplicitmustar} and Corollary  \ref{corSZstepzeroConditions}.
Also observe that \eqref{eqnHTdefcriticalmu} converges to zero as $\be\rightarrow\infty$.
In this paper we assume
$\bv$ is small and that $\mu$ is bigger than \eqref{eqnHTdefcriticalmu}
by a number that is at least a norm of $\bv$ raised to a power that 
is a bit bigger than one. For details see \eqref{eqnINTmustardef}.

After $n$ block spin (and scaling) transformations, the partition function
will be represented by a functional integral on the 
lattice\footnote{We shall define a family of lattices 
$\cX_j^{(n)}$ in Definition \ref{defHTbackgrounddomaction}.}
\begin{equation*}
\cX_0^{(n)} =  \big( \bbbz \times \bbbz^3 \big) / 
 \big(\sfrac{L_\tp}{L^{2n}}\bbbz \times \sfrac{L_\sp}{L^n}\bbbz^3 \big)
\end{equation*}
where $L>2$ is the odd natural number chosen above.
The asymmetry in the time and spatial variables arises from the 
``parabolic scaling'' of Definition \ref{defHTscaling}, below.

\begin{definition}[Blockspin Transformation]\label{defHTblockspintr}
Fix a nonnegative even function $\,q\,$ in $L^1(\bbbz\times\bbbz^3)$.
\begin{enumerate}[label=(\alph*), leftmargin=*]
\item
For a field $\psi$ on $\,\cX_0^{(n)}\,$ define the ``averaged '' field 
$\,Q\psi\,$ on 
\begin{equation*}
\cX_{-1}^{(n+1)} =  \big( L^2 \bbbz \times L\bbbz^3 \big) / 
 \big(\sfrac{L_\tp}{L^{2n}}\bbbz \times \sfrac{L_\sp}{L^n}\bbbz^3 \big)
\end{equation*}
by
\begin{equation*}
(Q\psi)(y) = \smsum_{x\in \bbbz\times\bbbz^3} q(x) \psi(y+[x])
\end{equation*}
where $\,[x]\,$ denotes the class of $\,x\in\bbbz\times\bbbz^3\,$ in the
quotient space $\,\cX_0^{(n)}\,$.

\item
If $\,F(\psi_*,\psi)\,$ is a function of complex valued fields
$\,\psi_*,\psi\,$ on $\,\cX_0^{(n)}\,$, we define the 
\emph{block spin transform}
of $\,F\,$ (with respect to $q$ and a constant $a>0$) as the function 
\begin{equation*}
(\bbbt F)(\th_*,\th)
=\sfrac{1}{N^{(n)}_\bbbt}\hskip-2pt\int\hskip-2pt\Big[\hskip-3pt
 \prod_{x\in\cX_0^{(n)}}\hskip-3pt\sfrac{d\psi(x)^*\wedge d\psi(x)}{2\pi i}\Big]
e^{-a L^{-2}\< \th_*\!-Q\psi^*,\th-Q\psi\>_{-1}}\,F(\psi^*,\psi)
\end{equation*}
of the fields $\,\th_*,\th\,$ on $\,\cX_{-1}^{(n+1)}\,$.
Here, for any two fields $\,f,g\,$ on $\,\cX_{-1}^{(n+1)}\,$
\begin{equation*}
\< f,g\>_{-1} = L^5 \smsum_{y\in \cX_{-1}^{(n+1)}} f(y) g(y)
\end{equation*}
and the normalization constant is
$\,
N^{(n)}_\bbbt
= \int\Big[
 \prod_{y\in\cX_{-1}^{(n+1)}}\hskip-3pt
  \sfrac{d\th(y)^*\wedge d\th(y)}{2\pi i}\Big]
e^{-a L^{-2}\< \th^*,\th\>_{-1}}
\,$. We choose $a=1$.
\end{enumerate}
\end{definition}

\begin{remark}\label{remHTpropblockspin}
\begin{enumerate}[label=(\alph*), leftmargin=*]
\item
As
\begin{equation*}
1=\sfrac{1}{N^{(n)}_\bbbt}\hskip-2pt\int\hskip-2pt\Big[\hskip-5pt
 \prod_{y\in\cX_{-1}^{(n+1)}}\hskip-3pt
  \sfrac{d\th(y)^*\wedge d\th(y)}{2\pi i}\Big]
e^{-a L^{-2}\< \th^*\!-Q\psi^*,\th-Q\psi\>_{-1}}
\end{equation*}
for all functions $\,F(\psi_*,\psi)\,$  fields on $\,\cX_0^{(n)}\,$
\begin{equation*}
\int\hskip-2pt\Big[\hskip-5pt
 \prod_{x\in\cX_0^{(n)}}\hskip-3pt
   \sfrac{d\psi(x)^*\wedge d\psi(x)}{2\pi i}\Big]
   F(\psi^*,\psi)
= \int\hskip-2pt\Big[\hskip-5pt
 \prod_{y\in\cX_{-1}^{(n+1)}}\hskip-3pt
   \sfrac{d\th(y)^*\wedge d\th(y)}{2\pi i}\Big]
  (\bbbt F)(\th^*,\th)
\end{equation*}

\item
We make a specific choices of $q$ in Definition
\ref{defHTbasicnorm}.d. The main result, Theorem \ref{thmTHmaintheorem},
will apply for all sufficiently large $L$.
\end{enumerate}
\end{remark}

The lattice $\cX_{-1}^{(n+1)}$ is coarser than the unit lattice $\cX_0^{(n)}$. 
We choose to scale it back to a unit lattice.

\begin{definition}[Scaling]\label{defHTscaling}
Let $\bbbl$ be the linear isomorphism
\begin{equation*}
\bbbl:   \cX_0^{(n+1)} \rightarrow \cX_{-1}^{(n+1)}  
\qquad\qquad  (x_0,\bx) \mapsto (L^2x_0,L\bx)
\end{equation*}
For a field $\,\th\,$ on $\cX_{-1}^{(n+1)} $, we define the scaled field 
\begin{equation*}
(\bbbs\th)(x)= L^{3/2}\,\th\big(\bbbl x\big)
\end{equation*}
on $\,\cX_0^{(n+1)}\,$. For a function $\,F(\th_*,\th)\,$ of fields on
$\cX_{-1}^{(n+1)} $, we define the function $\,(\bbbs F)(\Psi_*,\Psi)\,$
of fields on $\,\cX_0^{(n+1)}\,$ by
\begin{equation*}
(\bbbs F)(\Psi_*,\Psi)
   = F\big(\bbbs^{-1}\Psi_*,\bbbs^{-1}\Psi)
\end{equation*}
\end{definition}

\begin{remark}\label{remHTbasicremarkonscaling}
\begin{enumerate}[label=(\alph*), leftmargin=*]
\item 
For any function $\,\cF(\th_*,\th)\,$ of fields on
$\cX_{-1}^{(n+1)} $
\begin{equation*}
\int\hskip-2pt\Big[\hskip-5pt
 \prod_{y\in\cX_{-1}^{(n+1)}}\hskip-3pt
   \sfrac{d\th(y)^*\wedge d\th(y)}{2\pi i}\Big]
   \cF(\th^*,\th)
= \sfrac{1}{L^{3|\cX_0^{(n+1)}|}} \int\hskip-2pt\Big[\hskip-5pt
 \prod_{x\in \cX_0^{(n+1)}}\hskip-3pt
   \sfrac{d\Psi(x)^*\wedge d\Psi(x)}{2\pi i}\Big]
  (\bbbs F)(\Psi^*,\Psi)
\end{equation*}

\item 
The exponents $\sfrac{3}{2}$ of $L^{3/2}\,\th\big(\bbbl x\big)$ and 
$2$ of $L^2x_0$ in the definition of $\bbbs$ have been chosen so that
\begin{equation*}
\<\psi^*,\partial_0\psi\>_0
= \<\bbbs^{-1}\psi^*,\partial_0(\bbbs^{-1}\psi)\>_{-1}
\qquad
\<\psi^*,\De\psi\>_0
= \<\bbbs^{-1}\psi^*,\De(\bbbs^{-1}\psi)\>_{-1}
\end{equation*}
That is, $\<\psi^*,\partial_0\psi\>_0$ and $\<\psi^*,\De\psi\>_0$ 
are ``marginal''. Because the time derivative $\partial_0$ is first order
while the spatial Laplacian $\De$ is second order, we refer to 
Definition \ref{defHTscaling} as ``parabolic scaling''.
\end{enumerate}
\end{remark}

Applying Remarks \ref{remHTpropblockspin}.a and \ref{remHTbasicremarkonscaling}.a to
\eqref{eqnHTstartingpoint} we see that, for any natural number $n$ with 
$\,L^n \le \min\{\sqrt{L_\tp},\,L_\sp\}\,$,
\begin{equation}\label{eqnHTmultiBS}
   \int \Big[ \smprod_{x\in \cX_0} 
        \sfrac{ d\psi(x)^\ast\wedge d\psi(x)}{2\pi \imath}\Big] \,
           e^{\cA_0(\psi^*,\psi) }
= \sfrac{1}{\tilde\cZ_n}
\int\hskip-2pt\Big[\hskip-5pt
 \prod_{x\in\cX_0^{(n)}}\hskip-3pt
   \sfrac{d\psi(x)^*\wedge d\psi(x)}{2\pi i}\Big]
   \big[(\bbbs\bbbt)^n \big(e^{\cA_0 } \big)\big](\psi^*,\psi)
\end{equation}
with $\tilde \cZ_n=\smprod_{j=1}^n L^{3|\cX_0^{(j)}|}$.
In this paper we argue that, for $\,n\le \np\,$ (with $\np$ specified
in Definition \ref{defHTbasicnorm}.b), the function 
$\,\big[(\bbbs\bbbt)^n \big(e^{\cA_0  } \big)\big](\psi^*,\psi)\,$ has
--- up to errors which can reasonably be expected
to be nonperturbatively small (see \cite{PARL} or 
\cite[\S2.2.2]{LH}) --- 
a logarithm whose dominant term, $A_n$, described in Definition  
\ref{defHTbackgrounddomaction}.b below, exhibits a much deeper potential
well. See \eqref{eqnINTdeepwell}.

The representation \eqref{eqnHTmultiBS} of the partition function
is built by iterating block spin transformations.
We use stationary phase to analyze each block spin integral, as in 
Definition \ref{defHTblockspintr}.b, in the integrand of \eqref{eqnHTmultiBS}. 
It is then natural to express the (dominant part) of the integrand in
terms of the composition of the stationary phase critical fields for the
various block spin transformations. We shall call this composition the
``background field''. The definition that we are about to give for the
background field does not appear to have anything to do with compositions.
The ``composition law'' Proposition \ref{propBGAomnibus}.b shows that the
background field is indeed a composition of critical fields.

\begin{definition}[Background field and dominant
part of the action]\label{defHTbackgrounddomaction}
\ 
\begin{enumerate}[label=(\alph*), leftmargin=*]
\item
For $\,j\ge -1\,$ and $\,n\ge 0\,$ define the lattices
\begin{equation*}
\cX_{j}^{(n)} 
= \big(\sfrac{1}{L^{2j}}\bbbz \big/ \sfrac{L_\tp}{L^{2(n+j)}}\bbbz\big)
   \times \big(\sfrac{1}{L^{j}}\bbbz^3 \big/
                      \sfrac{L_\sp}{L^{(n+j)}}\bbbz^3 \big)
\end{equation*}
The subscript in $\cX_{j}^{(n)}$ determines the ``coarseness'' of the lattice
--- nearest neighbour points are a distance $\sfrac{1}{L^{2j}}$ apart in
the time direction and a distance $\sfrac{1}{L^j}$ apart in
spatial directions. The superscript in $\cX_{j}^{(n)}$ determines the number
of points in the lattice --- $|\cX_{j}^{(n)}|=|\cX_0|/L^{5n}$ for all $j$.
On $\,\cX_j^{(n)}\,$, we use the integral notation 
$\
\int_{X^{(n)}_j} du = \sfrac{1}{L^{5j}} \smsum_{u\in\cX_j^{(n)} }
\ $.
The maps
\begin{equation*}
\bbbl : \cX_j^{(n)}\rightarrow \cX_{j-1}^{(n)}
    \qquad  \qquad
  (u_0,\bu) \mapsto (L^2u_0, L\bu)
\end{equation*}
are linear isomorphisms. We routinely view $\cX_j^{(n)}$ as a sublattice
of $\cX_{j+k}^{(n-k)}$, for each $1\le k\le n$. We also abbreviate
$\cX_n^{(0)}$ by $\cX_n$. 

\centerline{\null\hskip0.4in\includegraphics{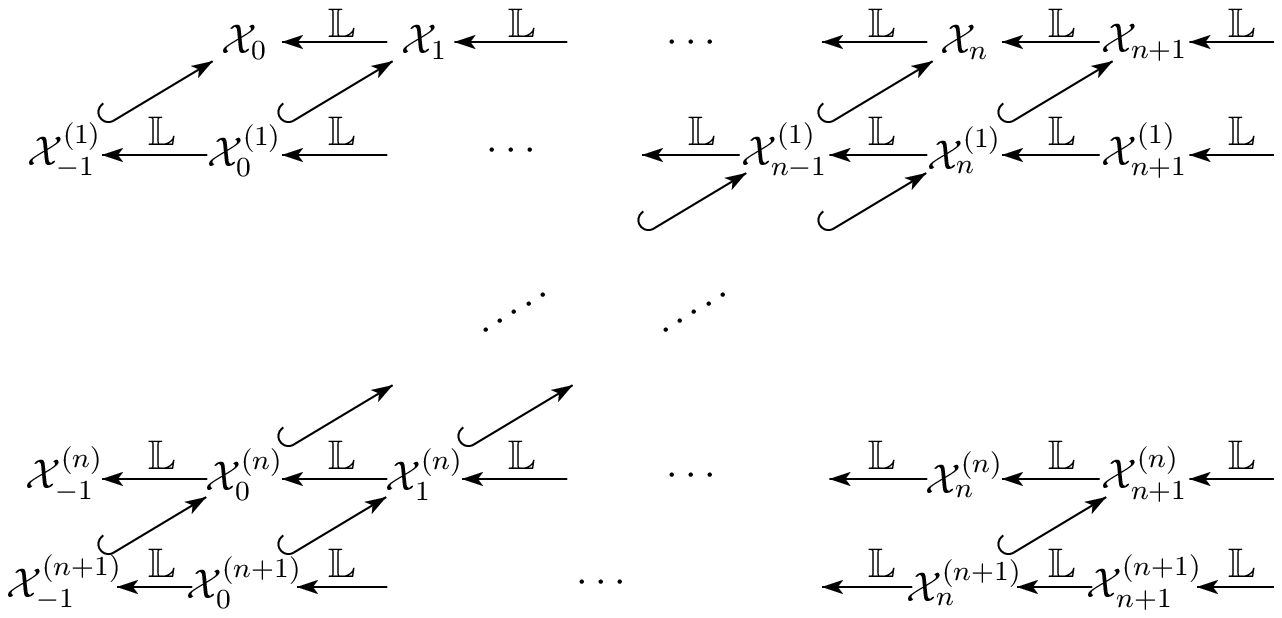}}

We denote by $\,\cH_j^{(n)}\,$ the space of fields on $\cX_j^{(n)}$,
endowed with the {\it real} inner product
\begin{equation*}
\< \al_1,\al_2\>_j
= \int_{X^{(n)}_j} \al_1(u) \,\al_2(u) \ du
\end{equation*}
Set $\cH_n=\cH_n^{(0)}$.
For a field $\,\al\in\cH_j^{(n)}\,$,
define  the field $\,\bbbl_*(\al) \in \cH_{j-1}^{(n)}\,$  by 
$\,
\bbbl_*(\al)(\bbbl u) = \al(u)
\,$,
and the field $\,Q^{(j)}\al \in \cH_{j-1}^{(n+1)}\,$ by
\begin{equation*}
Q^{(j)}\al = \bbbl_*^{-j}\,Q \,\bbbl_*^j\, \al
\end{equation*}
Set
\begin{equation*}
Q_n = Q^{(1)} \circ \cdots \circ Q^{(n)}: 
          \cH_n=\cH_n^{(0)} \rightarrow \cH_0^{(n)} 
\end{equation*}
$Q_j$ is an iterated averaging operation that maps 
the space $\,\cH_j^{(n-j)}\,$ of fields on the fine lattice $\,\cX_j^{(n-j)}\,$
to the space $\,\cH_0^{(n)}\,$ of fields on the unit lattice $\,\cX_0^{(n)}\,$.

\centerline{\includegraphics{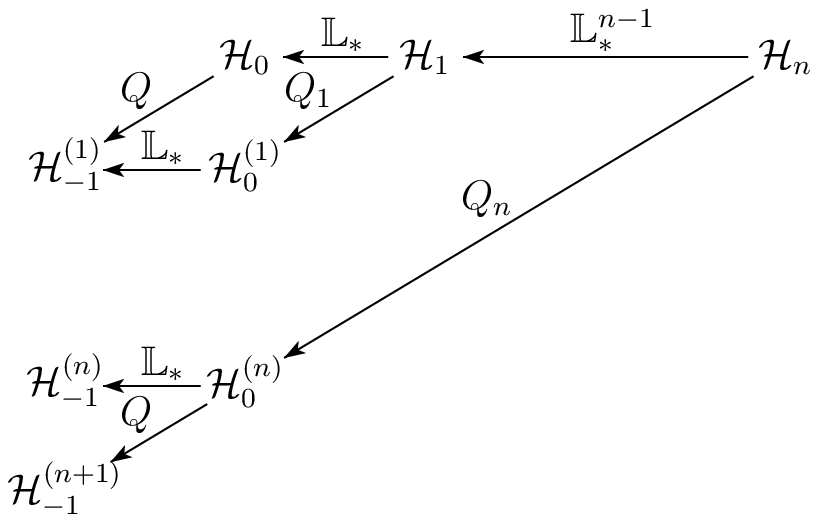}}

\noindent The ``horizontal'' operators $\bbbl_*$ and $\bbbl_*^{n-1}$ are
isomorphisms.

The operator $D_n$ on fields $\phi\in \cH_n$ on $\,\cX_n\,$ is defined
by
\begin{equation*}
D_n = L^{2n}\ \bbbl_*^{-n} \, D_0\, \bbbl_*^n
\end{equation*}

\item
For $\,\mu \in \bbbc\,$, 
fields $\,\psi_*,\psi \in \cH_0^{(n)}\,$ and $\,\phi_*,\phi\in\cH_n$,
and a quartic monomial $\cV$ in the fields $\phi_*,\phi$,
define, for $n\ge 1$,
\begin{equation}\label{eqnHTAndef}
\begin{split}
A_n(\psi_*,\psi,\phi_*,\phi,\mu,\cV)
&=\<\psi_*-Q_n\phi_*,\fQ_n\big(\psi-Q_n\phi\big)\>_0
 +\int_{\cX_n} \phi_*(u)\,(D_n\phi)(u)\,du \\
&\hskip0.5in     + \cV(\phi_*,\phi)
        - \mu\int_{\cX_n} \phi_*(u)\,\phi(u)\,du\\
& =\<\psi_*-Q_n\phi_*,\fQ_n\big(\psi-Q_n\phi\big)\>_0
 +\< \phi_*,\,D_n\phi\>_n \\
&\hskip0.5in     + \cV(\phi_*,\phi)
-\mu \< \phi_*,\,\phi\>_n
\end{split}
\end{equation}
where 
\begin{equation*}
\fQ_n=\begin{cases}
          a\big(\bbbone
             +\smsum\limits_{j=1}^{n-1}\sfrac{1}{L^{2j}}Q_jQ_j^*\big)^{-1} 
             &\text{if $n\ge 2$}\\
             \noalign{\vskip0.05in}
            a\bbbone &\text{if $n=1$}          
       \end{cases}
\end{equation*}
and $Q_j^*$ denotes the transpose of $Q_j$ with respect to the 
``real'' inner product $\,\< f_1,f_2\>_j\,$.

For $n=0$, set
\begin{align*}
A_0(\psi_*,\psi,\mu,\cV)
& = \< \psi_*,\,D_0\psi\>_0  + \cV(\psi_*,\psi)
-\mu \< \psi_*,\,\psi\>_0\\
\end{align*}
In the case $n=0$, we shall use
$\cV=\cV_0$ and $\mu=\mu_0$ so that $A_0$ is the dominant part of the action 
$\cA_0$ of \eqref{eqnHTaIn}.
For $n\ge 0$ we shall use, in $A_n$, a quartic monomial $\cV$ that is a perturbation of
\begin{equation*}
\cV_n^{(u)}(\phi_*,\phi)=\half\int_{\cX_n^4} du_1 \cdots du_4\  
                  V_n^{(u)}(u_1,u_2,u_3,u_4)\,  
                  \phi_*(u_1) \phi(u_2)\phi_*(u_3) \phi(u_4)
\end{equation*}
where
\begin{equation*}
V_n^{(u)}(u_1,u_2,u_3,u_4)
= L^{14 n}\, V_0(\bbbl^n u_1,\bbbl^n u_2,\bbbl^n u_3,\bbbl^n u_4)
\end{equation*}
is the kernel $V_0$ rescaled\footnote{For the origin of  the $L^{14 n}$
see \cite[Lemma \lemSAscaletoscale]{PAR2}.} 
to scale $n$. 
For $\mu$ we will use a 
``renormalized chemical potential'' $\mu_n$ which will be described 
Theorem \ref{thmTHmaintheorem}.

\item
In Proposition \ref{propHTexistencebackgroundfields} and 
\cite[Proposition \propBGEphivepssoln]{BGE},
we solve, for $n\ge 1$, the \emph{background field equations} 
\begin{equation*}
\sfrac{\partial\hfill}{\partial\phi_*}A_n(\psi_*,\psi,\phi_*,\phi,\mu,\cV)=
\sfrac{\partial\hfill}{\partial\phi}A_n(\psi_*,\psi,\phi_*,\phi,\mu,\cV)=0
\end{equation*}
We show that, for $\psi_*,\psi$ sufficiently near $0\in\cH_0^{(n)}$,
                  $\mu\in\bbbc$ sufficiently small, and
                  $\cV$ sufficiently small, 
%                  $\cV$ sufficiently close of $\cV_n^{(u)}$, 
there are fields
\begin{equation*}
\phi_{*n}(\psi_*,\psi,\mu,\cV),\ \phi_{n}(\psi_*,\psi,\mu,\cV)
\end{equation*}
such that 
$\,
\phi_*=\phi_{*n}(\psi_*,\psi,\mu,\cV)\ ,\ 
\phi= \phi_{n}(\psi_*,\psi,\mu,\cV)
\,$ 
solves these equations.
The maps\footnote{We routinely use the ``optional $*$''
notation $\al_{(*)}$ to denote ``$\al_*$ or $\al$''. 
}
 $(\psi_*,\psi,\mu,\cV)\mapsto \phi_{(*)n}(\psi_*,\psi,\mu,\cV)$ 
are analytic and are uniquely determined by this property. 
They are called the \emph{background fields}.
\end{enumerate}
\end{definition}

We shall show that for all $0\le n\le \np$, with the number $\np$ specified in 
Definition \ref{defHTbasicnorm}.b, below, the dominant contribution to 
$\,\big[(\bbbs\bbbt)^n \big(e^{\cA_0 } \big)\big](\psi^*,\psi)\,$
is of the form
\begin{equation*}
\exp\Big\{- A_n\big(\psi_*,\psi,
\ \phi_{*n}(\psi_*,\psi,\mu_n,\cV_n), \phi_{n}(\psi_*,\psi,\mu_n,\cV_n),
\ \mu_n,\cV_n \big) 
+p_n(\psi^*,\psi)\Big\}\chi_n(\psi^*,\psi)
\end{equation*}
with a number  $\,\mu_n \,$, which we call 
the \emph{renormalized chemical potential at scale $n$},
a quartic monomial $\cV_n$ close to the monomial $\cV_n^{(u)}$ of Definition \ref{defHTbackgrounddomaction}.b, which we call the 
\emph{renormalized interaction at scale $n$},
a ``perturbative correction'' $\,p_n(\psi^*,\psi)\,$,
and a ``small field cut off function''
$\,\chi_n\,$ which is discussed in \cite{PARL}. 
The renormalized chemical 
potential $\mu_n$ will grow with $n$ like $L^{2n}(\mu_0-\mu_*)$; 
see Theorem \ref{thmTHmaintheorem}.

For constant, not too big, fields $\,\psi(x) =\psi\,$ and $\psi_*(x)=\psi^*\,$, 
the background field $\,\phi_n\,$ is again constant, again obeys
$\,\phi_{*n}=\phi_n^*\,$ and is approximately  equal to $\,\psi\,$.
See \cite[Remark \remBGEcnstFld]{BGE}. So
$\,
\<\psi_*-Q_n\phi_*,\fQ_n\big(\psi-Q_n\phi\big)\>_0 \approx0
\,$, 
and, in this case, the dominant part of the effective
action $\,A_n(\psi_*,\psi,\ \phi_{*n}, \phi_{n},\ \mu_n,\cV_n) \,$ 
\begin{equation}\label{eqnINTdeepwell}
\sfrac{1}{L^{-5n}|\cX_0|} A_n \approx  
\half\sfrac{\rv_0}{L^n} |\psi|^4-\mu_n |\psi|^2 
= \half \sfrac{\rv_0}{L^n} 
\big(|\psi|^2 - L^{n}\sfrac{\mu_n}{\rv_0} \big)^2 
 - L^{n}\sfrac{\mu_n^2}{2\rv_0}
\end{equation}
has a much better developed  potential well. 

%%%%%%%%%%%%%%%%%%%%%%
\subsection{ The stationary phase approximation}\label{sectINTstatPhase}
%%%%%%%%%%%%%%%%%%%%%%%

We want to argue that for all $1\le n\le \np$, and small fields 
$\psi$, a good approximation to
$\,\big[(\bbbs\bbbt)^n \big(e^{\cA_0 } \chi_0\big)\big](\psi^*,\psi)\,$
is of the form 
\begin{equation}\label{eqnHTbasicformintegr}
\exp\Big\{- A_n\big(\psi^*,\psi, \phi_{*n}, \phi_n,\,\mu_n,\cV_n \big)
+p_n(\psi^*,\psi)
\Big\}\chi_n(\psi^*,\psi)
\end{equation}
with the background fields 
\begin{equation*}
\phi_{*n} = \phi_{*n}(\psi^*,\psi,\mu_n,\cV_n)
\qquad
\phi_{n} = \phi_{n}(\psi^*,\psi,\mu_n,\cV_n)
\end{equation*}
as in Definition \ref{defHTbackgrounddomaction},  
with the renormalized chemical potential $\mu_n\,$ and the renormalized
interaction
$\cV_n$ as above,
and with a ``perturbative correction'' $\,p_n(\psi_*,\psi)\,$
which is an analytic function of the small fields $\,\psi_*,\,\psi\,$.
To substantiate this claim, we will prove, that up to errors which can
be expected to be nonperturbatively small, 
\begin{equation*}
(\bbbs\bbbt)\Big(
\exp\Big\{- A_n\big(\psi^*,\psi, \phi_{*n},\phi_n,\,\mu_n,\cV_n \big)
+p_n(\psi^*,\psi)
\Big\} \chi_n(\psi^*,\psi) \Big)
\end{equation*}
is again of the form \eqref{eqnHTbasicformintegr}, with $n$ replaced by $n+1$.

When $n\ge 1$, application of the block spin transformation to the 
function \eqref{eqnHTbasicformintegr} leads to the integral
\refstepcounter{equation}\label{eqnHTfirstintblockspin}
\begin{equation}
\sfrac{1}{N^{(n)}_\bbbt}
\hskip-2pt  \Big[ \hskip-3pt \prod_{x\in\cX_0^{(n)}}  
\int \hskip -3pt
 \sfrac{d\psi(x)^*\wedge d\psi(x)}{2\pi i}\Big]\
e^{
-a L^{-2}\< \th_*\!-Q\psi^*,\th-Q\psi\>_{-1}
- A_n(\psi^*,\psi, \phi_{*n},\phi_n,\,\mu_n ,\cV_n)
+p_n(\psi^*,\psi)
}\chi_n
\tag{\ref{eqnHTfirstintblockspin}.a}\end{equation}
Similarly, when $n=0$, application of the block spin transformation to 
the function $e^{\cA_0}\chi_0$ leads to the integral
\begin{equation}
\sfrac{1}{N^{(0)}_\bbbt}\Big[ \hskip-3pt \prod_{x\in\cX_0}  
\int \sfrac{d\psi(x)^*\wedge d\psi(x)}{2\pi i}\Big]\
e^{
-a L^{-2}\< \th_*\!-Q\psi^*,\th-Q\psi\>_{-1}
- A_0(\psi^*,\psi,\,\mu_0,\cV_0 )
+\cR_0(\psi^*,\psi)
+\cE_0(\psi^*,\psi)
}\chi_0
\tag{\ref{eqnHTfirstintblockspin}.b}\end{equation}

We compute the dominant contributions to the integrals 
(\ref{eqnHTfirstintblockspin}.a,b) by a 
``stationary phase'' type calculation. The first step is to 
calculate the approximate critical point of the integrand.
In Proposition \ref{propHTexistencecriticalfields}, below, and 
Proposition \ref{propBGAomnibus}.a,
we prove that the critical field equations\footnote
{When, $n=0$, drop the arguments $\phi_{*n}$, $\phi_n$ from $A_n$.}
\begin{equation}\label{eqnHTcriticalfield}
\begin{split}
\nabla_{\psi_*}\big\{
   \sfrac{a}{L^2}\< \th_*\!-Q\psi_*,\th-Q\psi\>_{-1}
  +A_n(\psi_*,\psi,\phi_{*n},\phi_n,\mu_n,\cV_n)\big\}
&=0\\
\nabla_\psi\big\{
   \sfrac{a}{L^2}\< \th_*\!-Q\psi_*,\th-Q\psi\>_{-1}
  +A_n(\psi_*,\psi,\phi_{*n},\phi_n,\mu_n,\cV_n)\big\}
&=0
\end{split}
\end{equation}
have a solution 
\begin{equation*}
\psi_{*n}(\th_*,\th,\mu_n,\cV_n)\ , \qquad 
\psi_{n}(\th_*,\th,\mu_n,\cV_n) 
\end{equation*}
for $(\th_*,\th)$ in a neighbourhood of the origin in 
$\,\cH_{-1}^{(n+1)}\times \cH_{-1}^{(n+1)}\,$.

Typically, $\psi_{*n}(\th_*,\th,\mu_n,\cV_n)$ is \emph{not} 
the complex conjugate of $\psi_n(\th_*,\th,\mu_n,\cV_n)$, 
even when $\th_*=\th^*$.
Therefore we consider the integral \eqref{eqnHTfirstintblockspin}
as the integral of the holomorphic differential form
\begin{equation}\label{eqnHTfirstdiffform}
\sfrac{1}{N^{(n)}_\bbbt}
e^{
-a L^{-2}\< \th_*\!-Q\psi_*,\th-Q\psi\>_{-1}
- A_n(\psi_*,\psi, \phi_{*n},\phi_n,\,\mu_n,\cV_n )
+p_n(\psi_*,\psi)
}
 \prod_{x\in\cX_0^{(n)}}   \sfrac{d\psi_*(x)\wedge d\psi(x)}{2\pi i}
\end{equation}
over part of the real $2|\cX_0^{(n)}|$--dimensional set
$
\set{ (\psi_*, \psi)}{ \psi_*=\psi^*}
$
in the complex space $\,\cH_0^{(n)}\times\cH_0^{(n)}\,$.
The change of variables 
\begin{align}\label{eqnHTfirstchangeintvar}
\psi_*&=\psi_{*n}(\th_*,\th,\mu_n,\cV_n)+\de\psi_*\qquad
\psi=\psi_n(\th_*,\th,\mu_n,\cV_n)+\de\psi
\end{align}
maps the domain of integration to an appropriate subset, $I_n(\th_*,\th)$ of
\begin{align*}
\big\{(\de\psi_*,\de\psi) \in \cH_0^{(n)} \times \cH_0^{(n)} \,\big|\, 
  \de\psi_*=\de\psi^* +\psi_n(\th_*,\th,\mu_n,\cV_n)^*
         \!-\psi_{*n}(\th_*,\th,\mu_n,\cV_n) \big\}
\end{align*}
We write the integral \eqref{eqnHTfirstintblockspin} as 
\begin{equation*}
\int_{I_n(\th_*,\th)} \tilde\om_n
\end{equation*}
where $\,\tilde\om_n\,$ is the holomorphic differential form obtained from
\eqref{eqnHTfirstdiffform} through the substitution \eqref{eqnHTfirstchangeintvar}.
The leading part of $\,\tilde\om_n\,$ is
$
e^{-\<\de \psi_*,\, (\sfrac{a}{L^2}Q^*Q+\De^{(n)}) \de\psi\>} 
\prod\limits_{x\in\cX_0^{(n)}}   \sfrac{d\de\psi_*(x)\wedge d\de\psi(x)}{2\pi i}
$
where
\begin{equation}\label{eqnHTden}
\De^{(n)}=\left.
   \begin{cases}
      \big(\bbbone+\fQ_n Q_n D_n^{-1} Q_n^*\big)^{-1}\fQ_n & \text{if $n\ge 1$}\\
      \noalign{\vskip0.05in}
      D_0       &\text{if $n=0$}
      \end{cases}\right\}
    :\cH_0^{(n)}\rightarrow\cH_0^{(n)}
\end{equation}
See Lemma \ref{lemSTdeA} and \cite[Lemma \lemBSdeltaAalernew]{BlockSpin}. 
In \cite[Corollary \:\corPOCsquareroot]{POA} we show that the operator
$\,(\sfrac{a}{L^2}Q^*Q+\De^{(n)})\,$
is invertible. To diagonalize the quadratic form in the resulting Gaussian 
integral, let $\,D^{(n)}\,$ be an operator square root of 
\begin{equation}\label{eqnHTcn}
C^{(n)}=(\sfrac{a}{L^2}Q^*Q+\De^{(n)})^{-1} 
\end{equation}
Denote by 
$\,\om_n( \th_*,\th,\mu_n;p_n) \,$ 
the differential form (in the fields $\,\ze_*,\ze\,$ on $\,\cX_0^{(n)}\,$) 
obtained from  $\tilde \om_n$ through the second substitution
\begin{equation*}
\de\psi_*=D^{(n)*} \ze_*\qquad
\de\psi = D^{(n)} \ze 
\end{equation*}
Note, again, that $\,D^{(n)*}\,$ is the transpose of $\,D^{(n)}\,$.
As in \cite{UV}, \cite[Appendix A]{CPS} and \cite[\S2.2.1]{LH}
we construct a $\,(2|\cX_0^{(n)}|+1)$--dimensional set $\cY$ whose 
boundary consists of
\begin{itemize}[leftmargin=*, topsep=2pt, itemsep=0pt, parsep=0pt]
\item
$\
\set{(\ze_*,\ze)}{ (D^{(n)*} \ze_*, D^{(n)} \ze) \in I_n(\th_*,\th)}
\ $
\item
$\
B_n=\set{(\ze_*,\ze)}{ \ze_*=\ze^*\ ,\ |\ze(x)| \le r_n \ {\rm for\ all \ }x \in \cX_0^{(n)}}
\ $
\item
 components on which we would expect 
$\,\om_n( \th_*,\th,\mu_n;p_n) \,$ to be nonperturbatively small.
\end{itemize}
Here, $\,r_n\,$ behaves like one over a very small power of a norm of $\cV_n$ . 
See Definition \ref{defHTbasicnorm}.c. Applying Stokes Theorem to the 
holomorphic -- and hence closed -- differential form, we expect,
as  in part (c) of \cite[\S2.2.2]{LH}, that the difference between
$
\bbbt\Big(
e^{- A_n(\psi^*,\psi, \phi_{*n},\phi_n,\, \mu_n,\cV_n)
+p_n(\psi^*,\psi)}\chi_n \Big)
$
and
$
\int_{B_n} \om_n( \th_*,\th,\mu_n;p_n)
$
is nonperturbatively small. See \cite[Step 3]{PARL}.

\begin{definition}[Approximate Blockspin Transformation]\label{defHTapproximateblockspintr}
Let $\,F(\psi_*,\psi)\,$ be an analytic function of complex valued fields
$\,\psi_*,\psi\,$ on $\,\cX_0^{(n)}\,$.
The approximate blockspin transform at scale $n$ of $\,F\,$
(with respect to $q$, the constant $a>0$ and the radius $r_n$
and the chemical potential $\mu$ and quartic interaction $\cV$) is
\begin{align*}
&(\bbbt_n^{(SF)} F)(\th_*,\th;\mu,\cV) \\
&=\sfrac{1}{\tilde N^{(n)}_\bbbt} \Big[\hskip-6pt
 \prod_{x\in\cX_0^{(n)}}  
\int\limits_{|\ze(x)|\le r_n} \hskip -9pt
       \sfrac{d\ze(x)^*\wedge d\ze(x)}{2\pi i}\Big]
e^{-a L^{-2}\< \th_*\!-Q\psi_*,\th-Q\psi\>_{-1}}\,F(\psi_*,\psi)
\bigg|_{\atop{\psi_* = \psi_{*n}(\th_*,\th,\mu,\cV) + D^{(n)*}\ze^*}
             {\psi = \psi_{n}(\th_*,\th,\mu,\cV) + D^{(n)}\ze \ \ \ }
       }
\end{align*}
where
$\,\tilde N^{(n)}_\bbbt = \sfrac{1}{\det C^{(n)}} N^{(n)}_\bbbt\,$.
\end{definition}

As said above, we expect the difference between
\begin{equation*}
\bbbt\Big(
e^{- A_n(\psi_*,\psi, \phi_{*n},\phi_n,\,\mu_n,\cV_n)
+p_n(\psi_*,\psi)} \chi_n\Big)
\quad {\rm and} \quad
\bbbt_n^{(SF)}\Big(
e^{- A_n(\psi_*,\psi, \phi_{*n},\phi_n,\,\mu_n,\cV_n)
+p_n(\psi_*,\psi)}\Big)
\end{equation*}
to be nonperturbatively small.

Our main result, Theorem \ref{thmTHmaintheorem}, is a representation for
\begin{equation}\label{eqnHTwhatistoberepresented}
\Big( (\bbbs \bbbt_{n-1}^{(SF)}) \circ(\bbbs \bbbt_{n-2}^{(SF)})\circ
   \cdots \circ (\bbbs \bbbt_0^{(SF)}) \Big)
\Big(e^{\cA_0 } \Big)
\end{equation}
where the starting point 
$\,e^{\cA_0 } \,$ is the output 
\eqref{eqnHTstartingpoint} of the ultraviolet flow,
and $n\le \np$.

%%%%%%%%%%%%%%%%%%%%%%%
\subsection{The perturbative corrections}\label{sectINTpertCor}
%%%%%%%%%%%%%%%%%%%%%%%%

As said before, we shall show that, for $n\le \np$,  \eqref{eqnHTwhatistoberepresented} 
has a logarithm, whose dominant term is of the form
$\,-A_n(\psi^*,\psi, \phi_{*n},\phi_n,\,\mu_n,\cV_n)\,$
with renormalized chemical potential $\,\mu_n\,$  and a renormalized quartic interaction $\cV_n$ close to $\cV_n^{(u)}$. We will write the (perturbative) correction to it in
the form
\begin{equation*}
\cR_n\big(\phi_{*n}(\psi_*,\psi,\mu_n,\cV_n),
\phi_n(\psi_*,\psi,\mu_n,\cV_n)\big)
+\cE_n(\psi_*,\psi)
\end{equation*}
where $\,\phi_{*n},\phi_n\,$ are the background fields of Definition
\ref{defHTbackgrounddomaction}, $\,\cR_n(\phi_*,\phi)\,$ is a low degree
polynomial in fields on the fine lattice $\cX_n$, and $\cE_n(\psi_*,\psi)$
is an analytic function on a neighborhood of the origin in $\,\cH_0^{(n)}\times\cH_0^{(n)}\,$.
Some motivation for the need to distinguish between ``high degree'' and 
``low degree'' monomials and for our choice of the particular form of the
``low degree monomials'' is provided in 
\cite[Remark \remOSFdrelmotivation]{PAR2}.

\begin{remark}\label{remINTwhybackground}
We choose to express the ``low degree''
parts, $A_n$ and $\cR_n$, of the effective action as functions of the background
field $\phi_{(*)n}$, which are in turn functions of $\psi_{(*)}$, rather
than directly as functions of $\psi_{(*)}$. Here is a brief motivation.
During the course of each renormalization
group step we perform an integral over   $\psi$. To do so, we make the change
of variables, $\psi=\psi_n+\de\psi$, (see \eqref{eqnHTfirstchangeintvar} ),
where $\psi_n$ is a critical field. The leading part of the critical field is a 
linear operator, which is not particularly small, acting on the field
``$\psi$ of the next scale'' (which is a scaled version of $\th$). See
Proposition \ref{propHTexistencecriticalfields}, below. If we
simply substitute this leading part into a monomial in $\psi_{(*)}$ of
degree $p$, we again get a monomial of degree $p$, but our bound 
on the kernel of the monomial can grow because of the linear 
operator. On the other hand if we substitute the full critical field into 
a monomial in $\phi_{(*)n}(\psi_*,\psi)$, we get, by the
composition law Proposition \ref{propBGAomnibus}, followed by the appropriate
scaling, the monomial in $\phi_{(*)n+1}(\psi_*,\psi)$
with the identical kernel.

On the other hand, we choose to express the ``high degree''
part, $\cE_n$, of the effective action directly as a function of $\psi_{(*)}$.
If we were to express it, instead, through the background 
field $\phi_{(*)n}$, it would be defined on the fine lattice $\cX_n$
but would only be translation invariant with respect to the unit sublattice
$\cX_0^{(n)}$. This would complicate the process of localization
and renormalization.

\end{remark}

The functions $\,\cE_n(\psi_*,\psi)\,$ and $\,\cR_n(\phi_*,\phi)\,$  will
depend on the fields in their arguments both directly and through partial
derivatives of the fields. To make this precise, we write
\begin{align}\label{eqnTHdefexpandedstates}
\tilde \cH_j^{(n)} 
      &= \set{ \tilde\al = \big(\al,\{\al_\nu \}_{\nu=0,1,2,3}\big) }
                     { \al,\al_\nu\in \cH_j^{(n)}}
\end{align}

We shall write
\begin{equation*}
\cE_n(\psi_*,\psi)
=\tilde\cE_n\big((\psi_*, \{\partial_\nu\psi_*\}), (\psi,\{\partial_\nu\psi\})
          \big)
\end{equation*}
with an analytic function $\tilde \cE_n$  on a neighbourhood of the origin in
$\,\tilde\cH_0^{(n)}\times \tilde\cH_0^{(n)}\,$. In the next subsection we describe
how we measure the size of $\tilde \cE_n$.
Similarly we shall write
\begin{equation*}
\cR_n(\phi_*,\phi)
=\tilde\cR_n\big((\phi_*,\{\partial_\nu\phi_*\}), (\phi,\{\partial_\nu\phi\})
          \big)
\end{equation*}
with a polynomial $\,\tilde \cR_n\,$ on 
$\,\tilde\cH_n^{(0)}\times \tilde\cH_n^{(0)}\,$.
Here, for a field $\al$ on $\cX_j^{(n)}$ and $\nu=0,1,2,3\,$ we define 
the forward derivative by
\begin{equation*}
(\partial_\nu \al)(x) = 
\begin{cases}
    L^{2j} \big( \al\big(x+ \sfrac{1}{L^{2j}} e_0\big)-\al(x)\big)
         &\text{if $\nu=0$} \\
    L^j \big( \al\big(x+ \sfrac{1}{L^j} e_\nu\big)-\al(x)\big)
         &\text{if $\nu=1,2,3$}
\end{cases}
\end{equation*}
where $e_\nu$ is a unit vector in the $\nu^{\rm th}$ direction.
To make this precise we use the
%%%%%%
\begin{definition}[Monomial type]\label{defINTmonomialtype}
For a vector $\vp=(p_u,p_0,p_\sp)$ of 
nonnegative integers, a monomial of type $\vp$ in the fields 
$\tilde\al_*,\tilde\al\in\tilde\cH_j^{(n-j)}$ is a function of the form
\begin{equation*}
\int_{\cX_j^{(n-j)}} du_1\cdots du_p\ M(u_1,\cdots,u_p)
             \smprod_{\ell=1}^p\al_{\si_\ell}(u_\ell)
\end{equation*}
where each $\al_{\si_\ell}$ is one of 
$\al_*,\al, \big\{\al_{*\nu},\al_\nu\big\}_{\nu=0}^3$ but with
\begin{itemize}[leftmargin=*, topsep=2pt, itemsep=0pt, parsep=0pt]
\item 
the number of $\al_{\si_\ell}$'s that are $\al_*$ or $\al$ being $p_u$ and 
\item 
the number of $\al_{\si_\ell}$'s  that are $\al_{*0}$ or $\al_0$ being $p_0$ and
\item
the number of $\al_{\si_\ell}$'s that are $\al_{*\nu}$ or $\al_\nu$ for 
some $1\le\nu\le 3$ being $p_\sp$.
\end{itemize}
In the monomial above there are $p_u$ undifferentiated fields, 
$p_0$ fields corresponding to time derivatives and 
$p_\sp$  fields corresponding to space derivatives. The subscript 
$u$ stands for ``undifferentiated'' and the subscript $\sp$ stands for 
``spatial''. 

\noindent
A polynomial of type $\vp$ is a sum of monomials of type $\vp$.
\end{definition}

\noindent
$\,\tilde\cR_n\,$ will be a sum, over $\vp \in \fD$, 
of polynomials, $\tilde\cR_n^{(\vp)}$,  of type $\vp$ in the fields 
$\tilde\phi_*,\tilde\phi\in\tilde\cH_n^{(0)}$ where
\begin{equation}\label{eqnINTfDdef}
\fD=\big\{(1,1,0)\,,\,
          (0,1,1)\,,\,
          (0,0,2)\,,\,
          (6,0,0)\big\}
\end{equation}
The motivation for this choice of $\fD$ is provided in 
\cite[Remark \remOSFdrelmotivation]{PAR2}.
 
In the next subsection we  describe
how we measure the size of the kernels in this representation.

%%%%%%%%%%%%%%%%%%%%%%%
\subsection{ Norms for measuring the size of the perturbative corrections}\label{sectINTnorms}
%%%%%%%%%%%%%%%%%%%%%%%%

Let $\cX$ be any lattice that is equipped with a metric $d$ and a 
``cell volume'' $\vol$. As an example, the lattice $\cX_j^{(n-j)}$ has $\vol=\sfrac{1}{L^{5j}}$.
The following Definition describes how we measure the size of the kernels 
$\,R_n^{(\vp)}\,$ as above.

\begin{definition}\label{defHTkernelnorm}
Let $\,f(u_1,\cdots,u_r)\,$ be a function on $\,\cX^r\,$. For a mass 
$\,\fm\ge 0\,$ we set
\begin{equation*}
\|f\|_\fm
=\max_{i=1\cdots,r} \max_{u_{i}}
 \int du_1 \cdots du_{i-1}\,du_{i+1} \cdots du_r\ |f(u_1,\cdots,u_r)|\, e^{\fm\tau(u_1,\cdots,u_r)}
\end{equation*}
where \emph{the tree length} $\,\tau(u_1,\cdots,u_r)\,$ is the minimal length 
of a tree in $\,\cX\,$ that has  $\,u_1,\cdots,u_r\,$ among its vertices, and
$\int du\ g(u)=\vol\sum_{u\in\cX} g(u)$.
\end{definition}

As in \cite{GK2}, our perturbative corrections are analytic functions of the
fields.
The following definitions describe how we measure the size of complex valued
analytic functions of fields, like  $\cE_n(\psi_*,\psi)$ and
$\tilde\cE_n\big((\psi_*, \{\psi_{*\nu}\}), (\psi,\{\psi_\nu\})
          \big)
$. The norms in the following definition are special cases of the norms
in \cite[Definition 2.6]{CPC}.

\begin{definition}\label{defHTabstractnorm}
\ 
\begin{enumerate}[label=(\alph*), leftmargin=*]
\item 
For a field $\,\al\,$ on $\,\cX\,$ and 
$\,\vec x = (x_1,\cdots,x_r) \in \cX^r\,$ we set
$\ \al(\vec x) = \smprod_{i=1}^r \al(x_i)\,$.

\item 
A power series  $\,\cF\,$  in the fields $\,\al_1,\cdots,\al_s$, 
on $\cX\,$ has a unique expansion
\begin{equation*}
\cF(\al_1,\cdots,\al_s)
      = \sum_{r_1,\cdots,r_s\ge 0}\vol^{r_1+\cdots+r_s}\hskip-5pt  
         \sum_{\atop{\vec x_i\in\cX^{r_i}}{ 1\le i\le s}}
            f_{r_1,\cdots,r_s}\big(\vec x_1,\cdots,\vec x_s\big)\,
            \smprod_{i=1}^s \al_i(\vec x_i)
\end{equation*}
where the coefficients $f\big(\vec x_1,\cdots,\vec x_s\big)$ are invariant
under permutations of the components of each vector $\vec x_i$.
\item 
For each choice of ``weights'' $\ka_1,\cdots,\ka_s>0$,
for the fields $\al_1,\cdots,\al_s$, we define \emph{the norm of $\cF$ with
mass $\fm$ and weights $\ka_1,\cdots,\ka_s>0$} to be
\begin{equation*}
 \sum_{r_1,\cdots,r_s\ge 0}
  \big\| f_{r_1,\cdots,r_s}\big(\vec x_1,\cdots,\vec x_s\big)\big\|_\fm\ \
  \smprod_{i=1}^s \ka_i^{r_i}
\end{equation*}
\end{enumerate}
\end{definition}

Similarly,  Definition \ref{defDEFkrnel} describes how we measure the size of analytic
maps like the background field map $(\psi_*,\psi)\mapsto
\phi_n(\psi_*,\psi,\mu_n,\cV_n)$. 

%%%%%%%%%%%%%%%%%%%%%%%
\subsection{The Starting Point Setup}\label{sectINTstartPoint}
%%%%%%%%%%%%%%%%%%%%%%%%

We shall state our results in terms of an abstraction of the
output of the temporal ultraviolet limit outlined following
\eqref{eqnHTstartingpoint} and \eqref{eqnHTaIn}. We assume that we are given
a mass $m>0$, positive odd integers $L_\tp$ and $L_\sp$, 
a small real number $\eps>0$, and
\begin{itemize}[leftmargin=*, topsep=2pt, itemsep=0pt, parsep=0pt]
\item
 a kinetic energy operator
\begin{equation*}
\bh_0 = \nabla^*\bH\nabla
\end{equation*}
where $\bH$ is a real, translation invariant, reflection invariant, 
strictly positive definite operator on the space, $L^2\big({(\bbbz^3)}^*)$, 
of functions  on the set, ${(\bbbz^3)}^*$, of nearest neighbor bonds of 
the lattice $\bbbz^3$. The operator 
$\nabla:L^2\big(\bbbz^3\big)\rightarrow L^2\big({(\bbbz^3)}^*\big)$ is given by
\begin{equation*}
\big(\nabla f\big)\big(\<\bx,\by\>\big) = f(\by)-f(\bx)
\end{equation*}
We assume that the kernel of $\bH$ is exponentially decaying with 
$\|\bh_0\|_{2m}$ finite.
\item
a kernel $\bV_0(x_1,x_2,x_3,x_4)$
on $\big((\bbbz/L_\tp\bbbz) \times \bbbz^3 \big)^4$ 
that is invariant under $x_1\leftrightarrow x_3$ and under 
$x_2\leftrightarrow x_4$ and under the symmetry group $\fS$
of Definition \ref{defSYfullSymmetry}. We assume that its average 
\begin{equation*}
\rv_0=\sum_{x_2,x_3,x_4\in\bbbz^3} \bV_0(0,x_2,x_3,x_4)>0
\qquad\text{and}\qquad
\fv_0= 2\|\bV_0\|_{2m}
\end{equation*}
are sufficiently small.
\item
a real chemical potential $\mu_0$ obeying
\begin{equation*}
 \mu_* +  \fv_0^{\sfrac{4}{3} -16\eps}  \le \mu_0  
                       \le    \fv_0^{\sfrac{8}{9} +\eps}
\end{equation*}
where\footnote{We show, in Lemma \ref{lemSZexplicitmustar},
that, to leading order, $\mu_*$ is $\th$ times \eqref{eqnHTdefcriticalmu}.}
\begin{equation}\label{eqnINTmustardef}
\mu_*=2\int_{{((\bbbz/L_\tp\bbbz) \times \bbbz^3)}^3} dx_1 \cdots dx_3\  
              \bV_0(0,x_1,x_2,x_3)\ \bD_0^{-1}(x_3,0)
\end{equation} 
with $\bD_0=\bbbone - e^{-\bh_0} -e^{-\bh_0} \partial_0\,$. 
\end{itemize}
The periodized versions, on the lattice
$
\cX_0 = \big( \bbbz \times \bbbz^3 \big) / 
                          \big(L_\tp\bbbz \times L_\sp\bbbz^3 \big)
$,
of $\bh_0$, $\bV_0$, $\bD_0$ are denoted $\oh_0$, $V_0$ and $D_0$,
respectively. We also assume that we are given
\begin{itemize}[leftmargin=*, topsep=2pt, itemsep=0pt, parsep=0pt]
\item an $\fS$ invariant polynomial 
$\,\tilde\cR_0(\tilde\psi_*,\tilde\psi)=\sum_{\vp\in\fD}
               \tilde\cR_0^{(\vp)}(\tilde\psi_*,\tilde\psi)\,$
on  $\,\tilde\cH_0\times \tilde\cH_0\,$.  
Each $\tilde\cR_0^{(\vp)}$ is a polynomial of type $\vp$ with a
real valued kernel\footnote{That is, 
$\overline{\tilde\cR_0^{(\vp)}((\psi_*, \{\psi_{*\nu}\}), (\psi,\{\psi_\nu\}) )}
= \tilde\cR_0^{(\vp)}((\psi_*^*, \{\psi_{*\nu}^*\}), (\psi^*,\{\psi_\nu\}^*) )$. }
that obeys the bound
\begin{align*}
\big\|\tilde\cR_0^{(\vp)}\big\|_m
 &\le \fr_{\vp}(0)
 = \begin{cases}\fv_0^{2-\eps} &\text{if  $\vp =(6,0,0)$} \\
                         \noalign{\vskip0.05in}
                                     \fv_0^{1-4\eps} &\text{otherwise}
    \end{cases}
\end{align*}
%%%%%%%%%%%%%%%%%%%%%%%%%%%%%%%
\item
$\cE_0(\psi_*,\psi)$ 
is an $\fS$ invariant, particle--number preserving function with real valued 
kernels and with $\cE_0(0,0)=0$, that
\begin{itemize}[leftmargin=*, topsep=2pt, itemsep=0pt, parsep=0pt]
\item
 is of degree at least four both in $\psi_*$ and in $\psi$, and 
\item 
whose norm 
\begin{equation*}
\|\cE_0\|^{(0)} \le \fv_0^\eps 
\end{equation*}
where $\|\ \cdot\ \|^{(0)}$ is the norm with mass $m$ which associates
the weight $\,\ka(0)=\sfrac{1}{\fv_0^{1/3-\eps}}\,$ to the fields 
$\,\psi_*, \,\psi\,$.
\end{itemize}
\end{itemize}
We set
\begin{equation}\label{eqnHTaZero}
\cA_0(\psi_*,\psi) = -\<\psi_*,\,D_0\psi\>_0
       -\cV_0(\psi_*,\psi)
       +\mu_0 \<\psi_*,\,\psi\>_0
       +\cR_0(\psi_*,\psi) 
       +\cE_0(\psi_*,\psi)
\end{equation}
where
\begin{align*}
\cV_0(\psi_*,\psi) & = \half\int_{\cX_0^4} dx_1 \cdots dx_4\  
                  V_0(x_1,x_2,x_3,x_4)\,  
                  \psi_*(x_1) \psi(x_2)\psi_*(x_3) \psi(x_4)
\\
\cR_0(\psi_*,\psi) 
&= \tilde\cR_0\big((\psi_*,\{\partial_\nu\psi_*\}),(\psi,\{\partial_\nu\psi\})
          \big)
\end{align*}

Under reasonable conditions on the various parameters,
the small field part of the output of the temporal ultraviolet limit,
described following \eqref{eqnTHsmallfieldoutput}, satisfies these conditions.
This is stated in more detail and proven in 
Corollary \ref{corSZstepzeroConditions}.

%%%%%%%%%%%%%%%%%%%%%%
\subsection{The main results}\label{sectINTmainResults}
%%%%%%%%%%%%%%%%%%%%%%%

We start by defining a number of parameters that will be used in the
statement of the main results.

\begin{definition}\label{defHTbasicnorm}
\ 
\begin{enumerate}[label=(\alph*), leftmargin=*]
\item
Set
\begin{align*}
\ka(n)&=\sfrac{L^{\eta n}}{\fv_0^{1/3-\eps}} &
\quad &{\rm with} \quad & \eta&=\sfrac{1}{2}+\sfrac{1}{3}\sfrac{\log\fv_0}{\log(\mu_0-\mu_*)}
\\
\ka'(n)&=\sfrac{L^{\eta' n}}{\fv_0^{1/3-\eps}} &
\quad &{\rm with} \quad & \eta'&=\sfrac{3}{2}-\sfrac{\log\fv_0}{\log(\mu_0-\mu_*)}-\eps
\\
\fe_\fl(n)&=L^{\eta_\fl n}\,\fv_0^{\sfrac{1}{3}-2\eps } &
\quad &{\rm with} \quad &
\eta_\fl&=\big(\frac{2}{3}-4\eps\big)\,\sfrac{\log\fv_0}{\log(\mu_0-\mu_*)}
\end{align*}
With the notation of Definition \ref{defHTabstractnorm}, we define the 
norm $\,\|\tilde \cE\|^{(n)},$ of an analytic function
$\ 
\tilde \cE(\tilde\psi_*,\tilde\psi)
\ $,
as the norm with mass $m$ which associates
the weight $\,\ka(n)\,$ 
to the fields $\,\psi_*, \,\psi\,$, and the weight
$\,\ka'(n)\,$ to the fields $\,\psi_{\nu *},\psi_\nu\,$, $\,\nu=0,\cdots,3\,$.
Similarly, we define $\,\|\tilde \cE\|_m$ as the norm with mass $m$ 
which associates the weight 1 to all fields. The parameter $\fe_\fl(n)$
will be used as an upper bound on the size of the output of the 
fluctuation integral in the $n^{\rm th}$ step.

\item
The number of steps we perform in this paper, using the ``parabolic flow'', 
is the largest integer $\np$ such that 
\begin{equation*}
L^{2\np}(\mu_0-\mu_*)\le \big(\sfrac{ \fv_0}{L^{\np}} \big)^{5\eps}
\iff L^{(2+5\eps)\np}(\mu_0-\mu_*) \le \fv_0^{5\eps}
\end{equation*}

\item
For the radius of integration of the fluctuation variables in Definition
\ref{defHTapproximateblockspintr}, we choose
\begin{equation*}
r_n = \sfrac{1}{4} \ka_\fl(n+1) \qquad {\rm where}\ 
\ka_\fl(n)= \big( \sfrac{L^n}{\fv_0} \big)^{\eps/2}
\end{equation*}

\item 
To specify the averaging profile $q$ of Definition 
\ref{defHTblockspintr}.a we fix\footnote{The reasons for requiring
that $\fq\ge 4$ are discussed in \cite[Remark \remPBOqfour]{POA}.}
 a (small) even natural number $\fq\ge 4$ and
denote by $1_{\sq}(x)$ the characteristic
function of the rectangle $\big[-\sfrac{L^2-1}{2}, \sfrac{L^2-1}{2} \big]
\times \big[-\sfrac{L-1}{2}, \sfrac{L-1}{2} \big]^3$ in 
$\bbbz\times\bbbz^3$.  Set
\begin{equation*}
q=\sfrac{1}{L^{5\fq}} 
\overbrace{ 1_{\sq}*1_{\sq}*\cdots*1_{\sq} }^{\fq\ {\rm times}}
\end{equation*}
to be the convolution of $1_{\sq}$ with itself $\fq$ times, normalized
to have integral one. Properties of $q$ are discussed in
\cite[\S\secPOblockSpin]{POA}.
\end{enumerate}
\end{definition}

\begin{remark}\label{remHTbasicnorm}
By construction 
\begin{align*}
\sfrac{3}{4}+8\eps < \sfrac{\log\fv_0}{\log(\mu_0-\mu_*)}< \sfrac{9}{8}-\eps
 \qquad\qquad
 \sfrac{3}{4}+2\eps  < \eta < \sfrac{7}{8} -\sfrac{\eps}{3}
 \qquad\qquad
 \sfrac{3}{8}  < \eta' < \sfrac{3}{4} -8\eps
\end{align*}
Also, by Definition \ref{defHTbasicnorm}.b and the condition 
$\mu_0-\mu_*\ge \fv_0^{\frac{4}{3}-16\eps}$ of \S \ref{sectINTstartPoint},
$$
L^{\np}\le \sfrac{1}{\fv_0^{\frac{2}{3}-8\eps}}
$$
\end{remark}

For the approximate block spin transformations to
be well defined, we need to make sure that the background fields of Definition
\ref{defHTbackgrounddomaction}.c and the critical fields of 
\eqref{eqnHTcriticalfield} are defined for small fields. The main 
technical work here is to show that
the linearized  equations for the background field 
(see \cite[Definition \defBSbackfld.a and (\eqnBSbckgndequ')]{BlockSpin})
are solvable. This is guaranteed by

\begin{theorem}[Green's Functions]\label{HTthminvertibleoperators}
There are constants  $\mu_{\rm up}, m_0>0$ and $\Gam_\op$, that are
independent of $n$ and $L$, such that the following hold.
Let $\,0\le n\le \np$, $0\le m\le\sfrac{m_0}{2}$ and 
$\,|\mu|\le\mu_{\rm up}\,$. The operators
$\,
D_n+Q_n^*\fQ_n Q_n
\,$
and
$\,
D_n+Q_n^*\fQ_n Q_n-\mu
\,$
on $\,\cH_0^{(n)}\,$
are invertible. We set
\begin{equation*}
S_n=(D_n+Q_n^*\fQ_n Q_n)^{-1} \qquad
S_n(\mu)=(D_n+Q_n^*\fQ_n Q_n-\mu)^{-1} 
\end{equation*}
Then
\begin{align*}
\|S_n\|_{2m}\,,\, \|S_n(\mu) \|_{2m}& \le \Gam_\op
\end{align*}
\end{theorem}
\noindent This theorem is proven in \cite[Proposition \POGmainpos]{POA}.

\begin{proposition}[Background fields]\label{propHTexistencebackgroundfields}
Let $\,1\le n\le \np\,$. Let $\mu$ be a complex number and $\cV(\phi_*,\phi)$
be a quartic monomial with 
$
\|\cV\|_{2m}\ka(n)^2+|\mu|
$
sufficiently small. Then there are analytic maps 
\begin{equation*}
(\psi_*,\psi)\mapsto 
\phi_{*n}(\psi_*,\psi,\mu,\cV),\ \phi_{n}(\psi_*,\psi,\mu,\cV)
\end{equation*}
to $\cH_n$, that are defined for all 
   $(\psi_*,\psi)\in\cH_0^{(n)}\times \cH_0^{(n)}\,$
obeying $\big|\psi_*(x)\big|\,,\,\big|\psi(x)\big| \le\ka(n)$
for all $x\in \cX_0^{(n)}$, and that fulfill the background field equations 
\begin{equation*}
\sfrac{\partial\hfill}{\partial\phi_*}A_n(\psi_*,\psi,\phi_*,\phi,\mu,\cV)=
\sfrac{\partial\hfill}{\partial\phi}A_n(\psi_*,\psi,\phi_*,\phi,\mu,\cV)=0
\end{equation*}
Furthermore 
\begin{align*}
\phi_{*n}(\psi_*,\psi,\mu,\cV)
     &=S_n(\mu)^* Q_n^* \fQ_n\,\psi_*
      +\phi_{*n}^{(\ge 3)}(\psi_*,\psi,\mu,\cV)\\
\phi_n(\psi_*,\psi,\mu,\cV)
     &=S_n(\mu)Q_n^* \fQ_n\,\psi
      +\phi_n^{(\ge 3)}(\psi_*,\psi,\mu,\cV)
\end{align*} 
with analytic maps $\phi_{*n}^{(\ge 3)}$ and $\phi_n^{(\ge 3)}$ that 
are of degree at least three in $(\psi_*,\psi)$. 
\end{proposition}
\noindent
This Proposition, with more details and bounds, is proven in 
\cite[Proposition \propBGEphivepssoln]{BGE}, using a contraction mapping argument.

\begin{proposition}[Critical fields]\label{propHTexistencecriticalfields}
Let $\,0\le n< \np\,$. Let $\mu$ be a complex number and $\cV(\phi_*,\phi)$
be a quartic monomial with 
$
\|\cV\|_{2m}\ka(n+1)^2+|\mu|
$
sufficiently small. Then there are analytic maps 
\begin{equation*}
(\th_*,\th)\mapsto 
\psi_{*n}(\th_*,\th,\mu,\cV),\ \psi_{n}(\th_*,\th,\mu,\cV)
\end{equation*}
to $\cH_0^{(n)}$, that are defined for all 
   $(\th_*,\th)\in\cH_{-1}^{(n+1)}\times \cH_{-1}^{(n+1)}\,$
obeying $\big|\th_*(y)\big|\,,\,\big|\th(y)\big| <\sfrac{\ka(n+1)}{L^{3/2}}$
for all $y\in \cX_{-1}^{(n+1)}$, and that fulfill the the critical 
field equations
\begin{align*}
   \left. \atop{
   \nabla_{\psi_*}\big\{
   \sfrac{a}{L^2}\< \th_*\!-\!Q\psi_*,\th\!-\!Q\psi\>_{-1}
  \!+\!A_n\big(\psi_*,\psi,\phi_*,\phi,\mu,\cV\big)
      \big|_{\phi_{(*)}=\phi_{(*)n}(\psi_*,\psi,\mu,\cV)}\big\}\!=0} 
   {\nabla_\psi\big\{
              \sfrac{a}{L^2}\< \th_*\!-\!Q\psi_*,\th\!-\!Q\psi\>_{-1}
       \!+\!A_n\big(\psi_*,\psi,\phi_*,\phi,\mu,\cV\big)
          \big|_{\phi_{(*)}=\phi_{(*)n}(\psi_*,\psi,\mu,\cV)}\big\}
    \!=0}\!
\right\}&\,\text{if $n\ge 1$}\\
\noalign{\vskip0.05in}
\left.\atop{
    \nabla_{\psi_*}\big\{
   \sfrac{a}{L^2}\< \th_*\!-Q\psi_*,\th-Q\psi\>_{-1}
     +A_0\big(\psi_*,\psi,\mu,\cV\big)\big\} =0}
        {\nabla_\psi\big\{
   \sfrac{a}{L^2}\< \th_*\!-Q\psi_*,\th-Q\psi\>_{-1}
  +A_0\big(\psi_*,\psi,\mu,\cV\big)\big\}=0}\!
   \right\}&\,\text{if $n=0$}
\end{align*}
Furthermore 
\begin{align*}
\psi_{*n}(\th_*,\th,\mu,\cV)
&= \sfrac{a}{L^2} C^{(n)}(\mu)^*Q^*\,\th_*
    + \psi_{*n}^{(\ge 3)}(\th_*,\th,\mu,\cV)
\\
\psi_n(\th_*,\th,\mu,\cV)
&= \sfrac{a}{L^2} C^{(n)}(\mu)Q^*\,\th
   + \psi_n^{(\ge 3)}(\th_*,\th,\mu,\cV)
\end{align*}
where\footnote{By \cite[Remark \remBSedA.a]{BlockSpin} 
$\De^{(n)}(0)$ is the $\De^{(n)}$ of \eqref{eqnHTden} and consequently
$C^{(n)}(0)$ is the $C^{(n)}$ of \eqref{eqnHTcn}.}
\begin{align*}
C^{(n)}(\mu)&=\big(\sfrac{a}{L^2}Q^*Q+\De^{(n)}(\mu)\big)^{-1} \\
\De^{(n)}(\mu) &= \begin{cases}
          \fQ_n -\fQ_n\,Q_n S_n(\mu) Q_n^* \fQ_n &\text{if $n\ge 1$}\\
          D_0-\mu & \text{if $n=0$}
          \end{cases}
\end{align*}
are well defined operators.  Both $\psi_{*n}^{(\ge 3)}$  and 
$\psi_n^{(\ge 3)}$ are of degree at least three in $(\th_*,\th)$.
\end{proposition}
\noindent
This Proposition, with more details and bounds,
is proven in Proposition \ref{propBGAomnibus} and
\cite[Proposition \propCFpsisoln and Remark \remCFhatpsitopsi]{BGE}.

\begin{definition}\label{defINTrelmonomial} 
The ``scaling/weight relevant'' monomials are those of type 
\begin{align*}
&\vp\in\fD_\rel 
   = \fD\cup \Big\{(2,0,0)\ ,\  (1,0,1)\ ,\  (3,0,1),\  (4,0,0)\Big\} \\
&\hskip0.4in = \Big\{(2,0,0)\,,\,(1,1,0)\,,\,(1,0,1)\,,\,
          (0,1,1)\,,\,(0,0,2)\,,\,(4,0,0)\,,\,
          (3,0,1)\,,\,(6,0,0)\Big\}
\end{align*}
The motivation for this choice of $\fD_\rel$ is provided in 
\cite[Remark \remOSFdrelmotivation]{PAR2}.
\end{definition}

\begin{theorem}\label{thmTHmaintheorem}
Assume that the parameter $\eps$ of \S\ref{sectINTstartPoint} is sufficiently
small, that the parameter $L$ of Definition  \ref{defHTblockspintr} is 
sufficiently large, and that $\fv_0$ is sufficiently small, 
depending on $\eps$ and $L$.  
There exists, for each  $1\le n\le \np$,
\begin{itemize}[leftmargin=*, topsep=2pt, itemsep=2pt, parsep=0pt]
\item 
a number $\,\mu_n\,$ with 
          $\,\big|\mu_n-L^{2n}(\mu_0-\mu_*)\big|\le \fv_0^{1-\eps}
                   +L^{2n}\fv_0^{\sfrac{4}{3}-15\eps}\,$

\item a quartic interaction $\,\cV_n(\phi_*,\phi)\,$ 
with $\,\|\cV_n-\cV_n^{(u)}\|_{2m}\le \sfrac{1}{L^n}\fv_0^{\sfrac{5}{3}-7\eps}$

\item a $\fS$ invariant polynomial 
$\,\tilde\cR_n(\tilde\phi_*,\tilde\phi)=\sum_{\vp\in\fD}
               \tilde\cR_n^{(\vp)}(\tilde\phi_*,\tilde\phi)\,$
on  $\,\tilde\cH_n\times \tilde\cH_n\,$.
Each $\tilde\cR_n^{(\vp)}$ is a polynomial of type $\vp$
that obeys the bound
\begin{align*}
\big\|\tilde\cR_n^{(\vp)}\big\|_m
 &\le  \begin{cases}
   \fv_0^{2-6\eps} \,L^{-4n} 
    &  \text{if $\vp = (6,0,0)$}
\\ \noalign{\vskip0.1in}
    \fv_0^{1-6\eps}
    &  \text{if $\vp = (1,1,0),\ (0,1,1),\ (0,0,2)$}
\end{cases}
\end{align*}

\item
an $\fS$ invariant analytic function  
$\,
\tilde\cE_n(\tilde\psi_*,\tilde\psi)
\,$
with $\tilde\cE_n(0,0)=0$, whose power series expansion does not
contain scaling/weight relevant monomials 
and which has norm
\begin{equation*}
\|\tilde\cE_n\|^{(n)} \le \fv_0^\eps
\end{equation*}

\item 
and a normalization constant $\cZ_n$
\end{itemize}
such that
\begin{align*}
&
\Big( (\bbbs \bbbt_{n-1}^{(SF)}) \circ(\bbbs \bbbt_{n-2}^{(SF)})\circ
   \cdots \circ (\bbbs \bbbt_0^{(SF)}) \Big)
\Big(e^{\cA_0} \Big)
\\
&\hskip 2cm
= \sfrac{1}{\cZ_n}\exp\Big\{- A_n(\psi_*,\psi, \phi_{*n}, \phi_n,\,\mu_n,\cV_n)
+\cR_n(\phi_{*n},\phi_n)+\cE_n(\psi_*,\psi)
\Big\}
\end{align*}
with the $\,A_n\,$ and the background fields 
$$
\phi_{*n} = \phi_{*n}(\psi_*,\psi,\mu_n,\cV_n)
\qquad
\phi_{n} = \phi_{n}(\psi_*,\psi,\mu_n,\cV_n)
$$
as in Definition \ref{defHTbackgrounddomaction}, and
\begin{align*}
\cR_n(\phi_*,\phi )
&= \tilde\cR_n\big((\phi_*,\{\partial_\nu\phi_*\}),(\phi,\{\partial_\nu\phi\})
          \big)
\\
\cE_n(\psi_*,\psi )
&= \tilde\cE_n\big((\psi_*,\{\partial_\nu\psi_*\}),(\psi,\{\partial_\nu\psi\})
          \big)
\end{align*}
Here, for each $0\le j\le n-1$, $\bbbt_j^{(SF)}$ is the approximate
block spin transformation of Definition \ref{defHTapproximateblockspintr} 
with chemical potential $\mu=\mu_j$ and quartic interaction $\cV_j$.
\end{theorem}

\begin{remark}\label{remHTpreciseinduction}
Theorem \ref{thmTHmaintheorem} will be proven by induction on $n$. 
The proof runs over \S\ref{chapSTstrategy}--\cite[\S\chapRen]{PAR2}, 
and is completed at the end of \cite[\S\chapRen]{PAR2}.
In the inductive proof for the estimates on $\mu_n$, $\cV_n$ and $\tilde\cR_n$, 
we will prove  slightly stronger estimates that are more suited for the 
induction.
Set $\mu^*_0=\mu_0$ and, for $n\ge 1$,
\begin{equation*}
\mu_n^* = L^{2n}\mu_0
 -\sfrac{2}{|\cX_0^{(n)}|}
\int_{\cX_n^4} du_1 \cdots du_4\  
    V_n^{(u)}(u_1,u_2,u_3,u_4)\ S_n(u_4,u_1)
\end{equation*}
Then we will show that there are constants $\CC_{\de\cV}$ and $\CC_\cR$
such that, for all $0\le n\le \np$,
\begin{equation}\label{eqnHTinductiveVmuestimates}
\begin{split}
|\mu_n-\mu_n^*| &\le L^{2n}\,\fv_0^{1-8\eps}
     \smsum_{\ell=1}^n \sfrac{1}{L^{(2-3\eps)\ell}}
        \big[\fv_0^{\sfrac{1}{3}-5\eps} +L^{2\ell}(\mu_0-\mu_*)\big]
\\
{\|\cV_n-\cV_n^{(u)}\|}_{2m}
&\le \sfrac{\CC_{\de\cV}}{L^n} 
    \smsum_{\ell=1}^n \sfrac{L^\ell }{\ka(\ell)^4}\fe_\fl(\ell-1)
\end{split}
\end{equation}
and 
\begin{equation}\label{eqnHTinductiveRestimates} 
\big\|\tilde\cR_n^{(\vp)}\big\|_m \le \fr_{\vp}(n,\CC_\cR)
\end{equation}
where
\begin{align*}
\fr_{\vp}(n,C)
=\begin{cases}
 \Pi^n_0(C)\ \fr_{\vp}(0)
 +C \sum\limits_{\ell=1}^{n} 
   \sfrac{ \fe_\fl(\ell-1) }{\ka(\ell) \ka'(\ell)}
    \Pi^n_\ell(C) & \text{if $\vp =(1,1,0)$} \\
%%%%%%%%%%%
 \sfrac{1}{L^n}\Pi^n_0(C)\ \fr_{\vp}(0)
 +\sfrac{C}{L^n} \sum\limits_{\ell=1}^{n} 
   L^{\ell}\sfrac{\fe_\fl(\ell-1) }{\ka'(\ell)^2}
    \Pi^n_\ell(C) &\text{if $\vp =(0,1,1)$} \\
%%%%%%%%%%
 \Pi^n_0(C)\ \fr_{\vp}(0)
 +C \sum\limits_{\ell=1}^{n} 
   \sfrac{ \fe_\fl(\ell-1) }{\ka'(\ell)^2}
    \Pi^n_\ell(C) &\text{if $\vp =(0,0,2)$}\\
%%%%%%%%%%%
 \sfrac{1}{L^{4n}}\Pi^n_0(C)\ \fr_{\vp}(0)
 +\sfrac{C}{L^{4n}} \sum\limits_{\ell=1}^{n} 
   L^{4\ell}\sfrac{\fe_\fl(\ell-1) }{\ka(\ell)^6}
    \Pi^n_\ell(C) &\text{if $\vp =(6,0,0)$}
\end{cases} 
\end{align*}
with
\begin{align*}
\Pi^n_\ell(C)
=\prod_{j=\ell+1}^n\Big(1+C\sfrac{\fe_\fl(j-1)}{\ka(j)^2}\Big)
\end{align*}
More precisely, we shall show that there are constants $\CC_{\de\cV}$ 
and $\CC_\cR$ such that if \eqref{eqnHTinductiveVmuestimates} and 
\eqref{eqnHTinductiveRestimates} are valid for some $n\ge 0$, they are also
valid for $n+1$. Observe that \eqref{eqnHTinductiveVmuestimates} and 
\eqref{eqnHTinductiveRestimates} are trivially satisfied when $n=0$.
See 
\cite[Lemma \lemRENmunppties]{PAR2}, for $|\mu_n-\mu^*_n|$,
\cite[Lemma \lemRENreninteraction]{PAR2}, for the construction of 
           $\CC_{\de\cV}$,
and \cite[Lemma \lemRENcRcE]{PAR2}, for the construction of $\CC_\cR$.

\noindent
That \eqref{eqnHTinductiveVmuestimates} and \eqref{eqnHTinductiveRestimates}
imply the bounds on $\mu_n$, $\cV_n$ and $\tilde\cR_n$ of Theorem
\ref{thmTHmaintheorem} is proven in Corollary \ref{corPARmunvn}
and Lemma \ref{lemPARcompradan}.d.

\end{remark}

We are grateful for some very useful comversations with 
Martin Lohmann and Serena Cenatiempo.

%%%%%%%%%%%%%%%%%%%%%%%
\subsection{Outline}\label{sectINToutline}
%%%%%%%%%%%%%%%%%%%%%%%%

\noindent
The rest of this paper contains the more algebraic parts of the proof of
Theorem \ref{thmTHmaintheorem}. The more analytical part of the proof is
given in \cite{PAR2}. 
Here is a more detailed outline of these two papers, 
as well as an indication of their connections to \cite{PARL,BGE,POA}.
\begin{itemize}[leftmargin=*, topsep=2pt, itemsep=2pt, parsep=0pt]
\item 
In Appendix \ref{appSZrewrite}, we review the results of 
\cite{UV} and rewrite the main output of that paper in the form
of  \S\ref{sectINTstartPoint}.

\item
\S\ref{chapSCscaling} provides a number of simple preparatory
results regarding the interaction of the scaling operation
of Definition \ref{defHTscaling}  with objects that
will be encountered during the course of the construction. 

\item
Various algebraic properties of the background and critical
fields, such as the ``composition rule'', are reviewed in 
\S \ref{chapBGAbackgroundAlg}.

\item
The algebraic steps of the application of one
block spin transformation, $\bbbt_n^{(SF)}$, and subsequent 
scaling, $\bbbs$, leading up to the formulation of the 
``fluctuation integral'', are performed in \S\ref{chapSTstrategy}.

\item
The fluctuation integral is evaluated and bounded in 
\cite[\S\chapOSFfluct]{PAR2}.

\item
In \cite[\S\chapRen]{PAR2}, the output of the fluctuation integral 
is reorganized to complete the inductive proof of 
Theorem \ref{thmTHmaintheorem}. Part of the reorganization is the 
renormalization of the chemical potential and the interaction.
\medskip

\item
The translation and reflection symmetries we use
are discussed in Appendix \ref{appSYsymmetry}.

%\item
%In Appendix \ref{appBSalg} we discuss some purely algebraic 
%aspects of the block spin renormalization group in an abstract setting.

\item
A large number of inequalities relating our 
weight factors and various other parameters are proven in Appendix
\ref{appIneq}. In particular we prove that the detailed inequalities
of Remark \ref{remHTpreciseinduction} imply the simple bounds of 
Theorem \ref{thmTHmaintheorem}.
 
\item
Localization operations and decompositions, that are used
in the renormalization of the chemical potential, are discussed
in \cite[Appendix \appLocal]{PAR2}.

\item
In \cite[Appendix \appMustar]{PAR2} we identify the $\mu_*$
of \eqref{eqnHTdefcriticalmu} as the limit of the $\mu_n^*$'s 
of Remark \ref{remHTpreciseinduction}.

\item
The effect of scaling on the norms used in this paper is
discussed in \cite[Appendix \appSCscaling]{PAR2}. This is used to identify 
``scaling relevant and irrelevant'' monomials. See 
Definition \ref{defINTrelmonomial} and 
\cite[Remark \remOSFdrelmotivation]{PAR2}.

\medskip
\item
In \cite{PARL}, we give reasons, on a handwaving
level, why we expect that the errors introduced by approximating 
the blockspin transformation $\bbbt$ by the small field blockspin
transformation $\bbbt^{(SF)}$ are nonperturbatively small.

\medskip
\item
Estimates on the background and critical fields are
crucial for our construction. They are proven in \cite{BGE}. 
The upper bounds of that paper involve a number of
constants $\GGa_1, \GGa_2, \cdots$ that are all independent of $L$ 
and the scale index $n$. In \cite[Convention \convBGEconstants]{BGE}, we define $\GGa_\bg$ to be the maximum of the $\GGa_j$'s. We shall refer only
to $\GGa_\bg$, as opposed to the $\GGa_j$'s, in the main body of this paper.

\medskip
\item
The estimates in this paper, and in particular, the bounds on the background
and critical fields, depend heavily on bounds on various linear operators
like the averaging operators of Definitions \ref{defHTbackgrounddomaction}.a
and \ref{defHTbasicnorm}.d, the covariances \ref{remHTbasicnorm}, 
and the Green's functions of Theorem \ref{HTthminvertibleoperators}. 
Such bounds have been proven in \cite{POA}. They involve constants $\Gam_1, \Gam_2, \cdots$ that are 
all independent of $L$ and $n$. In \cite[Convention \convPOconstants]{POA}, 
we define $\Gam_\op$ to be the maximum of the $\Gam_j$'s. Again, 
we shall refer only to $\Gam_\op$, as opposed to the $\Gam_j$'s, in 
this paper.
\end{itemize}

\medskip
\noindent Here are the conventions that we use in naming the various
constants that appear in this paper.
\begin{itemize}[leftmargin=*, topsep=2pt, itemsep=2pt, parsep=0pt]
\item
The constants $\Gam_\op$ and $\mu_{\rm up}$ were defined in
\cite[Convention \convPOconstants\ and Proposition \POGmainpos]{POA}, 
respectively.  They are independent of $n$ and $L$.

\item 
The constants $\GGa_j$'s and their maximum $\GGa_\bg$
and the constants $\rrho_j$, and their minimum $\rrho_\bg$ are defined
in \cite{BGE}. 
They are independent of $n$ and $L$.

\item
The constants $\CC_{\de\cV}$, $\CC_\cR$, $\CC_\fl$,
and $\CC_\ren$ are the more important $n$ and $L$ independent constants
of the main body of this paper. They depend only on $\Gam_\op$, $\GGa_\bg$, $\rrho_\bg$ and $m$.

\item The constants $\LLa_{\de\mu}$, $\LLa_j$ and $\LLa'_j$ are 
independent of $n$, but depend on $L$.

\item 
The constants $\cc_\loc$, $\cc_A$, $\cc_\Om$, $\cc_{\de\cV}$,
$\GGa_\Phi$, $\cc_\gar$, $\cc_{\mu_*}$ and $\cc_j$ are the less important $n$ 
and $L$ independent constants.
They depend only on $\Gam_\op$, $\GGa_\bg$, $\rrho_\bg$ and $m$.
\end{itemize}

\newpage
%%%%%%%%%%%%%%%%%%%%%%%%%%%%%%%%%%
\section{Scaling}\label{chapSCscaling}
%%%%%%%%%%%%%%%%%%%%%%%%%%%%%%%%%%

We extend Definition \ref{defHTscaling} to
\begin{definition}[Scaling]\label{defSCscaling}
\ 
\begin{enumerate}[label=(\alph*), leftmargin=*]
\item
As in Definition  \ref{defHTbackgrounddomaction}.a, let $\bbbl$ be 
the linear isomorphism
\begin{equation*}
\bbbl : \cX_j^{(k)}\rightarrow \cX_{j-1}^{(k)}
    \qquad , \qquad
  (u_0,\bu) \mapsto (L^2u_0, L\bu)
\end{equation*}
For a field $\,\al\,$ on $\cX_{j-1}^{(k)} $, we define the scaled fields 
\begin{equation*}
(\bbbs\al)(u)= L^{3/2}\,\al\big(\bbbl u\big)\qquad
(\bbbs_\nu\al)(u)= \left.\begin{cases}L^{7/2}& \text{if $\nu=0$}\\
                                 L^{5/2}& \text{if $\nu\in\{1,2,3\}$}
                          \end{cases}\right\}
                     \,\al\big(\bbbl u\big)
\end{equation*}
on $\,\cX_j^{(k)}\,$. For $\tilde\al = 
               \big(\al,\{\al_\nu \}_{\nu=0}^3\big)\in\tilde\cH_{j-1}^{(k)}$
as in \eqref{eqnTHdefexpandedstates}, we define 
\begin{equation*}
\bbbs\tilde\al= \big(\bbbs\al,
            \{\bbbs_\nu\al_\nu \}_{\nu=0}^3\big)\in\tilde\cH_j^{(k)}
\end{equation*}

\item
For a complex valued function $\,F(\al_*,\al)\,$ of fields on
$\cX_{j-1}^{(k)} $, we define the function $\,(\bbbs F)(\be_*,\be)\,$
of fields on $\,\cX_j^{(k)}\,$ by
\begin{equation*}
(\bbbs F)(\be_*,\be)
   = F\big(\bbbs^{-1}\be_*,\bbbs^{-1}\be)
\end{equation*}
Similarly, for a function $\,\tilde F(\tilde \al_*,\tilde \al)\,$ on subset
of $\tilde\cH_{j-1}^{(k)}\times \tilde\cH_{j-1}^{(k)}$, we define the function 
$\,(\bbbs\tilde F)(\tilde \be_*,\tilde \be)\,$ on a corresponding subset
of $\tilde\cH_{j}^{(k)}\times \tilde\cH_{j}^{(k)}$  by
\begin{equation*}
(\bbbs \tilde F)(\tilde \be_*,\tilde \be)
   = \tilde F\big(\bbbs^{-1}\tilde\be_*,\bbbs^{-1}\tilde\be)
\end{equation*}
\end{enumerate}
\end{definition}

\begin{remark}\label{remSCscaling}
\ 
\begin{enumerate}[label=(\alph*), leftmargin=*]
\item  
The definition of $\bbbs$, acting on $\cH_{j-1}^{(k)}$, 
can be rephrased, using the notation $\bbbl_*$ of 
Definition  \ref{defHTbackgrounddomaction}.a, as
$
\bbbs =L^{3/2}\bbbl_*^{-1}
$.
In particular conjugation with $\bbbs$ is the same as conjugation with
$\bbbl_*^{-1}$.

\item 
The definition of $\bbbs_\nu$ is motivated by
\begin{equation*}
\bbbs_\nu\partial_\nu=\partial_\nu \bbbs\qquad 0\le\nu\le 3
\end{equation*}
If $\,\tilde F(\tilde \al_*,\tilde \al)\,$ is a function on a subset
of $\tilde\cH_{j-1}^{(k)}\times \tilde\cH_{j-1}^{(k)}$ and
\begin{equation*}
F(\al_*,\al)=\tilde F\big((\al_*,\{\partial_\nu\al_*\})\,,\,
                           (\al,\{\partial_\nu\al\})\big)
\end{equation*}
then
\begin{equation*}
(\bbbs F)(\be_*,\be)=(\bbbs\tilde F)        
                    \big((\be_*,\{\partial_\nu\be_*\})\,,\,
                            (\be,\{\partial_\nu\be\})\big)
\end{equation*}

\item
For $\al,\al'\in \cH_{j-1}^{(k)}$,
\begin{equation*}
\<\bbbs\al,\bbbs\al'\>_j
= L^{-2}\<\al,\al'\>_{j-1}
\end{equation*}

\item
The inverse map
$
\bbbs^{-1}:\cH_j^{(k)}\rightarrow\cH_{j-1}^{(k)}
$
is given by
$
(\bbbs^{-1}\be)(u)= L^{-3/2}\be(\bbbl^{-1}u)
$.
The adjoint $\bbbs^*=L^{-2}\bbbs^{-1}$, by part (c).

\item
By Definition \ref{defSCscaling}.b,
\begin{equation*}
\int\Big[\hskip-2pt
     \prod_{u\in\cX_j^{(k)}}\hskip-5pt\sfrac{d\be(u)^*\wedge d\be(u)}
                                                                  {2\pi i}\Big]
     (\bbbs F)(\be^*,\be)
=N^{(k)}_\bbbs\int\Big[\hskip-2pt
     \prod_{v\in\cX_{j-1}^{(k)}}\hskip-5pt\sfrac{d\al(v)^*\wedge d\al(v)}
                                                                 {2\pi i}\Big]
     F(\al^*,\al)
\end{equation*}
where the normalization constant $N^{(k)}_\bbbs=\big(L^3\big)^{|\cX_j^{(k)}|}
=\big(L^3\big)^{|\cX_{j-1}^{(k)}|}$.

\item 
For a complex valued function $\,F(\al_*,\al)\,$ of fields on
$\cX_{j-1}^{(k)} $,
\begin{equation*}
\sfrac{\partial\hfill}{\partial \be(u)} (\bbbs F)(\be_*,\be)
=L^{-3/2}\sfrac{\partial\hfil F\hfil}{\partial \al(\bbbl u)}
        \big(\bbbs^{-1}\be_*,\bbbs^{-1}\be\big)
\end{equation*}

\item
Let $A:\cH_{j-1}^{(k)}\rightarrow\cH_{j-1}^{(k)}$ be a linear operator with
kernel $A(\ \cdot\ ,\ \cdot\ )$. Then the kernel of 
$\bbbs A\bbbs^{-1}:\cH_j^{(k)}\rightarrow\cH_j^{(k)}$ is
\begin{equation*}
\big(\bbbs A\bbbs^{-1}\big)(u,u')
=L^5\,A\big(\bbbl u,\bbbl u'\big)
\end{equation*}

\item
Let
\begin{equation*}
\tilde\cM\big(\,(\al_*,\{\al_{*\nu}\})\,,\,
                            (\al,\{\al_\nu\})\,\big)
=\int_{\cX_{j-1}^{(k)}} dv_1\cdots dv_n\ M(v_1,\cdots,v_n)
             \smprod_{\ell=1}^n\al_{\si_\ell}(v_\ell)
\end{equation*}
be a monomial of degree $n$. Here each $\al_{\si_\ell}$ is one of 
$\al_*,\al,
          \big\{\al_{*\nu},\al_\nu\big\}_{\nu=0}^3$.
We denote by  
\begin{itemize}[leftmargin=*, topsep=2pt, itemsep=2pt, parsep=0pt]
\item 
$n_u$, the number of $\al_{\si_\ell}$'s that is either 
$\al_*$ or $\al$ and 
\item
$n_0$, the number of $\al_{\si_\ell}$'s that is either 
$\al_{*0}$ or $\al_0$ and 
\item
$n_\sp$, the number of $\al_{\si_\ell}$'s that is one of 
$\big\{\al_{*\nu},\al_\nu\big\}_{\nu=1}^3$.
\end{itemize}
Then
\begin{equation*}
\big(\bbbs\tilde\cM\big)\big(\,(\be_*,\{\be_{*\nu}\})\,,\,
                            (\be,\{\be_\nu\})\,\big)
=\int_{\cX_j^{(k)}} du_1\cdots du_n\  M^{(s)}(u_1,\cdots,u_n)
             \smprod_{\ell=1}^n\be_{\si_\ell}(u_\ell)
\end{equation*}
has kernel
\begin{equation*}
 M^{(s)}(u_1,\cdots,u_n)
=L^{\frac{7}{2}n_u +\frac{3}{2}n_0+\frac{5}{2}n_\sp}  
    M(\bbbl u_1,\cdots,\bbbl u_n)            
\end{equation*}

\end{enumerate}
\end{remark}
\begin{definition}\label{defSCacheck}
Let $n\ge 1$. 
The dominant contribution, $A_n$, to the effective action
was defined in Definition \ref{defHTbackgrounddomaction}.b.
Its scaled version is
\begin{equation*}
\check A_n(\th_*,\th,\check\phi_*,\check\phi,\mu,\cV)
=(\bbbs^{-1}A_n)(\th_*,\th,\check\phi_*,\check\phi,L^2\mu,\bbbs\cV)
=A_n(\bbbs\th_*,\bbbs\th,\bbbs\check\phi_*,\bbbs\check\phi,L^2\mu,\bbbs\cV)
\end{equation*}
where $\th_*,\th\in \cH_{-1}^{(n)}$, 
$\check\phi_*,\check\phi\in \cH_{n-1}^{(0)}$,
$\mu\in\bbbc$ and $\cV$ is a quartic monomial in the fields $\phi_*,\phi$.

\end{definition}

\begin{lemma}\label{lemSCacheckOne}
\ \begin{enumerate}[label=(\alph*), leftmargin=*]
\item  For each $n\ge 0$, $\cV^{(u)}_n=\bbbs^n\cV_0$,
            $D_n=L^{2n}\,\bbbs^n D_0\bbbs^{-n}$ and
            $Q^{(n)}=\bbbs^n Q\bbbs^{-n}$.

\item  Set, for each $n\ge 1$,  $\check Q_n = \bbbs^{-1} Q_n\bbbs$
and $\check \fQ_n=\sfrac{1}{L^2}\bbbs^{-1} \fQ_n\bbbs$. Then
$\check Q_n=QQ_{n-1}$ (with $Q_0=\bbbone$) and 
\begin{equation*}
\check \fQ_n
= \begin{cases}
   \sfrac{a}{L^2}\bbbone & \text{if $n=1$}\\
   \noalign{\vskip0.05in}
   \big(\sfrac{L^2}{a}\bbbone+Q\,\fQ_{n-1}^{-1}Q^*\big)^{-1} & \text{if $n\ge 2$}
    \end{cases}
\end{equation*}

\item
For all $n\ge 1$,
\begin{align*}
\check A_n(\th_*,\th,\check\phi_*,\check\phi,\mu,\cV)
&=\<\th_*-QQ_{n-1}\check\phi_*\,,\,
           \check\fQ_n\big(\th-QQ_{n-1}\check\phi\big)\>_{-1}
 +      \< \check\phi_*,\,D_{n-1}\check\phi\>_{n-1} \\
&\hskip2in     + \cV(\check\phi_*,\check\phi)
-\mu \< \check\phi_*,\,\check\phi\>_{n-1}
\end{align*}
In particular
$
\check A_1(\th_*,\th,\psi_*,\psi,\mu,\cV)
= \sfrac{a}{L^2}\< \th_*\!-Q\psi_*,\th-Q\psi\>_{-1} + A_0(\psi_*,\psi,\mu,\cV)
$
\end{enumerate}
\end{lemma}
\begin{proof} (a) By part (h) of Remark \ref{remSCscaling}, the kernel of $\bbbs^n\cV_0$ 
is 
\begin{equation*}
L^{14 n}\, V_0(\bbbl^n u_1,\bbbl^n u_2,\bbbl^n u_3,\bbbl^n u_4)
=V_n^{(u)}(u_1,u_2,u_3,u_4)
\end{equation*}
by Definition \ref{defHTbackgrounddomaction}.a. The remaining two 
claims follow immediately from Remark \ref{remSCscaling}.a and 
Definition \ref{defHTbackgrounddomaction}.a.

\Item (b) The first part follows immediately from part (a) and
Definition \ref{defHTbackgrounddomaction}.a. By 
Definition \ref{defHTbackgrounddomaction}.b, when $n\ge 2$,
\begin{align*}
\bbbs^{-1}\fQ_n^{-1}\bbbs&=\sfrac{1}{a}\big(\bbbone
             +\smsum_{j=1}^{n-1}\sfrac{1}{L^{2j}}QQ_{j-1} Q_{j-1}^*Q^*\big)\\
&=\sfrac{1}{a}\Big(\bbbone  +\sfrac{1}{L^2}Q
    \Big[\bbbone+\smsum_{j=1}^{n-2}\sfrac{1}{L^{2j}}Q_j Q_j^*\Big]Q^*\Big)\\
&=\sfrac{1}{L^2}\big(\sfrac{L^2}{a}\bbbone  +Q\fQ_{n-1}^{-1}Q^*\big)
\end{align*}

\Item (c) By definition
\begin{align*}
&\check A_n(\th_*,\th,\check\phi_*,\check\phi,\mu,\cV)
% & =A_n(\bbbs\th_*,\bbbs\th,\bbbs\check\phi_*,\bbbs\check\phi,L^2\mu)\\
 =\<\bbbs\th_*-Q_n\bbbs\check\phi_*\,,\,
           \fQ_n\big(\bbbs\th-Q_n\bbbs\check\phi\big)\>_0
 +\< \bbbs\check\phi_*,\,D_n\bbbs\check\phi\>_n \\
&\hskip2in     + (\bbbs\cV)(\bbbs\check\phi_*,\bbbs\check\phi)
-L^2\mu \< \bbbs\check\phi_*,\,\bbbs\check\phi\>_n\\
%%%
&\hskip0.2in =L^{-2}\<\th_*-\bbbs^{-1}Q_n\bbbs\check\phi_*\,,\,
           \bbbs^{-1}\fQ_n\bbbs\big(\th-\bbbs^{-1}Q_n\bbbs\check\phi\big)\>_{-1}
 +L^{-2}\< \check\phi_*,\,\bbbs^{-1}D_n\bbbs\check\phi\>_{n-1} \\
&\hskip2in     + \cV(\check\phi_*,\check\phi)
-\mu \< \check\phi_*,\,\check\phi\>_{n-1}\\
%%%
&\hskip0.2in =\<\th_*-QQ_{n-1}\check\phi_*\,,\,
           \check\fQ_n\big(\th-QQ_{n-1}\check\phi\big)\>_{-1}
 +      \< \check\phi_*,\,D_{n-1}\check\phi\>_{n-1} \\
&\hskip2in     + \cV(\check\phi_*,\check\phi)
-\mu \< \check\phi_*,\,\check\phi\>_{n-1}
\end{align*}

\end{proof}

\newpage
%%%%%%%%%%%%%%%%%%%%%%%%%%%%%%%%%%
\section{The Background Field and its Variations}\label{chapBGAbackgroundAlg}
%%%%%%%%%%%%%%%%%%%%%%%%%%%%%%%%%%

Let $1\le n\le\np$. If $|\mu|$ and $\|\cV\|_{2m}$ are small enough,
the background fields
\begin{equation*}
\phi_{(*)n}(\,\cdot\,,\,\cdot\,,\mu,\cV):
                   \cH_0^{(n)}\times \cH_0^{(n)}\rightarrow\cH_n
\end{equation*} 
were defined in Proposition \ref{propHTexistencebackgroundfields}. They are
solutions of the background field equations
\begin{equation*}
\sfrac{\partial\hfill}{\partial\phi_*}A_n(\psi_*,\psi,\phi_*,\phi,\mu,\cV)=
\sfrac{\partial\hfill}{\partial\phi}A_n(\psi_*,\psi,\phi_*,\phi,\mu,\cV)=0
\end{equation*}
Putting in the action $A_n$ of Definition \ref{defHTbackgrounddomaction}.b, we get
\begin{equation}\label{eqnBGAbgeqns}
\begin{split}
S^{* -1}_n(\mu)\phi_* +\nabla_\phi\cV(\phi_*,\phi)
       &=Q_n^* \fQ_n\psi_*\\
S^{-1}_n(\mu)\phi+\nabla_{\phi_*}\cV(\phi_*,\phi)
              &=Q_n^* \fQ_n\psi\\
\end{split}
\end{equation}
with the $S_n(\mu)=(D_n+Q_n^*\fQ_n Q_n-\mu)^{-1}$ of 
Theorem  \ref{HTthminvertibleoperators}.  See  
\cite[Remark \remBSremarkonbackgroundfieldsB]{BlockSpin}. To evaluate 
the gradients of
$\cV$ we use

\begin{definition}\label{defBGAgradV}
Let 
\begin{equation*}
\cM(\phi_*,\phi)=\half\int_{\cX_n^4} du_1 \cdots du_4\  
                  M(u_1,u_2,u_3,u_4)\,  
                  \phi_*(u_1) \phi(u_2)\phi_*(u_3) \phi(u_4)
\end{equation*}
be a quartic monomial whose kernel $ M(u_1,u_2,u_3,u_4)$
is invariant under $u_1\leftrightarrow u_3$ and under 
$u_2\leftrightarrow u_4$. We denote its gradients by
\begin{align*}
\cM'_*(u;\ze_{*1},\ze,\ze_{*2})
&=\int du_1 du_2 du_3\  M(u_1,u_2,u_3,u)\, \ze_{*1}(u_1)\ze(u_2)\ze_{*2}(u_3)\\
\cM'(u;\ze_1,\ze_*,\ze_2)
&=\int du_2 du_3 du_4\  M(u,u_2,u_3,u_4)\, \ze_1(u_2)\ze_*(u_3)\ze_2(u_4)
\end{align*}
\end{definition}

Using this notation, the background field equations become
\begin{equation}\label{eqnBGAbgeqnsB}
\begin{split}
S^{* -1}_n(\mu)\phi_* +\cV'_*(\phi_*,\phi,\phi_*)
       &=Q_n^* \fQ_n\psi_*\\
S^{-1}_n(\mu)\phi+\cV'(\phi,\phi_*,\phi)
              &=Q_n^* \fQ_n\psi
\end{split}
\end{equation}
In \cite[Proposition \propBGEphivepssoln]{BGE} we prove that these equations
have a solution $\phi_{(*)n}(\psi_*,\psi,\mu,\cV)$ which is analytic on 
the set of all $(\psi_*,\psi)$'s obeying  $|\psi_*(x)|,|\psi(x)|< \ka(n)$.

\begin{definition}\label{defBGAphicheck}
 The scaled versions
\begin{equation*}
\check\phi_{(*)n}(\,\cdot\,,\,\cdot\,,\mu,\cV):
                   \cH_{-1}^{(n)}\times \cH_{-1}^{(n)}\rightarrow\cH_{n-1}
\end{equation*}
of $\phi_{(*)n}$ are
\begin{align*}
\check\phi_{(*)n}(\th_*,\th,\mu,\cV) 
     =\bbbs^{-1}\big[\phi_{(*)n}(\bbbs\th_*,\bbbs\th,L^2\mu,\bbbs\cV)\big]
\end{align*}
That is
\begin{align*}
\check\phi_{(*)n}(\th_*,\th,\mu,\cV)(v) 
=L^{-3/2}\phi_{(*)n}\big(\bbbs\th_*,\bbbs\th,L^2\mu,\bbbs\cV\big)(\bbbl^{-1}v)
\end{align*}
They are  analytic on  the set of all $(\th_*,\th)$'s 
obeying  $|\th_*(x)|,|\th(x)|< \sfrac{\ka(n)}{L^{3/2}}$.

\end{definition}

\begin{remark}\label{remBGAphicheck}
By Definition \ref{defSCacheck} and Remark \ref{remSCscaling}.f,
\begin{align*}
\sfrac{\partial \hfil A_n\hfil}{\partial\phi(u)}
            \big(\bbbs\th_*,\bbbs\th,\phi_*,\phi,L^2\mu,\bbbs\cV)
&=\sfrac{\partial\, \bbbs\check A_n\hfil}{\partial\phi(u)}
            \big(\bbbs\th_*,\bbbs\th,\phi_*,\phi,\mu,\cV) \\
&=L^{-3/2}\sfrac{\partial\hfil\check A_n\hfil}{\partial\check\phi(\bbbl u)}
          \big(\th_*,\th,\bbbs^{-1}\phi_*,\bbbs^{-1}\phi,\mu,\cV)
\end{align*}
Consequently, by Definition \ref{defHTbackgrounddomaction}.c,
\begin{equation*}
\check\phi_{(*)n}(\th_*,\th,\mu,\cV)
=\bbbs^{-1}\big[\phi_{(*)n}\big(\bbbs\th_*,\bbbs\th,L^2\mu,\bbbs\cV)\big]
\end{equation*}
are critical fields for $\check A_n(\th_*,\th,\check\phi_*,\check\phi,\mu,\cV)$.
\end{remark}

\begin{proposition}\label{propBGAomnibus}
Define, 
$
\psi_{(*)0}(\th_*,\th,\mu,\cV)=\check\phi_{1(*)}(\th_*,\th,\mu,\cV)
$
and, for $n\ge 1$,
$$
\psi_{(*)n}(\th_*,\th,\mu,\cV)  
   = \big(\sfrac{a}{L^2}Q^* Q+\fQ_n\big)^{-1}
        \big\{\sfrac{a}{L^2}Q^*\th_{(*)} 
             +\fQ_n\,Q_n\,\check\phi_{(*)n+1}(\th_*,\th,\mu,\cV)\big\}
$$

\begin{enumerate}[label=(\alph*), leftmargin=*]
\item 
The $\psi_{(*)n}$'s solve the critical field equations of 
Proposition \ref{propHTexistencecriticalfields}. They are  analytic on 
the set of all $(\th_*,\th)$'s obeying  
$|\th_*(x)|,|\th(x)|< \sfrac{\ka(n+1)}{L^{3/2}}$.

\item 
For $n\ge 1$, we have the composition rule
\begin{align*}
\check\phi_{(*)n+1}(\th_*,\th,\mu,\cV) 
   &= \phi_{(*)n}\big(\,\psi_{*n}(\th_*,\th,\mu,\cV)\,,\,
                       \psi_n(\th_*,\th,\mu,\cV)\,,\,\mu,\cV \big)
\end{align*}

\item 
For all $n\ge 1$, 
\begin{align*}
&\check A_{n+1}(\th_*,\th,\check\phi_{*n+1}(\th_*,\th,\mu,\cV),
                  \check\phi_{n+1}(\th_*,\th,\mu,\cV),\mu,\cV)\\
&\hskip0.1in=\sfrac{a}{L^2}\< \th_*\!-Q\psi_{*n}(\th_*,\th,\mu,\cV),
                       \th-Q\psi_n(\th_*,\th,\mu,\cV)\>_{-1}\\
&\hskip0.3in + A_n\big(\psi_{*n}(\th_*,\th,\mu,\cV),\psi_n(\th_*,\th,\mu,\cV), 
                  \check\phi_{*n+1}(\th_*,\th,\mu,\cV),
                  \check\phi_{n+1}(\th_*,\th,\mu,\cV),\mu,\cV\big)
\end{align*}
\end{enumerate}
\end{proposition}
\begin{proof} \emph{Case $n\ge 1$:}\ \ \ 
We apply the strategy of 
\cite[Remark \remBSremarkonbackgroundfields.c]{BlockSpin} with
\refstepcounter{equation}\label{eqnBGAblosckspinSub}
\begin{equation}
\begin{aligned}
&\cH=\cH_0^{(n)}   & &\cH_-=\cH_n  & &\cH_+=\cH_{-1}^{(n+1)} 
\\
&Q_-=Q_n \qquad&  &\fQ=\fQ_n  \qquad&  &b=\sfrac{a}{L^2}
\\
&D=D_n-\mu & &P(\phi_*,\phi)=\cV(\phi_*,\phi)\hidewidth
\\
&\fA(\phi_*,\phi)=\<\phi_*,D_n\phi\>_n+ \cV(\phi_*,\phi)- \mu\<\phi_*,\phi\>_n
\hidewidth
\end{aligned}
\tag{\ref{eqnBGAblosckspinSub}.a}
\end{equation}
and the background/next scale background fields
\begin{equation}
\phi_{(*)\rm bg}(\psi_*,\psi)= \phi_{(*)n}(\psi_*,\psi,\mu,\cV)
\qquad\qquad
\check\phi_{(*)\rm bg}(\th_*,\th)= \check\phi_{(*)n+1}(\th_*,\th,\mu,\cV)
\tag{\ref{eqnBGAblosckspinSub}.b}
\end{equation}
Then \cite[Proposition \propFormalFldSlns]{BlockSpin} applies, and, in particular,
gives the proof of part (a) for $n\ge 1$. Part (b) follows by the
uniqueness provision of \cite[Proposition \propFormalFldSlns]{BlockSpin}
and part (c) of \cite[Proposition \propBSconcatbackgr]{BlockSpin}. 
Then part (c) follows by \cite[Proposition \propBSconcatbackgr.b]{BlockSpin}. 

\Item {\it Case $n=0$:}\ \ \ 
It suffices to observe that, by Lemma \ref{lemSCacheckOne}.c,  the
fields $\check\phi_{1(*)}(\th_*,\th,\mu,\cV)$
are critical (with respect to $\psi_{(*)}$) for 
$\sfrac{a}{L^2}\< \th_*\!-Q\psi_*,\th-Q\psi\>_{-1} + A_0(\psi_*,\psi,\mu)$.
\end{proof}

The main part of the action, $A_n$, is expressed in terms of the background field
$\phi_{(*)n}(\psi_*,\psi,\mu_n,\cV_n)$. 
(See Theorem \ref{thmTHmaintheorem}.) 
In the fluctuation integral we make a change of variables
$\psi_{(*)}=\psi_{(*)n}(\th_*,\th,\mu_n,\cV_n)+\de\psi_{(*)}$.
(Set $D^{(n)(*)}\ze^{(*)}=\de\psi_{(*)}$ in 
Definition \ref{defHTapproximateblockspintr}.)
So we must study the impact of this change of variables on $\phi_{(*)n}$.

\begin{definition}\label{defBGAbckgndVarn} 
Let $1\le n\le \np$, and let
$
\|\cV\|_{2m}\ka(n)^2+|\mu|
$
be sufficiently small as in Proposition \ref{propHTexistencebackgroundfields}.
\begin{enumerate}[label=(\alph*), leftmargin=*]
\item
Define $\de\phi_{*n}\big(\psi_{*},\psi,\de\psi_*,\de\psi,\mu,\cV\big)$ and 
$\de\phi_n\big(\psi_{*},\psi,\de\psi_*,\de\psi,\mu,\cV\big)$ by
\begin{align*}
\phi_{(*)n}\big(\psi_{*}+\de\psi_*,\psi+\de\psi,\mu,\cV\big)
&= \phi_{(*)n}\big(\psi_{*},\psi,\mu,\cV\big)
          +  \de\phi_{(*)n}\big(\psi_{*},\psi,\de\psi_*,\de\psi,\mu,\cV\big)
\end{align*}
and set
\begin{align*}
&\de\check\phi_{(*)n+1}\big(\th_{*},\th,\de\psi_*,\de\psi,\mu,\cV\big)\\
&\hskip1in=\de\phi_{(*)n}\big(\psi_{*n}(\th_*,\th,\mu,\cV)\,,\,
       \psi_n(\th_*,\th,\mu,\cV)\,,\,\de\psi_*\,,\,\de\psi,\mu,\cV\big)
\end{align*}

\item
Define 
$\de{\check\phi_{(*)n+1}}^{(+)}\big(\th_*,\th;\de\psi_*,\de\psi,\mu,\cV\big)$
by
\begin{align*}
\de\check\phi_{(*)n+1}\big(\th_*,\th;\de\psi_*,\de\psi,\mu,\cV\big)
&=   S_n^{(*)} Q_n^* \fQ_n\,\de\psi_{(*)}
+ \de{\check\phi_{(*)n+1}}^{(+)}\big(\th_*,\th;\de\psi_*,\de\psi,\mu,\cV\big)\\
\end{align*}
where $S_n=(D_n+Q_n^*\fQ_n Q_n)^{-1} $ as in Theorem 
\ref{HTthminvertibleoperators}.
\end{enumerate}
\end{definition}

\begin{remark}\label{remBGAbckgndVarn}
\begin{enumerate}[label=(\alph*), leftmargin=*]
\item
By the composition rule Proposition \ref{propBGAomnibus}.b,
\begin{align*}
&\phi_{(*)n}\big(\psi_{*n}(\th_*,\th,\mu,\cV)+\de\psi_*\,,\,
   \psi_n(\th_*,\th,\mu,\cV)+\de\psi,\mu,\cV\big) \\
&\hskip1in= \check\phi_{(*)n+1}(\th_*,\th,\mu,\cV)
      +\de\check\phi_{(*)n+1}\big(\th_*,\th;\de\psi_*,\de\psi,\mu,\cV\big)
\end{align*}

\item
The quantities $\de\phi_{(*)n}$, $\de\check\phi_{(*)n+1}$,
$\de{\check\phi_{(*)n+1}}^{(+)}$, $S_n(\mu)$ correspond to the quantities
$\de\phi_{(*)\rm bg}$, $\de\check\phi_{(*)\rm bg}$,
$\de{\check\phi_{(*)}}^{(+)}$, $\Su$ in \cite{BlockSpin} under the 
substitution (\ref{eqnBGAblosckspinSub}.a,b). 
Hence, by \cite[Remark \remBSdephieqn]{BlockSpin}, the fields 
$\de\check\phi_{(*)n+1}$ obey
\begin{align*}
 \de\check\phi_{*n+1}
&=S_n(\mu)^*Q_n^* \fQ_n\, \de\psi_*  
  - {S_n^*}\nabla_\phi \cV (\phi_*, \phi) 
      \bigg|^{\atop{\phi_*=\check\phi_{*n+1}(\th_*,\th,\mu,\cV)
                                         +\de\check\phi_{*n+1}}
              {\phi=\check\phi_{n+1}(\th_*,\th,\mu,\cV)+\de\check\phi_{n+1}}}
     _{\atop{\phi_*=\check\phi_{*n+1}(\th_*,\th,\mu,\cV)}
            {\phi=\check\phi_{n+1}(\th_*,\th,\mu,\cV)}}
 \\
\noalign{\vskip0.05in}
\de\check\phi_{n+1}
&= S_n(\mu) Q_n^* \fQ_n\, \de\psi 
- S_n\nabla_{\phi_*}\cV (\phi_*, \phi) 
      \bigg|^{\atop{\phi_*=\check\phi_{*n+1}(\th_*,\th,\mu,\cV)
                                             +\de\check\phi_{*n+1}}
        {\phi=\check\phi_{n+1}(\th_*,\th,\mu,\cV)+\de\check\phi_{n+1}}}
     _{\atop{\phi_*=\check\phi_{*n+1}(\th_*,\th,\mu,\cV)}
        {\phi=\check\phi_{n+1}(\th_*,\th,\mu,\cV)}}
\end{align*}

\item 
Since $S_n(\mu)=\big[\bbbone-\mu S_n\big]^{-1}S_n$, the equations
of part (b) may be rewritten
\begin{align*}
 \de\check\phi_{*n+1}
&={S_n^*}Q_n^* \fQ_n\, \de\psi_*  +\mu {S_n^*}\de\check\phi_{*n+1}
  - {S_n^*}\nabla_\phi \cV (\phi_*, \phi) 
      \bigg|^{\atop{\phi_*=\check\phi_{*n+1}(\th_*,\th,\mu,\cV)
                                         +\de\check\phi_{*n+1}}
        {\phi=\check\phi_{n+1}(\th_*,\th,\mu,\cV)+\de\check\phi_{n+1}}}
     _{\atop{\phi_*=\check\phi_{*n+1}(\th_*,\th,\mu,\cV)}
        {\phi=\check\phi_{n+1}(\th_*,\th,\mu,\cV)}}
 \\
\noalign{\vskip0.05in}
\de\check\phi_{n+1}
&= S_n Q_n^* \fQ_n\, \de\psi +\mu S_n\de\check\phi_{n+1}
- S_n\nabla_{\phi_*}\cV (\phi_*, \phi) 
      \bigg|^{\atop{\phi_*=\check\phi_{*n+1}(\th_*,\th,\mu,\cV)
                                             +\de\check\phi_{*n+1}}
        {\phi=\check\phi_{n+1}(\th_*,\th,\mu,\cV)+\de\check\phi_{n+1}}}
     _{\atop{\phi_*=\check\phi_{*n+1}(\th_*,\th,\mu,\cV)}
        {\phi=\check\phi_{n+1}(\th_*,\th,\mu,\cV)}}
\end{align*}
In particular, if $\mu=\cV=0$, then $\de\check\phi_{(*)n+1}
={S_n^{(*)}}Q_n^* \fQ_n\, \de\psi_{(*)}$. This is the motivation for the
definition of $\de{\check\phi_{(*)n+1}}^{(+)}$ in 
Definition \ref{defBGAbckgndVarn}.b.
\end{enumerate}
\end{remark}

\newpage
%%%%%%%%%%%%%%%%%%%%%%%%%%%%%%%%%%
\section{One Block Spin Transformation --- The Algebra}\label{chapSTstrategy}
%%%%%%%%%%%%%%%%%%%%%%%%%%%%%%%%%%

In this section we consider the output of the approximate block spin 
transformation $\bbbt_n^{(SF)}$ acting on $e^{\cA_0}$, with the $\cA_0$
of \eqref{eqnHTaZero} in the case $n=0$, and on $e^{- A_n+\cR_n+\cE_n}$, 
in the case $n\ge 1$ 
(see Theorem \ref{thmTHmaintheorem}). The main result of this section
is Proposition \ref{propSTmainProp}, which provides a representation of this
output that will be used in the (inductive) proof of 
Theorem \ref{thmTHmaintheorem}. If the conclusion of
Theorem \ref{thmTHmaintheorem} holds for some $1\le n<\np$,
then Proposition  \ref{propSTmainProp} gives a representation 
for 
\begin{align*}
&\Big(\! \bbbt_n^{(SF)}\circ(\bbbs \bbbt_{n-1}^{(SF)}) \circ
   \cdots \circ (\bbbs \bbbt_0^{(SF)})\!\Big) 
\Big(\!e^{\cA_0(\psi_*,\psi) }\! \Big) (\th_*,\th)
\end{align*}
which, up to a multiplicative constant, is of the form
\begin{equation*}
e^{\check\cC_n(\th_*,\th)}\ \check\cF_n(\th_*,\th)
\end{equation*}
where
\begin{itemize}[leftmargin=*, topsep=2pt, itemsep=2pt, parsep=0pt]
%%%%%%%%
\item
the ``contribution from the critical field'' is
\begin{align*}
\check\cC_n(\th_*,\th)
&=- \check A_{n+1}(\th_*,\th,\check\phi_{*n+1}(\th_*,\th,\mu_n,\cV_n),
                  \check\phi_{n+1}(\th_*,\th,\mu_n,\cV_n),\mu_n,\cV_n)\\
&\hskip0.5in
+\cR_n\big( \check\phi_{*n+1}(\th_*,\th,\mu_n,\cV_n),
                  \check\phi_{n+1}(\th_*,\th,\mu_n,\cV_n)\big)
+\check\cE_{n+1,1}(\th_*,\th)
\end{align*}
with
\begin{align*}
\check\cE_{n+1,1}(\th_*,\th)
=\cE_n\big( \psi_{*n}(\th_*,\th,\mu_n,\cV_n),\psi_n(\th_*,\th,\mu_n,\cV_n)\big)
\end{align*}
and the $\mu_n$, $\cV_n$, $\cR_n$ and $\cE_n$ of 
Theorem \ref{thmTHmaintheorem}\ for $n\ge 1$ and 
of \S\ref{sectINTstartPoint} for $n=0$,
%%%%%%
\item
and the ``fluctuation integral'' is
\begin{align*}
\check\cF_n(\th_*,\th)
&=\Big[\hskip-6pt
 \prod_{x\in\cX_0^{(n)}}  \int\limits_{|\ze(x)|\le r_n} \hskip -9pt
\sfrac{d\ze(x)^*\wedge d\ze(x)}{2\pi i} e^{-|\ze(x)|^2}\Big]\\
&\hskip0.25in\exp\Big\{\!\!-\de \check A_n(\th_*,\th,\de\psi_*,\de\psi) 
      +\de\check\cR_n\big(\th_*,\th,\de\psi_*,\de\psi\big) 
+\de\check\cE_n\big(\th_*,\th,\de\psi_*,\de\psi\big)\!\Big\}
\end{align*}
with $\de\psi_*= D^{(n)*}\ze^*$, 
       $\de\psi= D^{(n)}\ze$, 
$\ D^{(n)}$ being an operator square root of $C^{(n)}$, as in 
\eqref{eqnHTcn}, and
\begin{itemize}[leftmargin=*, topsep=2pt, itemsep=2pt, parsep=0pt] 
\item
for $n\ge 0$
\begin{align*}
\de\check\cE_n(\th_*,\th,\de\psi_*,\de\psi)
&=\cE_n\big(\psi_{*n}(\th_*,\th,\mu_n,\cV_n)+\de\psi_*,
               \psi_n(\th_*,\th,\mu_n,\cV_n)+\de\psi\big)\\&\hskip3cm -\cE_n\big(\psi_{*n}(\th_*,\th,\mu_n,\cV_n),\psi_n(\th_*,\th,\mu_n,\cV_n)\big)
\end{align*}
\item
for $n\ge 0$
\begin{align*}
&\de\check\cR_n(\th_*,\th,\de\psi_*,\de\psi)\\
&\hskip0.25in=\Big[\cR_n\big(\phi_*\!+\!\de\phi_*,\phi\!+\!\de\phi\big)
       \!-\!\cR_n\big(\phi_*,\phi\big)
        \Big]_{\atop{\phi_{(*)}=\check\phi_{(*)n+1}(\th_*,\th,\mu_n,\cV_n)}
                    {\de\phi_{(*)}=\de\check\phi_{(*)n+1}
                                      (\th_*,\th;\de\psi_*,\de\psi,\mu_n,\cV_n)}}
\end{align*}
where, for $n\ge 1$,  $\de\check\phi_{(*)n+1}$ was defined in 
Definition \ref{defBGAbckgndVarn}.a and, for $n=0$,
$\de\check\phi_{1(*)}=\de\psi_{(*)}$
and, 
\item
for $n\ge 1$,
\begin{align*}
\de \check A_n(\th_*,\th,\de\psi_*,\de\psi)
&= -\int_0^1 \!dt\,  \big< \de\psi_*,\fQ_n\, Q_n\,
\de{\check\phi}^{(+)}_{n+1}\big(\th_*,\th;t\,\de\psi_*,t\,\de\psi,\,
                                                 \mu_n,\cV_n\big) \big>_0 
\\
& \hskip 1cm
-\int_0^1 \!dt\,  \big< \fQ_n\, Q_n\,
\de{\check\phi_{*n+1}}^{(+)}\big(\th_*,\th;t\,\de\psi_*,t\,\de\psi,\,
                                 \mu_n,\cV_n\big)
,\, \de\psi \big>_0 
\end{align*}
and, for $n=0$,
\begin{align*}
&\de \check A_0(\th_*,\th,\de\psi_*,\de\psi) \\
&\hskip0.5in= \int_0^1\hskip-4pt  (1 - t)\,\sfrac{d^2\hfill}{dt^2} 
   \cV_0\big(\psi_{*0}(\th_*,\th,\mu_0,\cV_0)\!+\!t\de\psi_{*}\,,\,
               \psi_0(\th_*,\th,\mu_0,\cV_0)\!+\!t\de\psi\big)  \, dt\\
&\hskip1in -\mu_0 \< \de\psi_*,\,\de\psi\>_0
\end{align*}
The integral in $\de \check A_0$ is the
part of $\cV_0\big(\psi_{*0}(\th_*,\th,\mu_0,\cV_0)\!+\!\de\psi_{*}\,,\,
               \psi_0(\th_*,\th,\mu_0,\cV_0)\!+\!\de\psi\big)$
that is of degree at least two in $\de\psi_{(*)}$.
\end{itemize}
\end{itemize}
The significance of $\de\check A_n$ may be seen in 
\begin{lemma}\label{lemSTdeA} 
\ 
\begin{enumerate}[label=(\alph*), leftmargin=*]
\item For all $n\ge 1$,
\begin{align*}
&\Big[\sfrac{a}{L^2}\!\< \th_*\!-\!Q\psi_*,\th\!-\!Q\psi\>_{-1}\!
          +\!A_n\big(\psi_*,\psi,\phi_*,\phi,\mu_n,\cV_n\big)
      \Big]_{\atop{\!\psi_{(*)}=\psi_{(*)n}(\th_*,\th,\mu_n,\cV_n)+\de\psi_{(*)}\ \ \ }
            {\!\phi_{(*)}=  
                \check\phi_{n\!+\!1(*)}(\th_*,\th,\mu_n,\cV_n)+\de\check\phi_{(*)n+1}}}
\\&\hskip0.5in     
  -\Big[\sfrac{a}{L^2}\< \th_*\!-Q\psi_*,\th-Q\psi\>_{-1}
          +A_n\big(\psi_*,\psi,\phi_*,\phi,\mu_n,\cV_n\big)
           \Big]_{\atop{\psi_{(*)}=\psi_{(*)n}(\th_*,\th,\mu_n,\cV_n)\ \ \ }
                       {\!\phi_{(*)}= \check\phi_{(*)n+1}(\th_*,\th,\mu_n,\cV_n)}}\\
&= \big<\de\psi_*,{C^{(n)}}^{-1}\,\de\psi\big>_0
+\de\check A_n(\th_*,\th,\de\psi_*,\de\psi)
\end{align*}
\item
For $n =0$,
\begin{align*}
&\Big[\sfrac{a}{L^2}\!\< \th_*\!-\!Q\psi_*,\th\!-\!Q\psi\>_{-1}\!
          +\!A_0\big(\psi_*,\psi,\mu_0,\cV_0\big)
           \Big]^{\psi_{(*)}=\psi_{0(*)}(\th_*,\th,\mu_0,\cV_0)+\de\psi_{(*)}}
                _{\psi_{(*)}=\psi_{0(*)}(\th_*,\th,\mu_0,\cV_0)}\\
&\hskip0.5in= \big<\de\psi_*,{C^{(0)}}^{-1}\,\de\psi\big>_0
+\de \check A_0(\th_*,\th,\de\psi_*,\de\psi)
\end{align*}
\end{enumerate}
\end{lemma}
\begin{proof} 
(a) We use \cite{BlockSpin}  with the substitutions 
\begin{align*}
&\cH=\cH_0^{(n)}   & &\cH_-=\cH_n  & &\cH_+=\cH_{-1}^{(n+1)} 
\\
&Q_-=Q_n \qquad&  &\fQ=\fQ_n  \qquad&  &b=\sfrac{a}{L^2}
\\
&D=D_n & &P(\phi_*,\phi)=\cV_n(\phi_*,\phi)- \mu_n\<\phi_*,\phi\>_n
\hidewidth\\
&\fA(\phi_*,\phi)=\<\phi_*,D_n\phi\>_n+ \cV_n(\phi_*,\phi)
           - \mu_n\<\phi_*,\phi\>_n
\hidewidth
\end{align*}
They give
\begin{align*}
\cA_{\rm eff}(\th_*,\th;\psi_*,\psi;\phi_*,\phi) 
%&= b \< \th_*-Q\psi_*\,,\, \th-Q\psi \>_+ 
%   + \< \psi_*\!-\!Q_-\,\phi_*\,,\, \fQ(\psi\!-\!Q_-\,\phi) \> 
%   + \fA(\phi_*,\phi)\\
&=\sfrac{a}{L^2}\< \th_*-Q\psi_*\,,\, \th-Q\psi \>_{-1}
   + \< \psi_*-Q_n\,\phi_*\,,\, \fQ_n(\psi-Q_n\,\phi) \>_0 \\
&\hskip0.5in + \<\phi_*,D_n\phi\>_n+ \cV_n(\phi_*,\phi)
  - \mu_n\<\phi_*,\phi\>_n\\
&=\sfrac{a}{L^2}\< \th_*-Q\psi_*\,,\, \th-Q\psi \>_{-1}
   + A_n(\psi_*,\psi,\phi_*,\phi,\mu_n,\cV_n) 
\end{align*} 
Comparing \cite[(\eqnBSdefinitionScheckS), (\eqnBSdefCascovariance) 
and Remark \remBSedA.a]{BlockSpin} with Theorem \ref{HTthminvertibleoperators}, 
\eqref{eqnHTcn} and \eqref{eqnHTden}, we have $C=C^{(n)}$ and $S=S_n$.
Also 
\begin{align*}
\phi_{(*)\rm bg}(\psi_*,\psi)&= \phi_{(*)n}(\psi_*,\psi,\mu_n,\cV_n)\\
\psi_{(*)\rm cr}(\th_*,\th)&= \psi_{(*)n}(\th_*,\th,\mu_n,\cV_n)\\
\check\phi_{(*)\rm bg}(\th_*,\th)
     &= \check\phi_{(*)n+1}(\th_*,\th,\mu_n,\cV_n)
\end{align*}
are the background, critical and next scale background fields, respectively, 
in the sense of \cite[Definition \defBSbackfld]{BlockSpin}.

\noindent The claim now follows from 
\cite[Lemma \lemBSdeltaAalernew]{BlockSpin}.

\Item (b) Observe that 
\begin{equation*}
\sfrac{a}{L^2} \< \th_*-Q\psi_*,\th-Q\psi\>_{-1}
    +A_0\big(\psi_*,\psi,\mu_0,\cV_0\big)
\end{equation*} 
is the sum of the quadratic
form $\sfrac{a}{L^2} \< \th_*-Q\psi_*,\th-Q\psi\>_{-1}
+\< \psi_*,\,D_0\psi\>_0 -\mu_0 \< \psi_*,\,\psi\>_0$ 
and $\cV_0(\psi_*,\psi)$. Now imagine substituting in 
$\psi_{(*)}=\psi_{0(*)}(\th_*,\th,\mu_0,\cV_0)+\de\psi_{(*)}$ and expanding
in powers of $\de\psi_{(*)}$. The total contribution that is of degree
precisely one in $\de\psi_{(*)}$ vanishes by the criticality of $\psi_{(*)0}$. 
By Taylor's theorem with remainder,
\begin{align*}
&\Big[\sfrac{a}{L^2}\!\< \th_*\!-\!Q\psi_*,\th\!-\!Q\psi\>_{-1}\!
          +\!A_0\big(\psi_*,\psi,\mu_0,\cV_0\big)
           \Big]^{\psi_{(*)}=\psi_{0(*)}(\th_*,\th,\mu_0,\cV_0)+\de\psi_{(*)}}
                _{\psi_{(*)}=\psi_{0(*)}(\th_*,\th,\mu_0,\cV_0)}\\
&\hskip0.5in=\sfrac{a}{L^2} \< Q\de\psi_*,Q\de\psi\>_{-1}
+\< \de\psi_*,\,D_0\de\psi\>_0 -\mu_0 \< \de\psi_*,\,\de\psi\>_0\\
&\hskip1in
+ \int_0^1  (1 - t)\,\sfrac{d^2\hfill}{dt^2} 
   \cV_0\big(\psi_{0*}(\th_*,\th,\mu_0,\cV_0)+t\de\psi_{*}\,,\,
               \psi_0(\th_*,\th,\mu_0,\cV_0)+t\de\psi\big)  \, dt 
\end{align*}
\end{proof}

\begin{proposition}\label{propSTmainProp}
\ 
\begin{enumerate}[label=(\alph*), leftmargin=*]
\item
Let $1\le n<\np$. Let $\mu_n$, $\cV_n$, $\cR_n$ and $\cE_n$ be as in 
Theorem \ref{thmTHmaintheorem}. Set
\begin{align*}
\cA_n(\psi_*,\psi)&=\Big[- A_n(\psi_*,\psi, \phi_{*n}, \phi_n,\,\mu_n,\cV_n)
+\cR_n(\phi_{*n},\phi_n)\Big]_{\phi_{(*)n}=\phi_{(*)n}(\psi_*,\psi,\mu_n,\cV_n)}\\
&\hskip0.5in+\cE_n(\psi_*,\psi)
\end{align*}
Then
\begin{align*}
&\bbbt_n^{(SF)} 
\Big(e^{\cA_n(\psi_*,\psi)}\Big) (\th_*,\th;\mu_n,\cV_n)
= \sfrac{1}{\tilde N^{(n)}_\bbbt}\ e^{\check\cC_n(\th_*,\th)}
    \ \check\cF_n(\th_*,\th)
\end{align*}

\item
For $n=0$,\ \ \ 
$
\bbbt_0^{(SF)} 
\Big(e^{\cA_0(\psi_*,\psi)}\Big) (\th_*,\th;\mu_0,\cV_0)
= \sfrac{1}{\tilde N^{(0)}_\bbbt}\ e^{\check\cC_0(\th_*,\th)}
    \ \check\cF_0(\th_*,\th)
$.
\end{enumerate}
\end{proposition}

\begin{proof}

%%%%%%%%%%%%%%%%%%%%%%%%%%%%
\noindent{Case $n\ge 1$:}\ \ \ 
%%%%%%%%%%%%%%%%%%%%%%%%%%%
By  Definition \ref{defHTapproximateblockspintr},
\begin{equation}\label{eqnSTtheSFintegral}
\begin{split}
&\bbbt_n^{(SF)} 
\Big(e^{\cA_n(\psi_*,\psi)}\Big) (\th_*,\th;\mu_n,\cV_n)
\\
&\hskip0.2in= \sfrac{1}{\tilde N^{(n)}_\bbbt} \Big[\hskip-6pt
 \prod_{x\in\cX_0^{(n)}}  \int\limits_{|\ze(x)|\le r_n} \hskip -11pt
\sfrac{d\ze(x)^*\wedge d\ze(x)}{2\pi i}\Big]
e^{-\sfrac{a}{L^2}\< \th_*\!-Q\psi_*,\th-Q\psi\>_{-1}}
e^{\cA_n(\psi_*,\psi,\mu_n,\cV_n)}
\bigg|_{\psi_{(*)} = \psi_{(*)n} + \de\psi_{(*)}}
\end{split}
\end{equation}
where $\psi_{(*)n}=\psi_{(*)n}(\th_*,\th,\mu_n,\cV_n)$,
        $\de\psi_*= D^{(n)*}\ze^*$, 
       $\de\psi= D^{(n)}\ze$, 
and
\begin{align*}
\cA_n(\psi_*,\psi,\mu_n,\cV_n)
&=- A_n(\psi_*,\psi, \phi_*, \phi,\,\mu_n,\cV_n)
        +\cR_n(\phi_*,\phi)\\
&\hskip2in+\cE_n(\psi_*,\psi)
   \Big|_{\atop{\phi_*=\phi_{*n}(\psi_*,\psi,\mu_n,\cV_n)}
               {\phi=\phi_n(\psi_*,\psi,\mu_n,\cV_n)}}
\end{align*}
%%%%%%%%%%%%%%%%%%%%%%%%%%%%%%%%%%%%%%%%%%%%%%%%%%%%%%%%
When $\ze=0$ the exponent of the integral \eqref{eqnSTtheSFintegral}
reduces to
\begin{align*}
&-\sfrac{a}{L^2}\< \th_*\!-Q\psi_{*n}(\th_*,\th,\mu_n,\cV_n),
                       \th-Q\psi_n(\th_*,\th,\mu_n,\cV_n)\>_{-1}\\
  &\hskip0.25in - A_n\big(\psi_{*n},\psi_n, 
                  \check\phi_{*n+1}(\th_*,\th,\mu_n,\cV_n),
                  \check\phi_{n+1}(\th_*,\th,\mu_n,\cV_n),\mu_n,\cV_n\big)
                  \Big|_{\psi_{(*)n}=\psi_{*n}(\th_*,\th,\mu_n,\cV_n)}\\
&\hskip0.5in
+\cR_n\big( \check\phi_{*n+1}(\th_*,\th,\mu_n,\cV_n),
                  \check\phi_{n+1}(\th_*,\th,\mu_n,\cV_n)\big)\\
&\hskip0.5in
+\cE_n\big( \psi_{*n}(\th_*,\th,\mu_n,\cV_n),\psi_n(\th_*,\th,\mu_n,\cV_n)\big)\\
&=- \check A_{n+1}(\th_*,\th,\check\phi_{*n+1}(\th_*,\th,\mu_n,\cV_n),
                  \check\phi_{n+1}(\th_*,\th,\mu_n,\cV_n),\mu_n,\cV_n)\\
&\hskip0.5in
+\cR_n\big( \check\phi_{*n+1}(\th_*,\th,\mu_n,\cV_n),
                  \check\phi_{n+1}(\th_*,\th,\mu_n,\cV_n)\big)\\
&\hskip0.5in
+\cE_n\big( \psi_{*n}(\th_*,\th,\mu_n,\cV_n),\psi_n(\th_*,\th,\mu_n,\cV_n)\big)
\end{align*}
by Proposition \ref{propBGAomnibus}.b,c.
%%%%%%%%%%%%%%%%%%%%%%%%%%%%%%%%%%%%%%%%%%%%%%%%%%%%%%%%

The part of the exponent of the integral \eqref{eqnSTtheSFintegral}
that is of degree at least one in $\de\psi$ is
\begin{align*}
&-\big<\de\psi_*,{C^{(n)}}^{-1}\,\de\psi\big>_0
-\de \check A_n(\th_*,\th,\de\psi_*,\de\psi)
      +\de\check\cR_n\big(\th_*,\th,\de\psi_*,\de\psi\big) \\
&\hskip3.5in    +\de\check\cE_n\big(\th_*\,,\,\th\,,\, \de\psi_*\,,\,\de\psi\big)
\end{align*}
by Lemma \ref{lemSTdeA}.a. Since $\ D^{(n)}$ is an operator square root of $C^{(n)}$
\begin{equation*}
\big<\de\psi_*,{C^{(n)}}^{-1}\,\de\psi\big>_0
    \Big|_{\de\psi_*= D^{(n)*}\ze^*
            \atop \de\psi= D^{(n)}\ze}
=\<\ze^*,\ze\>_0
\end{equation*}

%%%%%%%%%%%%%%%%%%%%%%%%%%%%%%%%%%%%%%%%%%%%%%%%%%%%%%%%

%%%%%%%%%%%%%%%%%%%%%%%%%%%%%%%%%%%
\Item {Case $n=0$:}\ \ \ 
%%%%%%%%%%%%%%%%%%%%%%%%%%%%%%%%%%%
By  Definition \ref{defHTapproximateblockspintr},
\begin{equation}\label{eqnSTtheSFintegralnzero}
\begin{split}
&(\bbbt_0^{(SF)} e^{\cA_0})(\th_*,\th;\mu_0,\cV_0) \\
&\hskip0.1in=\sfrac{1}{\tilde N^{(0)}_\bbbt} \Big[\hskip-4pt
 \prod_{x\in\cX_0}\  \int_{|\ze(x)|\le r_0} \hskip -9pt
\sfrac{d\ze(x)^*\wedge d\ze(x)}{2\pi i}\Big]
e^{-a L^{-2}\< \th_*\!-Q\psi_*,\th-Q\psi\>_{-1}}\,e^{\cA_0(\psi_*,\psi)}
\bigg|_{\atop{\psi_* = \psi_{0*}(\th_*,\th,\mu_0,\cV_0) + \de\psi_*}
             {\psi = \psi_{0}(\th_*,\th,\mu_0,\cV_0) +\de\psi \ \ \ }}
\end{split}
\end{equation}
where, again, $\de\psi_*= D^{(0)*}\ze^*$, 
       $\de\psi= D^{(0)}\ze$.  
By \eqref{eqnHTaZero} and Definition \ref{defHTbackgrounddomaction}.b,
when $\ze=0$ the exponent of the integral \eqref{eqnSTtheSFintegralnzero}
reduces to
\begin{align*}
&-a L^{-2}\< \th_*-Q\psi_{0*}(\th_*,\th,\mu_0,\cV_0),
         \th-Q\psi_0(\th_*,\th,\mu_0,\cV_0)\>\\
    &\hskip0.5in - A_0\big(\psi_{0*}(\th_*,\th,\mu_0,\cV_0),\psi_0(\th_*,\th,\mu_0,\cV_0)
   ,\mu_0,\cV_0\big)\\
&\hskip1in
+\cR_0\big( \psi_{0*}(\th_*,\th,\mu_0,\cV_0),\psi_0(\th_*,\th,\mu_0,\cV_0)\big)
\\&\hskip1.5in+\cE_0\big( \psi_{0*}(\th_*,\th,\mu_0,\cV_0),\psi_0(\th_*,\th,\mu_0,\cV_0)\big)
\\
&=- \check A_1(\th_*,\th,\check\phi_{1*}(\th_*,\th,\mu_0,\cV_0),
                  \check\phi_1(\th_*,\th,\mu_0),\mu_0,\cV_0)\\
&\hskip0.5in
+\cR_0\big( \check\phi_{1*}(\th_*,\th,\mu_0,\cV_0),
                  \check\phi_1(\th_*,\th,\mu_0,\cV_0)\big)\\
&\hskip1in+\cE_0\big( \psi_{0*}(\th_*,\th,\mu_0,\cV_0),
            \psi_0(\th_*,\th,\mu_0,\cV_0)\big)
\end{align*}
by Lemma \ref{lemSCacheckOne}.c and Proposition \ref{propBGAomnibus}.b.

%%%%%%%%%%%%%%%%%%%%%%%%%%%%%%%%%%%%%%%%%%%%%%%%

The part of the exponent of the integral \eqref{eqnSTtheSFintegralnzero}
that is of degree at least one in $\de\psi$ is
\begin{align*}
&-\big<\de\psi_*,{C^{(0)}}^{-1}\,\de\psi\big>_0
-\de\check A_0(\th_*,\th,\de\psi_*,\de\psi)
      +\de\check\cR_0\big(\th_*,\th, \de\psi_*,\de\psi\big) \\
&\hskip3.5in  +\de\check\cE_0\big(\th_*\,,\,\th\,,\, \de\psi_*\,,\,\de\psi\big)
\end{align*}
by Lemma \ref{lemSTdeA}.b. 

\end{proof}

Corollary \ref{corSTmainCor}, below, gives a representation for
\begin{align*}
&\Big(\!(\bbbs \bbbt_n^{(SF)}) \circ
   \cdots \circ (\bbbs \bbbt_0^{(SF)})\!\Big) 
\Big(\!e^{\cA_0 }\! \Big) (\psi_*,\psi)\\
&\hskip1in=\Big(\! \bbbt_n^{(SF)}\circ(\bbbs \bbbt_{n-1}^{(SF)}) \circ
   \cdots \circ (\bbbs \bbbt_0^{(SF)})\!\Big) 
\Big(\!e^{\cA_0}\! \Big) (\bbbs^{-1}\psi_*,\bbbs^{-1}\psi)
\end{align*}
which, up to a multiplicative constant, is of the form
\begin{equation*}
e^{\cC_n(\psi_*,\psi)}\ \cF_n(\psi_*,\psi)
\end{equation*}
where
%%%%%%%%
\begin{itemize}[leftmargin=*, topsep=2pt, itemsep=2pt, parsep=0pt]
\item
the ``contribution from the critical field'' is
\begin{align*}
\cC_n(\psi_*,\psi)
&=\!- A_{n+1}(\psi_*,\psi,\phi_{*n+1}(\psi_*,\psi,L^2\mu_n,\bbbs\cV_n),
                        \phi_{n+1}(\psi_*,\psi,L^2\mu_n,\bbbs\cV_n),
                        L^2\mu_n,\bbbs\cV_n)\\
&\hskip0.1in
+(\bbbs \cR_n)\big(\phi_{*n+1}(\psi_*,\psi,L^2\mu_n,\bbbs\cV_n),
                    \phi_{n+1}(\psi_*,\psi,L^2\mu_n,\bbbs\cV_n)\big)
+\cE_{n+1,1}(\psi_*,\psi)
\end{align*}
with
\begin{align}
\cE_{n+1,1}(\psi_*,\psi)
&=(\bbbs\cE_n)\big( \hat\psi_{*n}(\psi_*,\psi,\mu_n,\cV_n),
                   \hat\psi_n(\psi_*,\psi,\mu_n,\cV_n)\big)\nonumber\\
\hat \psi_{(*)n}(\psi_*,\psi,\mu,\cV)
&=\bbbs\big[\psi_{*n}(\bbbs^{-1}\psi_*,\bbbs^{-1}\psi,\mu,\cV)\big]
\label{eqnSThatpsi}
\end{align}
and the $\mu_n$, $\cV_n$, $\cR_n$ and $\cE_n$ of 
Theorem \ref{thmTHmaintheorem}\ for $n\ge 1$ and 
of \S\ref{sectINTstartPoint} for $n=0$,
%%%%%%
\item
and the ``fluctuation integral'' is
\begin{equation}\label{eqnOSAfluctInt}
\begin{split}
\cF_n(\psi_*,\psi)
&=\Big[\hskip-3pt
 \prod_{w\in\cX_1^{(n)}}  \int\limits_{|z(w)|\le r_n} \hskip -6pt
\sfrac{dz(w)^*\wedge dz(w)}{2\pi i} e^{-|z(w)|^2}\Big]\\
&\hskip0.6in\exp\Big\{\!-\de A_n(\psi_*,\psi,z_*,z) 
      +\de\cR_n(\psi_*,\psi,z_*,z) 
+\de\cE_n\big(\psi_*,\psi,z_*,z\big)\!\Big\}
\end{split}
\end{equation}
with 
\begin{itemize}[leftmargin=*, topsep=2pt, itemsep=2pt, parsep=0pt]
\item
for $n\ge 0$
\begin{equation}\label{eqnOSAdeEndef}
\de\cE_n(\psi_*,\psi,z_*,z)
=(\bbbs \cE_n)(\Psi_*,\Psi)\Big|
       ^{\Psi_{(*)}=\hat\psi_{(*)n}(\psi_*,\psi,\mu_n,\cV_n)
                        +L^{3/2}\bbbs D^{(n)(*)}\bbbs^{-1}z_{(*)}}
       _{\Psi_{(*)}=\hat\psi_{(*)n}(\psi_*,\psi,\mu_n,\cV_n)}
\end{equation}
\item
for $n\ge 0$
\begin{equation}\label{eqnOSAdeRndef}
\de\cR_n(\psi_*,\psi,z_*,z)
=(\bbbs\cR_n)(\Phi_*,\Phi)\Big|
       ^{\Phi_{(*)}=\phi_{(*)n+1}(\psi_*,\psi,L^2\mu_n,\bbbs\cV_n)
                     +\de\hat\phi_{(*)n+1}(\psi_*,\psi,z_*,z)}
       _{\Phi_{(*)}=\phi_{(*)n+1}(\psi_*,\psi,L^2\mu_n,\bbbs\cV_n)}
\end{equation}
where, 
\begin{equation}\label{eqnSTdehatphidef}
\begin{split}
&\de\hat\phi_{(*)n+1}(\psi_*,\psi,z_*,z)\\
&\hskip0.2in=\begin{cases}
\bbbs\big[\de\check\phi_{(*)n+1}\big(\bbbs^{-1}\psi_*\,,\,
      \bbbs^{-1}\psi_*\,,\,
      D^{(n)*}\bbbl_* z_*\,,\,
      D^{(n)}\bbbl_* z\,,\,\mu_n,\cV_n\big)\big] & \text{if $n\ge 1$}\\
\noalign{\vskip0.1in}
L^{3/2}\bbbs D^{(0)(*)}\bbbs^{-1}z_{(*)} & \text{if $n=0$}
     \end{cases}
\end{split}
\end{equation}
\refstepcounter{equation}\label{eqnOSAdeAndef}
and
\item
for $n\ge 1$,
\begin{equation}
\begin{split}
\de A_n(\psi_*,\psi,z_*,z)
&= -L^{7/2} \int_0^1 \!dt\,  \big< z_*,\bbbs D^{(n)}\fQ_n\, Q_n\bbbs^{-1}\,
    [\de{\hat\phi}^{(+)}_{n+1}\big(\psi_*,\psi; t\,z_*,t\,z\big)] \big>_1
\\
& \hskip0.2in
-L^{7/2}\int_0^1 \!dt\,  \big< \bbbs D^{(n)*}\fQ_n\, Q_n\bbbs^{-1}\,
[\de{\hat\phi_{*n+1}}^{(+)}\big(\psi_*,\psi;t\,z_*,t\,z\big)]
,\, z \big>_1 
\end{split}
\tag{\ref{eqnOSAdeAndef}.a}
\end{equation}
and, for $n=0$,
\begin{equation}
\begin{split}
&\de A_0(\psi_*,\psi,z_*,z)\\
&\hskip0.2in=  \int_0^1\hskip-4pt  (1 - t)\,\sfrac{d^2\hfill}{dt^2} 
   (\bbbs\cV_0)\big(\hat\psi_*+t\de\psi_*\,,\,
               \hat\psi+t\de\psi\big)  \, dt
         \bigg|_{\atop{\hat\psi_{(*)}=\hat\psi_{0(*)}(\psi_*,\psi,\mu_0,\cV_0)}
                  {\de\psi_{(*)}=L^{3/2}\bbbs D^{(0)(*)}\bbbs^{-1}z_{(*)}}}\\
&\hskip1in -\mu_0 L^5\big<z_*,\,
                       \bbbs C^{(0)}\bbbs^{-1}z\big>_1
\end{split}
\tag{\ref{eqnOSAdeAndef}.b}
\end{equation}
where, for $n\ge 1$,
\begin{equation}\label{eqnOSAhatphiplus}
\de\hat\phi_{(*)n+1}^{(+)}(\psi_*,\psi,z_*,z)
=\de\hat\phi_{(*)n+1}(\psi_*,\psi,z_*,z)
    -L^{3/2}\bbbs S_n^{(*)}Q_n^*\fQ_n D^{(n)(*)}\bbbs^{-1}z_{(*)}
\end{equation}
\end{itemize}
\end{itemize}

\begin{corollary}\label{corSTmainCor}
\ 
\begin{enumerate}[label=(\alph*), leftmargin=*]
\item 
Let $1\le n<\np$. Let $\mu_n$, $\cR_n$ and $\cE_n$ be as in 
Theorem \ref{thmTHmaintheorem}. Set
\begin{align*}
\cA_n(\psi_*,\psi)&=\Big[- A_n(\psi_*,\psi, \phi_{*n}, \phi_n,\,\mu_n,\cV_n)
+\cR_n(\phi_{*n},\phi_n)\Big]_{\phi_{(*)n}=\phi_{(*)n}(\psi_*,\psi,\mu_n,\cV_n)}
\\&\hskip0.5in
+\cE_n(\psi_*,\psi)
\end{align*}
Then
\begin{align*}
&\big(\bbbs \bbbt_n^{(SF)} \big)
\Big(e^{\cA_n}\Big) (\psi_*,\psi;\mu_n,\cV_n)
= \sfrac{1}{\tilde N^{(n)}_\bbbt}\ e^{\cC_n(\psi_*,\psi)}
    \ \cF_n(\psi_*,\psi)
\end{align*}

\item 
For $n=0$,\ \ \ 
$
\big(\bbbs\bbbt_0^{(SF)}\big) 
\Big(e^{\cA_0}\Big) (\psi_*,\psi;\mu_0,\cV_0)
= \sfrac{1}{\tilde N^{(0)}_\bbbt}\ e^{\cC_0(\psi_*,\psi)}
    \ \cF_0(\psi_*,\psi)
$
\end{enumerate}
\end{corollary}
\begin{proof} By Proposition \ref{propSTmainProp}, it suffices to verify that
\begin{equation*}
(\bbbs\check\cC_n)(\psi_*,\psi)=\cC_n(\psi_*,\psi)\qquad
(\bbbs\check\cF_n)(\psi_*,\psi)=\cF_n(\psi_*,\psi)
\end{equation*}
for all $n\ge 0$. That we have $\bbbs\check\cC_n=\cC_n$ follows 
immediately from
\begin{align*}
&\check A_{n+1}(\th_*,\th,\check\phi_{*n+1}(\th_*,\th,\mu_n,\cV_n),
 \check\phi_{n+1}(\th_*,\th,\mu_n,\cV_n),\mu_n,\cV_n)
                           \Big|_{\th_{(*)}=\bbbs^{-1}\psi_{(*)}}
\\&\hskip0.5in
 =A_{n+1}(\psi_*,\psi,\phi_{*n+1}(\psi_*,\psi,L^2\mu_n,\bbbs\cV_n),
                       \phi_{n+1}(\psi_*,\psi,L^2\mu_n,\bbbs\cV_n),
                       L^2\mu_n,\bbbs\cV_n)
\end{align*}
by Definition \ref{defSCacheck} and Definition \ref{defBGAphicheck},
 and
\begin{align*}
&\cR_n\big( \check\phi_{*n+1}(\th_*,\th,\mu_n,\cV_n),
  \check\phi_{n+1}(\th_*,\th,\mu_n,\cV_n)\big)
                          \Big|_{\th_{(*)}=\bbbs^{-1}\psi_{(*)}}
\\&\hskip0.5in
=(\bbbs\cR_n)\big(\,\bbbs[\check\phi_{*n+1}(\th_*,\th,\mu_n,\cV_n)]\,,\, \bbbs[\check\phi_{n+1}(\th_*,\th,\mu_n,\cV_n)]\,\big)
                           \Big|_{\th_{(*)}=\bbbs^{-1}\psi_{(*)}}
\\&\hskip0.5in
=(\bbbs\cR_n)\big(\, \phi_{*n+1}(\psi_*,\psi,L^2\mu_n,\bbbs\cV_n)\,,\,
                       \phi_{n+1}(\psi_*,\psi,L^2\mu_n,\bbbs\cV_n)\, \big)
\end{align*}
and
\begin{align*}
&\cE_n\big( \psi_{*n}(\th_*,\th,\mu_n,\cV_n),\psi_n(\th_*,\th,\mu_n,\cV_n)\big)
         \Big|_{\th_{(*)}=\bbbs^{-1}\psi_{(*)}} 
\\&\hskip0.5in
=(\bbbs\cE_n)\big( \hat\psi_{*n}(\psi_*,\psi,\mu_n,\cV_n),
                   \hat\psi_n(\psi_*,\psi,\mu_n,\cV_n)\big)
\end{align*}
by \eqref{eqnSThatpsi}. 

Under the substitution $\ze(x)=(\bbbl_* z)(x)= z(\bbbl^{-1} x)$
\begin{equation}\label{eqnSTintuncheck}
\begin{split}
&\Big[\hskip-3pt
 \prod_{x\in\cX_0^{(n)}}  \int\limits_{|\ze(x)|\le r_n} \hskip -9pt
\sfrac{d\ze(x)^*\wedge d\ze(x)}{2\pi i} e^{-|\ze(x)|^2}\Big]\cG(\ze_*,\ze)
\\&\hskip1in
=\Big[\hskip-3pt
 \prod_{w\in\cX_1^{(n)}}  \int\limits_{|z(w)|\le r_n} \hskip -9pt
\sfrac{dz(w)^*\wedge dz(w)}{2\pi i} e^{-|z(w)|^2}\Big]\cG(\bbbl_*z_*,\bbbl_*z)
\end{split}
\end{equation}
By \eqref{eqnSThatpsi} and Remark \ref{remSCscaling}.a,
\begin{equation}\label{eqnSTdeEuncheck}
\begin{split}
\de\check\cE_n(\th_*,\th,\de\psi_*,\de\psi)
   \Big|_{\atop{\th_{(*)}=\bbbs^{-1}\psi_{(*)}}
               {\de\psi_{(*)}= D^{(n)(*)}\bbbl_*z_{(*)}}}
&=\cE_n(\hat\psi_*,\hat\psi)\Big|
       ^{\hat\psi_{(*)}=\bbbs^{-1}\hat\psi_{(*)n}(\psi_*,\psi,\mu_n,\cV_n)
                        + D^{(n)(*)}\bbbl_* z_{(*)}}
       _{\hat\psi_{(*)}=\bbbs^{-1}\hat\psi_{(*)n}(\psi_*,\psi,\mu_n,\cV_n)}\\
&\hskip-0.5in=(\bbbs \cE_n)(\Psi_*,\Psi)\Big|
       ^{\Psi_{(*)}=\hat\psi_{(*)n}(\psi_*,\psi,\mu_n,\cV_n)
                        +L^{3/2}\bbbs D^{(n)(*)}\bbbs^{-1}z_{(*)}}
       _{\Psi_{(*)}=\hat\psi_{(*)n}(\psi_*,\psi,\mu_n,\cV_n)}
\end{split}
\end{equation}
and, by Definition \ref{defBGAphicheck} and \eqref{eqnSTdehatphidef}
\begin{equation}\label{eqnSTdeRuncheck}
\begin{split}
&\de\check\cR_n(\th_*,\th,\de\psi_*,\de\psi)
   \Big|_{\atop{\th_{(*)}=\bbbs^{-1}\psi_{(*)}}
               {\de\psi_{(*)}= D^{(n)(*)}\bbbl_*z_{(*)}}}\\
&\hskip0.7in=\cR_n(\phi_*,\phi)\Big|
       ^{\phi_{(*)}=\bbbs^{-1}[\phi_{(*)n+1}(\psi_*,\psi,L^2\mu_n,\bbbs\cV_n)]
                        + \bbbs^{-1}[\de\hat\phi_{(*)n+1}(\psi_*,\psi,z_*,z)]}
    _{\phi_{(*)}=\bbbs^{-1}[\phi_{(*)n+1}(\psi_*,\psi,L^2\mu_n,\bbbs\cV_n)]}\\
&\hskip0.7in=(\bbbs\cR_n)(\Phi_*,\Phi)\Big|
       ^{\Phi_{(*)}=\phi_{(*)n+1}(\psi_*,\psi,L^2\mu_n,\bbbs\cV_n)
                     +\de\hat\phi_{(*)n+1}(\psi_*,\psi,z_*,z)}
       _{\Phi_{(*)}=\phi_{(*)n+1}(\psi_*,\psi,L^2\mu_n,\bbbs\cV_n)}
\end{split}
\end{equation}
Since, for $n\ge 1$,
\begin{align*}
&\de{\check\phi_{n+1}}^{(+)}\big(\th_*,\th;\de\psi_*,\de\psi,\mu_n,\cV_n\big)
      \Big|_{\atop{\th_{(*)}=\bbbs^{-1}\psi_{(*)}}
                  {\de\psi_{(*)}= D^{(n)(*)}\bbbl_*z_{(*)}}}\\
&\hskip0.5in
 =\de{\check\phi_{n+1}}\big(\th_*,\th;\de\psi_*,\de\psi,\mu_n,\cV_n\big)
      \Big|_{\atop{\th_{(*)}=\bbbs^{-1}\psi_{(*)}}
                  {\de\psi_{(*)}= D^{(n)(*)}\bbbl_*z_{(*)}}}
   -S_n Q^*_n\fQ_n D^{(n)}\bbbl_*z\\
&\hskip0.5in=\bbbs^{-1}[\de\hat\phi_{n+1}(\psi_*,\psi,z_*,z)]
   -L^{3/2}S_n Q^*_n \fQ_n D^{(n)}\bbbs^{-1}z\\
&\hskip0.5in=\bbbs^{-1}[\de\hat\phi_{n+1}^{(+)}(\psi_*,\psi,z_*,z)]
\end{align*}
by Definition \ref{defBGAbckgndVarn}.b, \eqref{eqnSTdehatphidef}
 and Remark \ref{remSCscaling}.a, we have
\begin{align*}
&\big< \de\psi_*,\fQ_n\, Q_n\,
\de{\check\phi}^{(+)}_{n+1}\big(\th_*,\th;\de\psi_*,\de\psi,\mu_n,
                              \cV_n\big) \big>_0
    \Big|_{\atop{\th_{(*)}=\bbbs^{-1}\psi_{(*)}}
                {\de\psi_{(*)}= D^{(n)(*)}\bbbl_*z_{(*)}}}\\
&\hskip0.5in=
L^2\big< \bbbs D^{(n)*}\bbbl_*z_*\,,\,\bbbs\fQ_n\, Q_n\,
\bbbs^{-1}[\de\hat\phi_{n+1}^{(+)}(\psi_*,\psi,z_*,z)] \big>_1\\
&\hskip0.5in=
L^{7/2}\big< \bbbs D^{(n)*}\bbbs^{-1}z_*\,,\,\bbbs\fQ_n\, Q_n\,
\bbbs^{-1}[\de\hat\phi_{n+1}^{(+)}(\psi_*,\psi,z_*,z)] \big>_1\\
&\hskip0.5in=
L^{7/2}\big< z_*\,,\,\bbbs D^{(n)}\fQ_n\, Q_n\,
\bbbs^{-1}[\de\hat\phi_{n+1}^{(+)}(\psi_*,\psi,z_*,z)] \big>_1
\end{align*}
by Remark \ref{remSCscaling}.a,c,d. Consequently, for $n\ge 1$,
\begin{equation}\label{eqnSTdeAuncheck}
\de \check A_n(\th_*,\th,\de\psi_*,\de\psi)
\Big|_{\atop{\th_{(*)}=\bbbs^{-1}\psi_{(*)}}
            {\de\psi_{(*)}= D^{(n)(*)}\bbbl_*z_{(*)}}}
=\de  A_n(\psi_*,\psi,z_*,z) 
\end{equation}
For $n=0$,
\begin{align*}
&\cV_0\big(\psi_{0*}(\th_*,\th,\mu_0,\cV_0)+\de\psi_{*}\,,\,
               \psi_0(\th_*,\th,\mu_0,\cV_0)+\de\psi\big)
\Big|_{\atop{\th_{(*)}=\bbbs^{-1}\psi_{(*)}}
            {\de\psi_{(*)}= D^{(0)(*)}\bbbl_*z_{(*)}}}\\
&\hskip0.1in=
(\bbbs\cV_0)\big(\hat\psi_{0*}(\psi_*,\psi,\mu_0,\cV_0)
          +L^{3/2}\bbbs D^{(0)*}\bbbs^{-1}z_*\,,\,
               \hat\psi_0(\psi_*,\psi,\mu_0,\cV_0)
            +L^{3/2}\bbbs D^{(0)}\bbbs^{-1}z\big)
\end{align*}
and
\begin{align*}
\< \de\psi_*,\,\de\psi\>_0\Big|_{
          \de\psi_{(*)}= D^{(0)(*)}\bbbl_*z_{(*)}}
%&\hskip0.1in=\big< D^{(0)*}\bbbl_*z_*,\,D^{(0)}\bbbl_*z\big>_0\\
&=L^2\big<\bbbs D^{(0)*}\bbbl_*z_*,\,\bbbs D^{(0)}\bbbl_*z\big>_1\\
&=L^5\big<\bbbs D^{(0)*}\bbbs^{-1}z_*,\,
          \bbbs D^{(0)}\bbbs^{-1}z\big>_1\\
&=L^5\big<z_*,\,
          \bbbs {C^{(0)}}\bbbs^{-1}z\big>_1
\end{align*}
by \eqref{eqnSThatpsi} and Remark \ref{remSCscaling}.a,c,d.
Therefore \eqref{eqnSTdeAuncheck} also holds for $n=0$. 
That $\bbbs\check\cF_n=\cF_n$ now follows from \eqref{eqnSTintuncheck},
\eqref{eqnSTdeEuncheck}, \eqref{eqnSTdeRuncheck} and \eqref{eqnSTdeAuncheck}.
\end{proof}

\newpage
\appendix
%&&&&&&&&&&&&&&&&&&&&&&&&&&&&&&&&&
\section{Compendium of Definitions}\label{appDefinitions}
%&&&&&&&&&&&&&&&&&&&&&&&&&&&&&&&&&

%%%%%%%%%%%%%%%%
\subsection{Lattices}\label{appDEFlattices}
%%%%%%%%%%%%%%%

We use many different lattices. Our initial system is a finite 
volume (continuous) spin system having one (complex valued) spin 
at each site of the lattice\footnote{Of course $\cX_0$ is a finite set 
and so is perhaps more accurately described as a discrete torus, 
rather than a lattice.}
\begin{equation*}
\cX_0=\big(\bbbz/L_\tp\bbbz\big)\times
         \big(\bbbz^3/L_\sp\bbbz^3\big) 
\end{equation*}
where $L_\tp\in L^2\bbbn$ and $L_\sp\in L\bbbn$ are the temporal and 
spatial sizes of this initial, finite volume, lattice
and $L\ge 3$ is a fixed odd natural number. $\cX_0$ is a unit lattice 
in the sense that the distance between nearest neighbours in the 
lattice is $1$. 
During each renormalization group step this lattice
is scaled down. In each of the first $\np$ steps, which are the
steps considered in this paper and in \cite{ParOv,PAR2}, we use 
(anisotropic) ``parabolic scaling'' and decrease
the lattice spacing in the temporal direction 
by a factor of $L^2$ and in the spatial directions by a factor of $L$. 
So after $n$ renormalization group steps the lattice spacing 
in the spatial directions is $\veps_n=\frac{1}{L^n}$ and in the 
temporal direction is  $\veps_n^2=\frac{1}{L^{2n}}$
and the torus $\cX_0$  has been scaled down to
\begin{equation*}
\cX_n= \big(\sfrac{1}{L^{2n}}\bbbz \big/ \sfrac{L_\tp}{L^{2n}}\bbbz\big)
   \times \big(\sfrac{1}{L^n}\bbbz^3 \big/\sfrac{L_\sp}{L^n}\bbbz^3 \big)
\end{equation*}
We call $\cX_n$ the ``$\veps_n$--lattice'' and denote by
\begin{equation*}
\cH_n =\bbbc^{\cX_n}
\end{equation*}
the space of all complex valued functions on $\cX_n$. We endow
$\cH_n$ with the norm and  bilinear\footnote{Note that the form is not
sesquilinear. We will explicitly write complex conjugates when we want
them.} form
\begin{align*}
\|f\|_n&=\vol_n\sum_{x\in\cX_n} |f(x)|^2 \quad&
\<f,g\>_n&=\vol_n\sum_{x\in\cX_n} f(x)g(x)
\end{align*}
where
\begin{equation*}
\vol_n=\veps_n^5
\end{equation*}
is the volume of a  cell in $\cX_n$.
We view $\cH_n$ as the Hilbert space $L^2(\cX_n)$ with (positive
definite) inner product $\<f^*,g\>_n$ and norm  $\|f\|_n$. 
Many of the operators acting on $\cH_n$ that we consider are
periodizations  of operators acting on $L^2$ of the ``universal cover''
\begin{equation*}
\cZ_n = \veps^2_n\bbbz\times \veps_n \bbbz^3
\end{equation*}
of $\cX_n$.

We implement each renormalization group step by performing a block 
averaging of the spins. During each of the steps in the
parabolic regime, we average\footnote{Technically,
this might be modified by smoothing.} over blocks of 
$L^2\times L\times L\times L$ sites. 
So a $k$--block, that is a block of sites averaged over in each
of the first $k$ renormalization group steps, consists of
$L^{2k}\times L^k\times L^k\times L^k$ sites when $k\le\np$.
The number of sites averaged over is $\sfrac{1}{\vol_k}$.

It is necessary to repeatedly compose critical point field
configurations. For this purpose, we introduce an
array of intermediate sublattices. For each $0\le k\le n+1$, define 
the sublattice $\cX_{n-k}^{(k)}$ of (centres of) $k$--blocks in 
the $\veps_n$--lattice $\cX_n$ and the corresponding Hilbert space, 
bilinear form and norm, to be
\begin{align*}
\cX_{n-k}^{(k)}&=
     \big(\veps^2_{n-k}\bbbz/\veps_n^2L_\tp\bbbz\big)
     \times\big(\veps_{n-k}\bbbz^3/\veps_nL_\sp\bbbz^3\big)\quad
\cH_{n-k}^{(k)}=L^2(\cX_{n-k}^{(k)})
\end{align*}
and
\begin{align*}
\<f,g\>^{(k)}_{n-k}&=\sfrac{\vol_n}{\vol_k}
            \sum_{x\in\cX^{(k)}_{n-k}}\!\!\! f(x)g(x) \quad&
\|f\|^{(k)}_{n-k}&=\sfrac{\vol_n}{\vol_k}
                 \sum_{x\in\cX^{(k)}_{n-k}}\!\! |f(x)|^2
\end{align*}
The lower index gives the ``scale'' of the lattice. 
That is, the distance between nearest neighbour points of the lattice. 
The upper index gives the block size and determines the number of points 
in the sublattice (the number of points in $\cX_n$ divided by the
number of points in a $k$--block). The sum of the upper and 
lower indices gives the number of the renormalization group step.
For example, $\cX^{(k)}_j$ 
\begin{itemize}[leftmargin=*, topsep=2pt, itemsep=0pt, parsep=0pt]
\item
has the lattice spacing $\veps_j^2$ in temporal directions and 
\item
has the lattice spacing $\veps_j$ in spatial directions and 
\item
has $(\veps^2_kL_\tp)(\veps_kL_\sp)^3$ points and
\item 
has volume $\ [\text{volume of single cell}]\times[\hbox{number of points}]
= \veps_j^5 (\veps^2_kL_\tp)(\veps_kL_\sp)^3$
\end{itemize}
Observe that $\cX_n^{(0)}=\cX_n$ and $\cH_n^{(0)}=\cH_n$ and that
$\cX_0^{(n)}=\big(\bbbz/\veps_n^2L_\tp\bbbz\big)\times
         \big(\bbbz^3/\veps_nL_\sp\bbbz^3\big)$ 
is a unit lattice in $\bbbz^4$.  

%%%%%%%%%%%%%%%%
\subsection{Scaling}\label{appDEFscaling}
%%%%%%%%%%%%%%%

Scaling is performed by the linear isomorphisms
\begin{equation*}
\bbbl : \cX_j^{(k)}\rightarrow \cX_{j-1}^{(k)}
    \qquad  \qquad
  (u_0,\bu) \mapsto (L^2u_0, L\bu)
\end{equation*}
 For a function $\,\al\in\cH_j^{(k)}\,$,
define  the function $\,\bbbl_*(\al) \in \cH_{j-1}^{(k)}\,$  by 
$\,
\bbbl_*(\al)(\bbbl u) = \al(u)
\,$. Set
$
\bbbs =L^{3/2}\bbbl_*^{-1}
$.
That is, for a field $\,\th\,$ on $\cX_{j-1}^{(k)} $,  
\begin{equation*}
(\bbbs\th)(x)= L^{3/2}\,\th\big(\bbbl x\big)
\end{equation*}
is a field on $\,\cX_j^{(k)}\,$.

%%%%%%%%%%%%%%%%
\subsection{Block Spin Operators}\label{appDEFblockspinops}
%%%%%%%%%%%%%%%

The block spin averaging operator 
$Q:\cH_0^{(n)}\rightarrow  \cH_{-1}^{(n+1)}$, which averages
fields $\psi(x)$, indexed by points $x$ of the unit lattice $\cX_0^{(n)}$,
over blocks centered on the points of the $L$--lattice $\cX_{-1}^{(n+1)}$ 
is defined by
\begin{equation}\label{eqnDEFaveop}
(Q\psi)(y) = \smsum_{x\in \bbbz\times\bbbz^3} q(x) \psi(y+[x])
\end{equation}
where $\,[x]\,$ denotes the class of $\,x\in\bbbz\times\bbbz^3\,$ in the
quotient space $\,\cX_0^{(n)}\,$.
The averaging profile $q$ is the $\fq$--fold convolution of 
the characteristic function, $1_{\sq}(x)$, of the rectangle 
$\big[-\sfrac{L^2-1}{2}, \sfrac{L^2-1}{2} \big]
\times \big[-\sfrac{L-1}{2}, \sfrac{L-1}{2} \big]^3$, normalized
to have integral one. That is,
\begin{equation*}
q=\sfrac{1}{L^{5\fq}} 
\overbrace{ 1_{\sq}*1_{\sq}*\cdots*1_{\sq} }^{\fq\ {\rm times}}
\end{equation*}
For bounds on $Q$, see \cite[Lemma \lemPBSuplusppties]{POA}.

The block spin averaging operator 
$Q_n:\cH_n\rightarrow  \cH_0^{(n)}$, which averages
fields $\phi(u)$, indexed by points $u$ of the fine lattice $\cX_n$,
over blocks centered on the points of the unit--lattice $\cX_0^{(n)}$ 
is defined by
\begin{equation}\label{eqnDEFqn}
Q_n = \big(\bbbl_*^{-1}Q\big)^n \bbbl_*^n
        % :  \cH_n=\cH_n^{(0)} \rightarrow \cH_0^{(n)} 
\end{equation}
For bounds on $Q_n$, see \cite[Remark \remPBSqnft.a and 
Lemma \lemPBSunppties]{POA}.

The operator
\begin{equation*}
\fQ_n=a\Big(\bbbone
             +\sum_{j=1}^{n-1}\sfrac{1}{L^{2j}}Q_jQ_j^*\Big)^{-1} 
\end{equation*}
appears in the term
$\< \psi^*-Q_n\,\phi_*\,,\, \fQ_n(\psi-Q_n\,\phi) \>_0 $
of the dominant part of the action.
For bounds on $\fQ_n$, see \cite[Remark \remPBSqnft.c and 
Proposition \propPBSAnppties]{POA}.
See \cite[Remark \remBSactionRecur]{BlockSpin} for the
recursion relation that builds $\fQ_n$.

%%%%%%%%%%%%%%%%
\subsection{Differential and Related Operators}\label{appDEFoperators}
%%%%%%%%%%%%%%%

The forward derivatives of $\al\in \cH_j^{(n)}$  are defined by
\begin{equation}\label{eqnDEFforwardDeriv}
(\partial_\nu \al)(x) =  
\sfrac{1}{\veps_{j,\nu}} \big[\al(x+ \veps_{j,\nu} e_\nu)-\al(x)\big]
\end{equation}
where $e_\nu$ is a unit vector in the $\nu^{\mathrm{th}}$ direction
and
\begin{equation*}
\veps_{j,\nu}=\begin{cases}
    \veps_j^2=\sfrac{1}{L^{2j}} & \text{for $\nu=0$}\\
    \veps_j=\sfrac{1}{L^j} & \text{for $\nu=1,2,3$} 
     \end{cases}
\end{equation*}
We associate to an operator $h_0$ on $L^2\big(\bbbz^3/L_\sp\bbbz^3\big)$ 
the operators
\begin{equation}\label{eqnDEFdndef}
D_n = L^{2n}\ \bbbl_*^{-n}
         \big(\bbbone - e^{-\oh_0} -e^{-\oh_0} \partial_0\big)\bbbl_*^n
\end{equation}
Here $\partial_0$ is the forward time derivative of \eqref{eqnDEFforwardDeriv}.
We assume that $h_0$ is the periodization of a translation
invariant operator $\bh_0$ on $L^2\big(\bbbz^3\big)$ whose Fourier 
transform $\hat\bh_0(\bp)$ 
\begin{itemize}[leftmargin=*, topsep=2pt, itemsep=0pt, parsep=0pt]
%\item
% is periodic with respect to $2\pi\bbbz^3$ 
\item
 is entire in $\bp$ and invariant under 
$\bp_\nu\rightarrow-\bp_\nu$ for each $1\le\nu\le3$
\item
 is nonnegative when $\bp$ is real and is strictly positive
 when $\bp\in\bbbr^3\setminus{2\pi\bbbz^3}$
\item
 obeys $\hat\bh_0(\bZ)
     =\sfrac{\partial\,\hat\bh_0\ }{\partial\bp_\nu}(\bZ)=0$
 for $1\le\nu\le3$ and has strictly positive Jacobian matrix
 $\Big[ 
   \sfrac{\partial^2\,\hat\bh_0\ }{\partial\bp_\mu\partial\bp_\nu}(\bZ)   
  \Big]_{1\le\mu,\nu\le3}$. 
\end{itemize}
Think of $\bh_0$, which is (a constant times) the single particle 
``kinetic energy'' operator, as being essentially a positive
constant times the discrete spatial laplacian. The operator 
$D_n$ is studied in \cite[\S\secPOdiffOps]{POA}.

A number of important operators are built from $D_n$. 
One is the covariance for the fluctuation integral in \cite{PAR2}.
It is
\begin{equation*}
C^{(n)}=(\sfrac{a}{L^2}Q^*Q+\De^{(n)})^{-1} 
\end{equation*}
where
\begin{equation*}
\De^{(n)}=\left.
   \begin{cases}
      \big(\bbbone+\fQ_n Q_n D_n^{-1} Q_n^*\big)^{-1}\fQ_n 
                 & \text{if $n\ge 1$}\\
      \noalign{\vskip0.05in}
      D_0       & \text{if $n=0$}
    \end{cases}\right\}
    :\cH_0^{(n)}\rightarrow\cH_0^{(n)}
\end{equation*}
It is bounded in \cite[Corollary \corPOCsquareroot]{POA}.

Another family of important operators built from $D_n$ are the
Green's functions
\begin{equation*}
S_n(\mu)=\big[D_n+Q^*_n\fQ_nQ_n-\mu\big]^{-1}
\end{equation*}
They are bounded in \cite[Proposition \POGmainpos]{POA}.

%%%%%%%%%%%%%%%%
\subsection{Norms}\label{appDEFnorms}
%%%%%%%%%%%%%%%

Let $\cX$ be any lattice that is equipped with a metric $d$ and a 
``cell volume'' $\vol$. As an example, the lattice $\cX_j^{(n-j)}$ has $\vol=\sfrac{1}{L^{5j}}$.
The following Definition describes how we measure the size of the kernels 
whose arguments run over $\cX$.

\begin{definition}\label{defDEFkernelnorm}
Let $\,f(u_1,\cdots,u_r)\,$ be a function on $\,\cX^r\,$. For a mass 
$\,\fm\ge 0\,$ we set
\begin{equation*}
\|f\|_\fm
=\max_{i=1\cdots,r} \max_{u_{i}}
 \int du_1 \cdots du_{i-1}\,du_{i+1} \cdots du_r\ |f(u_1,\cdots,u_r)|\, e^{\fm\tau(u_1,\cdots,u_r)}
\end{equation*}
where \emph{the tree length} $\,\tau(u_1,\cdots,u_r)\,$ is the minimal length 
of a tree in $\,\cX\,$ that has  $\,u_1,\cdots,u_r\,$ among its vertices, and
$\int du\ g(u)=\vol\sum_{u\in\cX} g(u)$.
\end{definition}

The following definitions describe how we measure the size of complex valued
analytic functions of fields. 
The following norms  are special cases of the norms
in \cite[Definition 2.6]{CPC}.

\pagebreak[2]
\begin{definition}\label{defDEFabstractnorm}
\ 
\begin{enumerate}[label=(\alph*), leftmargin=*]
\item 
For a field $\,\al\,$ on $\,\cX\,$ and 
$\,\vec x = (x_1,\cdots,x_r) \in \cX^r\,$ we set
$\ \al(\vec x) = \smprod_{i=1}^r \al(x_i)\,$.

\item 
A power series  $\,\cF\,$  in the fields $\,\al_1,\cdots,\al_s$, 
on $\cX\,$ has a unique expansion
\begin{equation*}
\cF(\al_1,\cdots,\al_s)
      = \sum_{r_1,\cdots,r_s\ge 0}\vol^{r_1+\cdots+r_s}\hskip-5pt  
         \sum_{\atop{\vec x_i\in\cX^{r_i}}{ 1\le i\le s}}
            f_{r_1,\cdots,r_s}\big(\vec x_1,\cdots,\vec x_s\big)\,
            \smprod_{i=1}^s \al_i(\vec x_i)
\end{equation*}
where the coefficients 
$f_{r_1,\cdots,r_s}\big(\vec x_1,\cdots,\vec x_s\big)$ are invariant
under permutations of the components of each vector $\vec x_i$.
\item 
For each choice of ``weights'' $\ka_1,\cdots,\ka_s>0$,
for the fields $\al_1,\cdots,\al_s$, we define \emph{the norm of $\cF$ with
mass $\fm$ and weights $\ka_1,\cdots,\ka_s>0$} to be
\begin{equation*}
 \sum_{r_1,\cdots,r_s\ge 0}
  \big\| f_{r_1,\cdots,r_s}\big(\vec x_1,\cdots,\vec x_s\big)\big\|_\fm\ \
  \smprod_{i=1}^s \ka_i^{r_i}
\end{equation*}
\end{enumerate}
\end{definition}

The following definition describes how we measure the size of analytic
maps like the background field map $(\psi_*,\psi)\mapsto
\phi_n(\psi_*,\psi,\mu_n,\cV_n)$. 
The norms in the following definition are special cases of the norms
in \cite[Definition \defSUBkrnel]{SUB}.

\begin{definition}\label{defDEFkrnel}
Let $\cX$ and $\cY$ be sublattices of a common finite lattice having metric $d$,
with $\cX$ having a ``cell volume'' $\vol$ and
with $\cY$ having a ``cell volume'' $\vol_\cY$. 
Write\footnote{If 
     $(\vec x_1,\cdots,\vec x_{s-1})\in \cX^{r_1}\times\cdots\times\cX^{r_{s-1}}$ 
     then $(\vec x_1,\cdots,\vec x_{s-1},-)$ denotes the corresponding element of 
     $\cX^{r_1}\times\cdots\times\cX^{r_{s-1}}\times\cX^0$
     with ``no $s^{\rm th}$ entry''.
     In particular, $\cX^0=\{-\}$ and $\al(-)=1$.} 
$\cX^{(s)}=\bigcup_{r_1,\cdots,r_s\ge 0}\cX^{r_1}\times\cdots\times \cX^{r_s}$.
%%%%%%%%%%%%%%%%%%%%%%%%%
\begin{enumerate}[label=(\alph*), leftmargin=*]
\item
An $s$--field map kernel is a function
\begin{equation*}
A:(y;\vec x_1,\cdots,\vec x_s)\in \cY\times\cX^{(s)}\mapsto
A(y;\vec x_1,\cdots,\vec x_s)\in\bbbc
\end{equation*}
which obeys $A(y;-,\cdots,-)=0$ for all $y\in \cY$.
%%%%%%%%%%%%%%%%%%%%%%%%%%%
\item 
If $A$ is an $s$--field map kernel, we define the
``field map'' 
\begin{equation*}
(\al_1,\cdots,\al_s)\rightarrow A(\al_1,\cdots,\al_s)
\end{equation*}
by
\begin{equation}
A(\al_1,\cdots,\al_s)(y)=
\sum_{r_1,\cdots,r_s\ge 0}\vol^{r_1+\cdots+r_s}\hskip-5pt  
         \sum_{\atop{\vec x_i\in\cX^{r_i}}{1\le i\le s}}
A(y;\vec x_1,\cdots,\vec x_s)\ \al_1(\vec x_1)\cdots \al_s(\vec x_s)
\end{equation}

\item 
We define the norm $\tn A\tn$, 
with mass $\fm$ and weight factors $\ka_1$, $\cdots$, $\ka_s$, 
of the $s$--field map $A$ by
\begin{equation*}
\tn A\tn=
    \sum_{ \atop{r_1,\cdots,r_s\ge 0}{r_1+\cdots+r_s\ge 1}}
    \big\|A\big\|_{r_1,\cdots,r_s}
%\max\big\{L(A)\,,\, R(A)\big\}
\end{equation*}
where
\begin{align*}
\big\|A\big\|_{r_1,\cdots,r_s}
  =\max\big\{ L(A;r_1,\cdots,r_s)\,,\,R(A;r_1,\cdots,r_s)\big\}
\end{align*}
and
\begin{align*}
L(A;r_1,\cdots,r_s)&=\max_{y\in \cY}\ \vol^{r_1+\cdots+r_s}\hskip-5pt
     \sum_{ \atop{\vec x_\ell\in \cX^{r_\ell}}{1\le\ell\le s}}
     \big|A(y;\vec x_1,\cdots,\vec x_s)\big|
     \ka_1^{r_1}\cdots\ka_s^{r_s}
      e^{\fm\tau(y,\vec x_1,\cdots,\vec x_s)}\\
R(A;r_1,\cdots,r_s)&=\max_{x'\in \cX}
    \max_{ \atop{\atop{1\le j\le s}{r_j\ne 0}}
                {1\le i\le r_j}}
\vol_\cY\sum_{y\in\cY}
\vol^{r_1+\cdots+r_s-1}\hskip-5pt
\sum_{ \atop{ \atop{\vec x_\ell\in \cX^{r_\ell}}{1\le\ell\le s} }
            {{(\vec x_j)}_i= x'}}\hskip-8pt
\big|A(y;\vec x_1,\cdots,\vec x_s)\big|\\
\noalign{\vskip-0.40in}&\hskip3.2in
\ka_1^{r_1}\cdots\ka_s^{r_s}
e^{\fm\tau(y,\vec x_1,\cdots,\vec x_s)}
\end{align*}
\end{enumerate}
\end{definition}

\begin{remark}\label{remDEFkrnel}
Denote by $\bbbc^\cX$ and $\bbbc^\cY$ the spaces of fields on 
$\cX$ and $\cY$, respectively.
If $A$ is an $s$--field map kernel whose norm, $\tn A\tn$,
is finite, then $(\al_1,\cdots,\al_s)\mapsto A(\al_1,\cdots,\al_s)$
is an analytic map from the polydisc 
\begin{equation*}
\set{(\al_1,\cdots,\al_s)\in \bbbc^\cX\times\cdots\times\bbbc^\cX}
             {\|\al_j\|_{L^\infty}<\ka_j,\ 1\le j\le s}
\end{equation*}
to the polydisc 
\begin{equation*}
\set{\be\in \bbbc^\cY} {\|\be\|_{L^\infty}<\tn A\tn}
\end{equation*}

\end{remark}

Most operators we deal with are bounded with respect to a 
norm of the following kind.

\begin{definition}\label{defDEFoperatornorm}
Let $\cX$ and $\cY$ be sublattices of a common lattice having metric $d$,
with $\cX$ having a ``cell volume'' $\vol_\cX$ and
with $\cY$ having a ``cell volume'' $\vol_\cY$. 
For any operator $A:\bbbc^\cX\rightarrow\bbbc^\cY$, with kernel
$A(y,x)$, and for any mass $m\ge 0$, we define the norm
\begin{equation*}
\|A\|_m
=\max\Big\{\sup_{y\in\cY}\,\sum_{x\in\cX} \vol_{\cX}\ 
                  e^{m|y-x|}|A(y,x)|\ ,\ 
\sup_{x\in\cX}\,\sum_{y\in\cY} \vol_{\cY}\ 
                  e^{m|y-x|}|A(y,x)|
\Big\}
\end{equation*}
In the special case that $m=0$, this is just the usual $\ell^1$--$\ell^\infty$
norm of the kernel.
\end{definition}

%&&&&&&&&&&&&&&&&&&&&&&&&&&&&&&&&&
\section{Symmetries}\label{appSYsymmetry}
%&&&&&&&&&&&&&&&&&&&&&&&&&&&&&&&&&

Fix any integers $j\ge 0$ and $n\ge j$. We discuss the natural symmetries
of the lattice $\cX_j^{(n-j)}$ (see Definition \ref{defHTbackgrounddomaction}.a)
and the corresponding symmetries induced on $\cH_j^{(n-j)}$ and on 
functions on $\cH_j^{(n-j)}$. Define $\veps_j=\sfrac{1}{L^j}$.

\begin{definition}\label{defSYfullSymmetry}
\ 
\begin{enumerate}[label=(\alph*), leftmargin=*]
\item
We define (unit) translation and reflection operators,
acting on the $\veps_j$--lattice $\cX_j^{(n-j)}$, by
\begin{equation*}
T_x u = u+x\qquad
\big(R_{\nu}u\big)_i=\begin{cases}
                            u_i &\text{if $i\ne\nu$}\\
                           -u_i &\text{if $i=\nu$}
                     \end{cases}
\end{equation*}
for all $x\in\cX_0^{(n)}$, $u\in\cX_j^{(n-j)}$ and $0\le\nu\le3$.
\item
We next define translation operators, acting on the field 
$\al:\cX_j^{(n-j)}\rightarrow\bbbc$, by
\begin{equation*}
\big(T_x\al\big)(u)=\al\big(T_{-x}u\big)=\al(u-x)
\end{equation*}
and reflection operators, acting on the fields $\al_{(*)}$ and $\al_{\nu(*)}$,
by
\begin{align*}
&\big(R_{\nu'}\al_{(*)}\big)(u)=\al_{(*)}(R_{\nu'}u)\\
&\big(R_{\nu'}\al_{\nu(*)}\big)(u)
=\begin{cases}-\al_{\nu(*)}(R_{\nu'}u-\veps_j^2e_{\nu'})&\text{if $\nu=\nu'=0$}\\
         -\al_{\nu(*)}(R_{\nu'}u-\veps_je_{\nu'})&\text{if $\nu=\nu'\ne0$}\\ 
           \al_{\nu(*)}(R_{\nu'}u) &\text{if $\nu\ne\nu'$}
 \end{cases}
\end{align*}
For the fields $\tilde\al=\big(\al,\{\al_\nu\}\big)\in\tilde\cH_j^{(n-j)}$,
as in \eqref{eqnTHdefexpandedstates}, define
\begin{equation*}
T_x\tilde\al = \big(T_x\al,\{T_x\al_\nu\big\}\big)
\qquad
R_{\nu'}\tilde\al = \big(R_{\nu'}\al,\{R_{\nu'}\al_\nu\big\}\big)
\end{equation*}
\item 
We next define translation and reflection operators, acting 
on functions of the fields by
\begin{align*}
(T_x\cF)(\tilde\al_*,\tilde\al)
&=\cF\big(T_x^{-1}\tilde\al_*,T_x^{-1}\tilde\al\big)\\
%%%
(R_0\cF)(\tilde\al_*,\tilde\al)
&=\overline{\cF\big(R_0^{-1}\tilde\al^*,R_0^{-1}\tilde\al_*^*\big)}\\
%%%
(R_{\nu'}\cF)(\tilde\al_*,\tilde\al)
&=\cF\big(R_{\nu'}^{-1}\tilde\al_*,R_{\nu'}^{-1}\tilde\al\big),
\ 1\le\nu'\le3\\
\end{align*}
\item 
We denote by $\fS$ the symmetry group generated by translations
and reflections acting on functions  $\cF\big(\tilde\al_*,\tilde\al\big)$. 
Since
\begin{equation*}
R_\nu^2=\bbbone\qquad
R_{\nu'}R_\nu=R_\nu R_{\nu'}\qquad R_\nu T_x=T_{R_\nu x}R_\nu
\end{equation*}
the group is finite. We denote by $\fS_\spat$ the subgroup of $\fS$
generated by translations and spatial reflections. It is of index two, meaning
that every element $g\in \fS$ is of one of the forms $g=g'$ or $g=g'R_0$, 
with $g'\in\fS_\spat$.
\item
 A function $\cF(\tilde\al_*,\tilde\al)$ is said to preserve particle 
number if
\begin{align*}
\cF\big(e^{-i\th}\tilde\al_*,e^{i\th}\tilde\al\big)
&=\cF(\tilde\al_*,\tilde\al)\\
\end{align*}
for all $\th\in\bbbr$.
\end{enumerate}
\end{definition}
\begin{remark}\label{remSYfullSymmetry}
Let $\tilde\cF(\tilde\al_*,\tilde\al)$ be given and
set
\begin{equation*}
\cF\big(\al_*,\al)
=\tilde\cF\big((\al_*,\{\partial_\nu\al_*\})\,,\,
               (\al,\{\partial_\nu\al\})\big)
\end{equation*}
then 
\begin{equation*}
(g\cF)\big(\al_*,\al)
=(g\tilde\cF)\big((\al_*,\{\partial_\nu\al_*\})\,,\,
               (\al,\{\partial_\nu\al\})\big)
\end{equation*}
for all $g\in\fS$.
\end{remark}

\begin{proof} 
It suffices to consider $g$ a generator of $\fS$. If $g$ is a translation operator,
the conclusion is obvious. If $g$ is a reflection, and, for example, 
$1\le\nu\le3$, observe that
\begin{alignat*}{3}
\veps_j\partial_\nu\big(R_\nu\al\big)(u)
&=\big[\big(R_\nu\al\big)(u+\veps_j e_\nu)-\big(R_\nu\al\big)(u)&
&\hskip-20pt=\al(R_\nu u-\veps_j e_\nu)-\al(R_\nu u)\\
&=-\big[\al(R_\nu u-\veps_j e_\nu+\veps_j e_\nu)-\al(R_\nu u-\veps_j e_\nu)\big]&
&=-\veps_j\partial_\nu\al(R_\nu u-\veps_je_\nu)
\end{alignat*}
\end{proof}

\begin{example}\label{exSYfullsymmetry}
If
\begin{equation*}
\cF\big(\al_*,\al\big)
   =\int_{\cX_j^{(n-j)}} du\,du'\ \al_*(u)\, K(u,u')\, \al(u')
\end{equation*}
is invariant under $\fS$, then
\begin{equation*}
K(u+x,u'+x)=K(u,u')\quad 
\overline{K(R_0 u',R_0 u)}=K(u,u')\quad 
K(R_\nu u,R_\nu u')=K(u,u')
\end{equation*}
for all $u,u'\in\cX_j^{(n-j)}$, $x\in\cX_0^{(n)}$ and $1\le\nu\le3$.
\end{example}

\begin{lemma}\label{lemSYfullsymmetry}
Let $1\le\nu\le3$ and assume that
\begin{align*}
\cF_2&=\int_{\cX_j^{(n-j)}} du_1\,du_2\ \al_*(u_1)\, K_2(u_1,u_2)\, \al_\nu(u_2)\\
\cF_4&=\int_{\cX_j^{(n-j)}} du_1\cdots du_4\ 
   K_4(u_1,u_2,u_3,u_4)\ \al_*(u_1)\al(u_2)\al_*(u_3)\al_\nu(u_4)
\end{align*}
are invariant under $\fS_\spat$. Then
\begin{align*}
&K_2(u_1+x,u_2+x)=K_2(u_1,u_2) &
K_2(R_{\nu'} u_1,R_{\nu'} u_2)&=K_2(u_1,u_2)\\
&K_2(u_1,u_2)=-K_2(R_\nu u_1,R_\nu  u_2-\veps_j e_\nu)
\end{align*}
and
\begin{align*}
K_4(u_1+x,u_2+x,u_3+x,u_4+x)&=K_4(u_1,u_2,u_3,u_4) \\
K_4(R_{\nu'}u_1,R_{\nu'}u_2,R_{\nu'}u_3,R_{\nu'}u_4)&=K_4(u_1,u_2,u_3,u_4)\\
K_4(u_1,u_2,u_3,u_4)&=-K_4(R_\nu u_1,R_\nu u_2,R_\nu u_3,R_\nu  u_4-\veps_j e_\nu)
\end{align*}
for all $u_1,\cdots,u_4\in\cX_j^{(n-j)}$, 
$x\in\cX_0^{(n)}$ and $1\le\nu'\le 3$  with $\nu'\ne\nu$.

\end{lemma}

\begin{proof} 
We prove the last $K_2$ identity. The other cases are similar.
\begin{align*}
&\int du_1\,du_2\ \al_*(u_1)\, K_2(u_1,u_2)\, \al_\nu(u_2) \\
&\hskip1in=-\int du_1\,du_2\ \al_*(R_\nu u_1)\, K_2(u_1,u_2)\, 
         \al_\nu(R_\nu u_2-\veps_j e_\nu)\\
&\hskip1in=-\int du_1\,du_2\ \al_*(u_1)\, K_2(R_\nu u_1,R_\nu  u_2)\, 
         \al_\nu(u_2-\veps_j e_\nu)\\
&\hskip1in=-\int du_1\,du_2\ \al_*(u_1)\, K_2(R_\nu u_1,R_\nu  u_2-\veps_j e_\nu)\, 
         \al_\nu(u_2)
\end{align*}

\end{proof}

\noindent
There are obvious Lemma \ref{lemSYfullsymmetry} analogs for monomials of
type $\al_{*\nu}\, \al$ and $\al_{*\nu}\al\,\al_*\,\al$.

\begin{remark}\label{remSYsymmetryNorm}
Let $\fm\ge 0$, $g\in\fS$, $0\le\nu,\nu'\le3$ and 
\begin{equation*}
\cF\big(\al_{*\nu'},\al_{\nu}\big)
             =\int du\,du'\ \al_{*\nu'}(u)\, K(u,u')\, \al_{\nu}(u')
\end{equation*}
Then, with the notation of Definition \ref{defHTbasicnorm},
\begin{equation*}
\|g\cF\|_\fm\le e^{2\veps_j \fm}\|\cF\|_\fm
\end{equation*}
\end{remark}
\begin{proof}
If $g$ is a translation, or a reflection $R_{\nu''}$ with 
$\nu''\ne\nu,\nu'$, or if $g=R_\nu$ and $\nu'=\nu$, then 
$\|g\cF\|_\fm=\|\cF\|_\fm$. If $\nu'\ne\nu$ and $g=R_\nu$, with $1\le\nu\le3$,
then
\begin{align*}
\max_u\sum_{u'}\big|-K(R_\nu  u,R_\nu u'-\veps_j e_\nu)\big|e^{\fm|u-u'|}
&=\max_u\sum_{u'}\big|K(u,u')\big|e^{\fm|R_\nu u-R_\nu u'+\veps_j e_\nu|}\\
&\le e^{\veps_j \fm}\|\cF\|_\fm
\end{align*}
Similarly, $\|g\cF\|_\fm\le e^{\veps_j^2\fm}\|\cF\|_\fm$ when 
$\nu'\ne\nu$ and $g=R_0$.
By the relations of Definition \ref{defSYfullSymmetry}.d, every $g\in\fS$
may be written as a product of a translation and reflections,
with each $R_{\nu''}$, $0\le\nu''\le3$ appearing at most once.
The claim follows.
\end{proof} 

%\newpage
%&&&&&&&&&&&&&&&&&&&&&&&&&&&&&&&&&
\section{Inequalities for Parameters}\label{appIneq}
%&&&&&&&&&&&&&&&&&&&&&&&&&&&&&&&&&

\begin{lemma}\label{lemPARcompradan}
Assume that $\eps$ is sufficiently small. 
%and that $\fv_0$ and 
%$\sfrac{1}{L}$ are sufficiently small, depending on $\eps$.
\begin{enumerate}[label=(\alph*), leftmargin=*]
\item
We have
\begin{align*}
L^{\np} &\le \big( \sfrac{\fv_0^{5\eps}}{\mu_0-\mu_*}\big)^{\sfrac{1}{2+5\eps} }
\end{align*}
and, for $0\le n\le \np$,
\begin{align*}
 \sfrac{\log\fv_0}{\log(\mu_0-\mu_*)} \log L^n
 & \le -\half (1-6\eps)\log \fv_0 +O(\eps^2|\log \fv_0|)
  \end{align*}

\item
For $0\le n\le \np$
\refstepcounter{equation}\label{eqnPARestrad}
\begin{align}
\fe_\fl(n)  & \le \fv_0^{\frac{3}{2}\eps}
\tag{\ref{eqnPARestrad}.a}\\
\sfrac{\fv_0}{L^n} \ka(n)^2 \ka_\fl(n)^2 
&\le  \fv_0^{\eps/6} \fe_\fl(n) 
&\tag{\ref{eqnPARestrad}.b}
\end{align}

\item
The quantity 
$\, \sfrac{\ka_\fl(n)}{\ka'(n)\fe_\fl(n) }\,$ is monotonically decreasing with $n$ and bounded above by $\,\fv_0^{\eps/2} \,$.

\item 
Let $\fv_0$ be sufficiently small, depending on $\eps$. Then, for all
$C>0$, the infinite product
$\,\Pi_0^\infty(C)=
\smprod_{j=1}^\infty  \big(1 +C  \sfrac{\fe_\fl(j-1)}{\ka(j)^2}\big) 
\,$
is finite and, for $n \ge 1$ and all $\vp \in \fD$
\begin{align*}
 \fr_\vp(n,C) \le (1+C)\,\Pi_0^\infty(C)\ \fv_0^{1-5\eps} \begin{cases}
    1
    &\text{if $\vp = (1,1,0)$}\\
    L^{-(2\eta'-\eta_\fl)n}
    &\text{if $\vp =(0,1,1)$}\\
    1 
    &\text{if $\vp =(0,0,2)$}\\
   \fv_0 \,L^{-4n} 
    &\text{if $\vp = (6,0,0)$}
\end{cases}
\end{align*}
\end{enumerate}
\end{lemma}

\begin{proof} (a) follows from Remark \ref{remHTbasicnorm}.

\Item (b)
By part (a),
\begin{align*}
\log \fe_\fl(n)  
%& =\big( \sfrac{1}{3} -2\eps \big) \log\fv_0 + \eta_\fl \log L^n
%\\
& =\big( \sfrac{1}{3} -2\eps \big) \log\fv_0 +
\big(\sfrac{2}{3}-4\eps\big)\,\sfrac{\log\fv_0}{\log(\mu_0-\mu_*)}\log L^n
\\
&\le\big( \sfrac{1}{3} -2\eps \big) 
\big( 1 -(1-6\eps)\big)\log\fv_0 +O(\eps^2 | \log\fv_0|)
\\
&=2 \eps \log\fv_0+O(\eps^2 | \log\fv_0|)
\end{align*}
and
\begin{align*}
\log \big( \sfrac{\fv_0}{L^n} \ka(n)^2 \ka_\fl(n)^2  \big) -\log\fe_\fl(n) 
& = \log\big( \fv_0^{1/3+\eps} L^{(2\eta -1+\eps )n}\big)
-\log\big(\fv_0^{\sfrac{1}{3}-2\eps } L^{\eta_\fl n}\ \big)
\\
&= 3\eps \log\fv_0 + (2\eta -1-\eta_\fl+\eps) \log L^n
\\
& \le  3\eps \log\fv_0 + \big( 4\eps \sfrac{\log\fv_0}{\log(\mu_0-\mu_*)}
+\eps\big) \log L^{n}
\\
& \le  3\eps \log\fv_0 
      + \sfrac{16}{3}\eps \sfrac{\log\fv_0}{\log(\mu_0-\mu_*)} \log L^{n}
\\
&\le  \sfrac{1}{3}\eps \log\fv_0 + O(\eps^2|\log \fv_0|)
\end{align*}

\Item (c) By Remark \ref{remHTbasicnorm},
\begin{align*}
\log \ka_\fl(n) - \log \ka'(n) -\log \fe_\fl(n)  
& =\big( \sfrac{\eps}{2} - \eta' -\eta_\fl\big) \log L^n       
- \big(\sfrac{\eps}{2} - \sfrac{1}{3}+\eps +\sfrac{1}{3} -2\eps \big) \log\fv_0
\\
&=-\big[\sfrac{3}{2}(1-\eps) - (\sfrac{1}{3}+4\eps) \sfrac{\log\fv_0}{\log(\mu_0-\mu_*)}\big]
 \log L^{n} 
+\sfrac{\eps}{2} \log\fv_0  
\end{align*}
is monotonically decreasing with $n$.

\Item (d) 
The fact that the infinite product is finite is immediate.
Clearly, 
\begin{equation*}
\fr_\vp(n,C) \le
(1+C)\,\Pi^\infty_0(C) \begin{cases}
 \fv_0^{1-4\eps}
 + \fv_0^{1-3\eps}\sum\limits_{\ell=1}^{n} 
   \sfrac{ L^{\eta_\fl\ell} }{L^{(\eta+\eta')\ell}}
     &\text{if $\vp =(1,1,0)$}\\
%%%%%%%%%%%
 \sfrac{\fv_0^{1-4\eps}}{L^n}
 +\sfrac{\fv_0^{1-3\eps}}{L^n} \sum\limits_{\ell=1}^{n} 
   L^{\ell}\sfrac{L^{\eta_\fl\ell} }{L^{2\eta'\ell}}
     &\text{if $\vp =(0,1,1)$}\\
%%%%%%%%%%
 \fv_0^{1-4\eps}
 + \fv_0^{1-3\eps}\sum\limits_{\ell=1}^{n} 
   \sfrac{ L^{\eta_\fl\ell} }{L^{2\eta'\ell}}
     &\text{if $\vp =(0,0,2)$}\\
%%%%%%%%%%%
 \sfrac{\fv_0^{2-\eps}}{L^{4n}}
 +\sfrac{\fv_0^{7/3-7\eps}}{L^{4n}} \sum\limits_{\ell=1}^{n} 
   L^{4\ell}\sfrac{L^{\eta_\fl\ell} }{L^{6\eta\ell}}
     &\text{if $\vp =(6,0,0)$}\\
\end{cases}
\end{equation*}
In the case that $\vp =(0,1,1)$, the successive terms in the sum increase by a factor of at most $L^{1-\const\eps}$, while in the other
cases they decrease by a factor of at least $L^{\const\eps}$.
\end{proof}

As in \cite[Definition \defOSFsdf\ and Lemma \lemOSFmainlem]{PAR2}, we 
define, for $\vp = (p_u,p_0,p_\sp)$,
\begin{align*}
\De(\vp)=\sfrac{3}{2}p_u+\sfrac{7}{2}p_0 +\sfrac{5}{2}p_\sp
 \qquad
\ka^\vp(n)  = \ka(n)^{p_u} \ka'(n)^{p_0+p_\sp}
\end{align*}
With this notation
\begin{equation}\label{eqnPARformforrpnC}
 \fr_\vp(n,C) =
\fr_\vp(0)  L^{(5-\De(\vp))n} \,\Pi_0^n(C)
+C \sum_{\ell=1}^{n} L^{(5-\De(\vp))(n-\ell)}\sfrac{ \fe_\fl(\ell-1) }{\ka^\vp(\ell)} \,\Pi_\ell^n(C)
\end{equation}

\begin{lemma}\label{lemPARestrpnC}
Assume that $\eps$ is sufficiently small. Let $C>0$ and assume that $\fv_0$ is so small that
$\,\eps |\log\fv_0| \ge 2\log (1+C)\,\Pi_0^\infty(C)\,$. Let $\vp \in \fD$ 
and $1\le n\le \np$. Then
\begin{enumerate}[label=(\alph*), leftmargin=*]
\item 
\hskip0.75in
$
 \sfrac{\ka^\vp(n)}{\ka(n)^2} \fr_\vp(n,C) 
\le  \fv_0^{\frac{2}{3}-  \frac{3}{2}\eps}
$
and
\begin{align*} 
\sfrac{\ka_\fl(n)}{\ka'(n)} \ka^\vp(n) \, \fr_\vp(n,C)
& \le \fv_0^{1/8} \fe_\fl(n)
 \qquad & &\text{if }\vp\ne (6,0,0)
\\
\sfrac{\ka_\fl(n)}{\ka(n)} \ka^\vp(n) \, \fr_\vp(n,C)
& \le \fv_0^\eps \fe_\fl(n) \qquad & &\text{if }\vp = (6,0,0)
\end{align*}

\item
We have
\begin{align*}
\fr_\vp(n,C)
      \le \fv_0^{\eps}
        \min\big\{\fv_0^{\frac{4}{3}-7\eps},\ \sfrac{\fv_0}{L^n}\big\} 
           \begin{cases}
              \ka(n) \ka_\fl(n)
                &\text{if $\vp=(1,1,0),\ (0,1,1),\ (0,0,2)$}\\
                                          \noalign{\vskip0.05in}
            \sfrac{1}{\ka_\fl(n)^2}
                &\text{if $\vp=(6,0,0)$}
\end{cases}
\end{align*}
\end{enumerate}
\end{lemma}

\begin{proof}  (a)
By Lemma \ref{lemPARcompradan}.d
\begin{equation}\label{eqnPARauxlemPARestrpnC}
\log \big(\ka^\vp(n) \,\fr_\vp(n,C)\,\big) 
\le \sfrac{\eps}{2}|\log \fv_0| +  \begin{cases}
(\sfrac{1}{3}-3\eps)\log \fv_0 +(\eta+\eta') \log L^n
    &\text{if $\vp = (1,1,0)$}\\
    (\sfrac{1}{3}-3\eps)\log \fv_0 +\eta_\fl \log L^n
    &\text{if $\vp =(0,1,1)$}\\
    (\sfrac{1}{3}-3\eps)\log \fv_0 +2\eta' \log L^n 
    &\text{if $\vp =(0,0,2)$}\\
   \eps \log \fv_0 +(6\eta-4) \log L^n
    &\text{if $\vp = (6,0,0)$}
\end{cases}
\end{equation}
Consequently
\begin{equation*}
\log \big( \sfrac{\ka^\vp(n)}{\ka(n)^2}  \fr_p(n,C) \big)
\le\begin{cases}
(1-\sfrac{11}{2}\eps)\log \fv_0 -(\eta-\eta') \log L^n
    &\hskip-3.5pt\text{if $\vp = (1,1,0)$}\\
\noalign{\vskip0.05in}
    (1-\sfrac{11}{2}\eps)\log \fv_0 -(2\eta-\eta_\fl)\log L^n
    &\hskip-3.5pt\text{if $\vp =(0,1,1)$}\\
\noalign{\vskip0.05in}
    (1-\sfrac{11}{2}\eps)\log \fv_0 -2(\eta-\eta') \log L^n 
    &\hskip-3.5pt\text{if $\vp =(0,0,2)$}\\
\noalign{\vskip0.05in}
 (\sfrac{2}{3}-  \sfrac{3}{2}\eps) \log \fv_0 -4(1-\eta) \log L^n
    &\hskip-3.5pt\text{if $\vp = (6,0,0)$}
\end{cases}
\end{equation*}
The first inequality of the Lemma is immediate. 

As
\begin{align*}
\log \sfrac{\ka_\fl(n)}{\ka'(n) } - \log \fe_\fl(n) 
& =\sfrac{\eps}{2} \log\fv_0  
-  ( \eta'+\eta_\fl-\sfrac{\eps}{2} ) \log L^n
 \\
 \log \sfrac{\ka_\fl(n)}{\ka(n) } - \log \fe_\fl(n) 
& = \sfrac{\eps}{2} \log\fv_0 
- ( \eta+\eta_\fl-\sfrac{\eps}{2} ) \log L^n
\end{align*}
the inequalities \eqref{eqnPARauxlemPARestrpnC} give 
for the case $\vp =(1,1,0)$
\begin{align*}
&\log\big( \sfrac{\ka_\fl(n)}{\ka'(n)} \ka^\vp(n)\,\fr_\vp(n,C)\,\big) 
- \log \fe_\fl(n)
\le  (\sfrac{1}{3} - 3\eps)\log \fv_0 
+(\eta-\eta_\fl +\sfrac{\eps}{2}) \log L^n
\\
&\hskip1.5in\le  (\sfrac{1}{3} -3\eps)\log \fv_0 
+\big(\sfrac{1+\eps}{2} -(\sfrac{1}{3}\! -\! 4\eps) \sfrac{\log\fv_0}{\log(\mu_0-\mu_*)} \big)
 \log L^n
\\
&\hskip1.5in\le (\sfrac{1}{3} - 3\eps)\log \fv_0 
+(\sfrac{1}{3} + 5\eps) \sfrac{\log\fv_0}{\log(\mu_0-\mu_*)} 
 \log L^n
\\
&\hskip1.5in\le  \sfrac{1}{8}\log \fv_0 
\end{align*}
again by Remark \ref{remHTbasicnorm} and Lemma \ref{lemPARcompradan}.a. 
Since $\eta'<\eta$, the same bound applies in the case $\vp =(0,0,2)$.
In the case $\vp =(0,1,1)$, the desired inequality is easy. 
Finally, in the case $\vp =(6,0,0)$
\begin{align*}
&\log\big( \sfrac{\ka_\fl(n)}{\ka(n)} \ka^\vp(n) 
                 \,\fr_\vp(n,C)\,\big) - \log \fe_\fl(n) 
\le \eps\log \fv_0
+(5\eta-4-\eta_\fl +\sfrac{\eps}{2}) \log L^n
\\
&\hskip1.5in  \le  \eps \log \fv_0
+\Big(-\sfrac{3}{2}
  +\big(1+4\veps\big)\sfrac{\log\fv_0}{\log(\mu_0-\mu_*)}  +\sfrac{\eps}{2}\Big) \log L^n
\\
&\hskip1.5in  \le  \eps \log \fv_0
\end{align*}

\Item (b) By Lemma \ref{lemPARcompradan}.d,
\begin{equation}\label{eqnPARestrpnC}
\fr_\vp(n,C)
      \le  \fv_0^\eps \begin{cases}
              \fv_0^{1-\frac{13}{2}\eps}
                &\text{if $\vp=(1,1,0),\ (0,1,1),\ (0,0,2)$}\\
                                            \noalign{\vskip0.05in}
            L^{-4n}\fv_0^{2-\frac{13}{2}\eps}
                &\text{if $\vp=(6,0,0)$}
 \end{cases}
\end{equation}
The case $(6,0,0)$ is obvious.
The remaining cases follow from
\begin{align*}
\frac{\fv_0^{1-\frac{13}{2}\eps}}{\ka(n)\ka_\fl(n)}
&=L^{(-\eta-\frac{\eps}{2})n}\fv_0^{\frac{4}{3}-7\eps}
\le\min\big\{\fv_0^{\frac{4}{3}-7\eps},\ 
               \sfrac{\fv_0}{L^n}\, L^{(1-\eta)\np}
                      \fv_0^{\frac{1}{3}-7\eps}\big\}
\le\min\big\{\fv_0^{\frac{4}{3}-7\eps},\ \sfrac{\fv_0}{L^n}\big\}
\end{align*}
since, by Remark \ref{remHTbasicnorm},\ \ \ 
$
L^{(1-\eta)\np}\fv_0^{\frac{1}{3}-7\eps}
\le L^{\frac{1}{4}\np}\fv_0^{\frac{1}{3}-7\eps}
\le \fv_0^{-\frac{1}{4}(\frac{2}{3}-3\eps)}
            \fv_0^{\frac{1}{3}-7\eps}
\le 1
$.
\end{proof}

\begin{lemma}\label{lemPARmunvn}
\ 
\begin{enumerate}[label=(\alph*), leftmargin=*]
\item
If $1\le n\le \np$,
\begin{equation*}
\fv_0^{1-8\eps}
     \smsum_{\ell=1}^n \sfrac{1}{L^{(2-3\eps)\ell}}
        \big[\fv_0^{\frac{1}{3}-6\eps} +L^{2\ell}(\mu_0-\mu_*)\big]
\le \half \fv_0^{\frac{4}{3}-15\eps}
\end{equation*}

\item \ \ \ 
$
 {\displaystyle\sum_{\ell=1}^n}\  \sfrac{L^\ell }{\ka(\ell)^4}\fe_\fl(\ell-1)
\le  \fv_0^{\frac{5}{3}-6\eps}
$.
\end{enumerate}
\end{lemma}
\begin{proof} (a) The claim follows from
\begin{align*}
\smsum_{\ell=1}^n \sfrac{1}{L^{(2-3\eps)\ell}}
        \big[\fv_0^{\frac{1}{3}-6\eps} +L^{2\ell}(\mu_0-\mu_*)\big]
&\le \fv_0^{\frac{1}{3}-6\eps} + \fv_0^{-\eps}L^{4\eps \np}(\mu_0-\mu_*) \\
&\le \fv_0^{\frac{1}{3}-6\eps} + \fv_0^{-\eps}(\mu_0-\mu_*)^{1-2\eps} 
\quad& &\text{by Definition \ref{defHTbasicnorm}.b}\\
&\le \fv_0^{\frac{1}{3}-6\eps} 
          + \fv_0^{-\eps}\fv_0^{\sfrac{8}{9}(1-2\eps)} 
\quad& &\text{by \S\ref{sectINTstartPoint}}
\end{align*}

\Item (b) By Definition \ref{defHTbasicnorm}.a,
\begin{align*}
\smsum_{\ell=1}^n \sfrac{L^\ell }{\ka(\ell)^4}\fe_\fl(\ell-1)
&\le \smsum_{\ell=1}^\infty L^{(1-4\eta+\eta_\fl)\ell}\,
            \fv_0^{\frac{5}{3}-6\eps}
\le \fv_0^{\frac{5}{3}-6\eps}
\end{align*}
\end{proof}

\pagebreak[2]
\begin{corollary}\label{corPARmunvn}
Let $1\le n\le \np$.
\begin{enumerate}[label=(\alph*), leftmargin=*]
\item
If the real number $\mu$ obeys
$|\mu-\mu_n^*| \le L^{2n}\,\fv_0^{1-8\eps}
     \smsum_{\ell=1}^n \sfrac{1}{L^{(2-3\eps)\ell}}
        \big[\fv_0^{\frac{1}{3}-6\eps} +L^{2\ell}(\mu_0-\mu_*)\big]$
then
\begin{equation*}
\big|\mu-L^{2n}(\mu_0-\mu_*)\big|\le \fv_0^{1-\eps}
                   +L^{2n}\fv_0^{\frac{4}{3}-15\eps}
\end{equation*}
and
\begin{equation*}
|\mu|\le 2L^{2n}(\mu_0-\mu_*)+\fv_0^{1-\eps}
\le 4\min\big\{\fv_0^{5\eps}\,,\,L^{2n}\fv_0^{\frac{8}{9}+\eps}\big\}
\end{equation*}

\item
If the quartic monomial $\cV$ obeys
$\big\|\cV-\cV^{(u)}_n\big\|_{2m}
   \le \sfrac{1}{L^n\fv_0^\eps} 
    \smsum_{\ell=1}^n \sfrac{L^\ell }{\ka(\ell)^4}\fe_\fl(\ell-1)$
then
\begin{equation*}
\|\cV-\cV_n^{(u)}\|_{2m}\le \sfrac{1}{L^n}\fv_0^{\frac{5}{3}-7\eps}
\qquad\text{and}\qquad
\big\|\cV\big\|_{2m} \le\sfrac{\fv_0}{L^n}
\end{equation*}
\end{enumerate}
\end{corollary}
\begin{proof} (a) 
By Lemma \ref{lemPARmunvn}.a and \cite[Lemma \lemOSImustar]{PAR2},
\begin{align*}
\big|\mu-L^{2n}(\mu_0-\mu_*)\big|
&\le \big|\mu_n^*-L^{2n}(\mu_0-\mu_*)\big| \\
&\hskip1in+L^{2n}\,\fv_0^{1-8\eps}
     \smsum_{\ell=1}^n \sfrac{1}{L^{(2-3\eps)\ell}}
        \big[\fv_0^{\frac{1}{3}-6\eps} +L^{2\ell}(\mu_0-\mu_*)\big] \\
&\le \fv_0^{1-\eps}
+\sfrac{L^{2n}}{2} \fv_0^{\frac{4}{3}-15\eps} 
\end{align*}
The second inequality now follows from $\fv_0^{\frac{4}{3}-16\eps}
\le \mu_0-\mu_*$ and Definition \ref{defHTbasicnorm}.b.

\Item (b) is trivial.
\end{proof}

\begin{lemma}\label{lemPARlnmuzeromustar}
Let $0\le n\le \np$.
\begin{enumerate}[label=(\alph*), leftmargin=*]
\item 
If $0\le\al\le 1$, then
$\ \ 
L^{2n}(\mu_0-\mu_*)\le L^{2\al n}\,\fv_0^{5\eps(1-\al)}\,(\mu_0-\mu_*)^\al
$.

\item 
\hskip0.5in $L^{2n}(\mu_0-\mu_*)
\le \min\big\{\fv_0^{5\eps}\,,\,L^{2n}\fv_0^{\frac{8}{9}+\eps}\big\}$

\item 
\hskip0.5in 
  $L^{2n}(\mu_0-\mu_*) \le  \sfrac{\fv_0^{1+\eps}}{L^n} \ka(n)^2 $

\item
\hskip0.5in 
  $L^{2n}(\mu_0-\mu_*)\, \ka_\fl(n)^2\le\fv_0^\eps\, \fe_\fl(n)$
\end{enumerate}
\end{lemma}

\begin{proof} (a) By Definition \ref{defHTbasicnorm}.b,
\begin{align*}
L^{2n}(\mu_0-\mu_*) 
&\le L^{2\al n}\big[L^{2\np}(\mu_0-\mu_*)\big]^{1-\al}(\mu_0-\mu_*)^\al 
\le L^{2\al n}\fv_0^{5\eps(1-\al)}(\mu_0-\mu_*)^\al
\end{align*}

\Item (b) 
By part (a) with $\al=0$, $L^{2n}(\mu_0-\mu_*)\le \fv_0^{5\eps}$.
By \S\ref{sectINTstartPoint}, 
$
L^{2n}(\mu_0-\mu_*)
\le L^{2n}\fv_0^{\frac{8}{9}+\eps}
$.

\Item (c) 
By part (a) with $\al = \eta-\half$ and Definition \ref{defHTbasicnorm}.b,
\begin{align*}
&\log\big[L^{2n}(\mu_0-\mu_*)\big]
-\log\big[\sfrac{\fv_0^{1+\eps}}{L^n} \ka(n)^2\big] \\
&\hskip0.5in
  \le 5\eps\big(\sfrac{3}{2}-\eta\big)\log\fv_0 +(\eta-\half)\log(\mu_0-\mu_*)
    -\big(\sfrac{1}{3}+3\eps\big)\log\fv_0 \\
&\hskip0.5in
  \le \sfrac{25}{8}\eps\log\fv_0 +\sfrac{1}{3}\log\fv_0
    -\big(\sfrac{1}{3}+3\eps\big)\log\fv_0 \\
&\hskip0.5in
  \le 0 
\end{align*}

\Item (d) 
By part (a) with $\al = \sfrac{\eta_\fl-\eps}{2}$, 
Definition \ref{defHTbasicnorm}.b and Remark \ref{remHTbasicnorm},
\begin{align*}
&\log\big[L^{2n}(\mu_0-\mu_*)\big]
-\log\big[\fv_0^{\eps} \sfrac{\fe_\fl(n)}{\ka_\fl(n)^2}\big] \\
&\hskip0.5in
  \le 5\eps(1-\sfrac{\eta_\fl-\eps}{2})\log\fv_0 
      +\sfrac{\eta_\fl-\eps}{2}\log(\mu_0-\mu_*)
    -\sfrac{1}{3}\log\fv_0 \\
&\hskip0.5in
  \le 5\eps(1-\sfrac{\eta_\fl}{2})\log\fv_0 
      -2\eps\log\fv_0
      -\sfrac{\eps}{2}\log(\mu_0-\mu_*)\\
&\hskip0.5in
  \le 5\eps(1-\sfrac{1}{2}\times\sfrac{2}{3}\times\sfrac{9}{8})\log\fv_0 
      -2\eps\log\fv_0
      -\sfrac{2\eps}{3}\log\fv_0\\
&\hskip0.5in \le 0
\end{align*}
\end{proof}

%\newpage
%%%%%%%%%%%%%%%%%%%%%%%%%%%%%%%%%%%%%%%%%%%%%%%%%%
\section{Rewriting the Output of the Ultraviolet Flow}\label{appSZrewrite}
%%%%%%%%%%%%%%%%%%%%%%%%%%%%%%%%%%%%%%%%%%%%%%%%%%

%%%%%%%%%%%%%%%%%%%%%%%%%%%%%%%%%%%%
\subsection{The Model}\label{sectSZrewriteModel}
%%%%%%%%%%%%%%%%%%%%%%%%%%%%%%%%%%%%

We now give a technically complete description of the output
of \cite{UV} and, in Proposition \ref{propSZprepforblockspin}, give the mathematically precise description of the starting point \eqref{eqnHTstartingpoint} of our analysis.
The models under consideration are characterized by
\begin{itemize}[leftmargin=*, topsep=2pt, itemsep=2pt, parsep=0pt]
\item a kinetic energy operator
\begin{equation*}
\bh = \nabla^*\cH\nabla
\end{equation*}
where $\cH$ is a real, translation invariant, reflection invariant, 
strictly positive definite operator on the space, $L^2\big({(\bbbz^3)}^*)$, 
of functions  on the set, ${(\bbbz^3)}^*$, of nearest neighbor bonds of 
the lattice $\bbbz^3$. 
\item
a real, symmetric, translation invariant, reflection invariant,
strictly positive definite two--body interaction $\bv$ on $\bbbz^3$ and
\item
 a real chemical potential $\mu$. 
\end{itemize}
We denote by $h$ and $v$ the periodizations of $\bh$ and $\bv$
to the finite lattice $X=\bbbz^3/L_\sp\bbbz^3$.

The results of \cite{UV} apply under the following conditions
on the above data. Pick any mass $\mass>0$ and constants $c_v,D_\cH,K_\mu>0$,
$\half<e_\mu\le 1$ and $0<c_\cH<C_\cH$. There is a number $1\ge \bar\fv>0$, 
depending on these constants, such that for all $0<\fv\le\bar\fv$, the results
of \cite{UV} hold for all $\mu$'s, $\bv$'s and $\cH$'s that satisfy
\begin{itemize}[leftmargin=*, topsep=2pt, itemsep=2pt, parsep=0pt]
\item
  $\sum\limits_{\atop{\bx\in\bbbz^3}{_{1\le i,j\le 3}}}e^{6\,\mass d(\bx,0) }
               \big|\cH\big(\<0,e_i\>\,,\,\<\bx,\bx+e_j\>\big)\big|
      \le D_\cH$, where $e_i$ is the unit vector in the $i^{\rm th}$ direction
      and $d(\bx,\by)$ is the Euclidean distance from $\bx$ to $\by$, and
\item      
    the eigenvalues of the periodization of $\cH$ lie between $c_\cH$ 
      and $C_\cH$ and
\item 
the norm
\begin{equation*}
\tn \bv\tn = \sup_{\bx\in\bbbz^3} \smsum_{\by \in\bbbz^3} \,e^{5\mass\,d(\bx,\by)}
\,|\bv(\bx,\by)|
\end{equation*}
obeys $\sfrac{1}{4}\fv\le \tn \bv\tn\le\half\fv$ and
\item
the smallest eigenvalue of $v$ is at least $c_v \tn \bv\tn$ and
\item
 $|\mu|\le K_\mu\fv^{e_\mu}$.
\end{itemize}

\subsection{The Output of \cite{UV}}\label{sectSZrewriteOutput}

Let $H$ be the second quantized Hamiltonian with kinetic energy operator
$h$ and two--body interaction $v$, and let $N$ be the number operator.
Fix, as in \cite[Hypothesis 2.14]{UV}, strictly positive exponents 
$e_\r$, $e_\R$, and $e_{\RP}$ that obey
\begin{align*}
3e_\R+4e_\r&<1\quad && 
1&\le 4e_\R+2e_\r \quad && 
&2(e_\R+e_\r)<e_\mu \quad &&
e_{\RP}+ e_\r&<1 \quad && 
\half&\le e_{\RP} 
\end{align*}
Think of $e_\r$ as being just slightly larger than $0$, 
$e_\R$ as being slightly smaller than $\sfrac{1}{3}$,
$e_\mu$ as being slightly smaller than $\sfrac{2}{3}$,  
and $e_\RP$ as being between one half and one. 
In \cite[Theorem 2.16]{UV} (a self--contained treatment of the pure small 
field part of the argument is also given in \cite{CPS}) we prove that 
there exist constants $K,\th>0$ (we may assume that $\th\le 1$) and a function $\,I_\theta(\al_* ,\be)\,$ 
of two complex valued fields $\,\al_*\,$ and
$\,\be$ on $X$ such that
\begin{equation}\label{eqnSZuvoutput}
\Tr\, e^{-\sfrac{1}{kT}\,(H-\mu N)} =
\int \hskip-28pt\prod_{\hskip22pt\lower7pt\hbox{$\sst\tau\in\th\bbbz
\cap(0,\frac{1}{kT}]$}}\hskip-22pt 
\Big[ \smprod_{\bx\in X} 
\sfrac{ d\al_\tau(\bx)^\ast\wedge d\al_\tau(\bx)}{2\pi \imath} \,
                      e^{-\al_\tau(\bx)^\ast \al_\tau(\bx)}\Big]
I_\theta(\al_{\tau-\th}^\ast ,\al_\tau) 
\end{equation}
for all temperatures $T>0$. Here $\al_0=\al_{\frac{1}{kT}}$. 
We also proved that it is possible to write $\,I_\theta\,$ as the sum 
of a dominant part $\,I_\theta^{(SF)}\,$, called the \emph{pure small 
field contribution}\footnote{$\,I_\theta^{(SF)}\,$ is the
$\Om=X$ term in the formula given for $I_\th(\al^*,\be)$ in 
\cite[Theorem 2.16]{UV}},  and terms, indexed by proper subsets of $X$,
which are nonperturbatively small, exponentially in the size of the subsets.
The dominant part has a logarithm. More precisely
\begin{equation}\label{eqnSZspa}
I^{(\rm SF)}_\th(\al^*,\be)
=\cZ_\th^{|X|} e^{\langle \al^*,\,j(\th) \be \rangle_X  +V_\th(\al^*,\be)
+\cD_\th(\al^*,\be)} \chi_\th(\al,\be)
\end{equation}
where
\begin{itemize}[leftmargin=*, topsep=2pt, itemsep=2pt, parsep=0pt]
\item
 $\,\cZ_\th\,$is a normalization constant, 
\item
 $\<f,g\>_X=\sum\limits_{\bx\in X}f(\bx)g(\bx)$ is the 
``real'' inner product of $f,g\in L^2(X)$,
\item
 $j(t)=e^{-t(h-\mu)}$
\item
$
V_\th(\al^*,\be) 
= -\int_0^\th  \<\big[ j(t)\al^*\big] \big[ j(\th -t)\be \big]\, , \,
   v \big[ j(t)\al^*\big] \big[ j(\th -t)\be \big] \>_X\ dt
$
\item
 the function $\,\cD_\th(\al_*,\be)\,$ is
 analytic in the fields $\,\al_*\,$ and $\,\be\,$ and is invariant under
the $U(1)$ symmetry $\al_*\rightarrow e^{-i t}\al_*$, 
$\be\rightarrow e^{it}\be$. 
Furthermore, it can be decomposed in the form
\begin{equation*}
\cD_\theta(\al_*,\be) 
= \cR_\theta(\al_*,\be) +\cE_\theta(\al_*,\be) 
\end{equation*}
with 
\begin{itemize}[leftmargin=*, topsep=2pt, itemsep=2pt, parsep=0pt]
\item
a function $\,\cR_\theta(\al_*,\be)\,$ that is bilinear in 
$\al_*$ and $\be$ whose norm, as in  Definition \ref{defHTabstractnorm}, 
with mass $2\mass$ and weight $\ka=2\big(\sfrac{1}{\th\fv}\big)^{e_\R+e_\r}$, 
for both $\al_*$ and $\be$, is bounded by
$
%K\, e^{-2\mass \fc}\, \theta^2\,  \r(\theta)^2\, \R(\theta)^2
%= K\ \fv^{2\mass\log\sfrac{1}{\fv}-2e_\R-4e_\r} \th^{2-2e_\R-4e_\r}\le 
K\,\th\,\fv^{\mass\log\frac{1}{\fv}}
$
and 
\item
 a function $\,\cE_\theta(\al_*,\be)\,$
that has degree at least two\footnote{By this we
mean that every monomial appearing in the power series expansion of these
functions contains a factor of the form 
$\,\al_*(\bx_1)\,\al_*(\bx_2)\,\be(\bx_3)\,\be(\bx_4)  \,$.} 
both in $\al_*$ and in $\be$,  whose norm with mass 
$2\mass$ and weight $\ka=2\big(\sfrac{1}{\th\fv}\big)^{e_\R+e_\r}$, is bounded by
$
% K\, \theta^2\, \tn v \tn^2\, \r(\theta)^2\, \R(\theta)^6
%= K\,\big(\sfrac{\tn v\tn}{\fv}\big)^2 (\fv\th)^{2-6e_\R-8e_\r}\le 
K\, (\fv\th)^{2-6e_\R-8e_\r}
$.
\end{itemize}
\item
The ``small field cut off function'' $\,\chi_\th(\al,\be)\,$ is one if 
\begin{itemize}[leftmargin=*, topsep=2pt, itemsep=2pt, parsep=0pt]
\item
 $|\al(\bx)|, |\be(\bx)|
   \le \big(\sfrac{1}{\th\fv}\big)^{e_\R+e_\r}$ for all $\bx\in X$
and
\item
 $|\nabla\al(b)|, |\nabla\be(b)|
 \le\big(\sfrac{1}{\th}\big)^{e_{\RP}}\big(\sfrac{1}{\th\fv}\big)^{e_\r}$ 
for all  bonds $b$ on $X$ and
\item
 $|\al(\bx)-\be(\bx)|\le \big(\sfrac{1}{\th\fv}\big)^{e_\r}$ 
   for all $\bx\in X$
\end{itemize}
and is zero otherwise.
\end{itemize}

%%%%%%%%%%%%%%%%%%%%%%%%%%%%%%%%%% 
\subsection{The Rewriting}\label{sectSZrewriteRewrite}
%%%%%%%%%%%%%%%%%%%%%%%%%%%%%%%%%

\begin{proposition}\label{propSZprepforblockspin}
Make the hypotheses of \S\ref{sectSZrewriteModel} and 
\S\ref{sectSZrewriteOutput}. Then
\begin{align*}
&\int \hskip-28pt\prod_{\hskip22pt\lower5pt\hbox{$\sst\tau\in\th\bbbz
\cap(0,\sfrac{1}{kT}]$}}\hskip-22pt 
\Big[ \smprod_{\bx\in X} 
\sfrac{ d\al_\tau(\bx)^\ast\wedge d\al_\tau(\bx)}{2\pi \imath} \,
                      e^{-\al_\tau(\bx)^\ast \al_\tau(\bx)}\Big]
I_\theta^{(SF)}(\al_{\tau-\th}^\ast ,\al_\tau) 
\\
&\hskip0.1in
= \cZ_\In^{|\cX_0|}\!
   \int \Big[ \smprod_{x\in \cX_0} 
        \sfrac{ d\psi(x)^\ast\wedge d\psi(x)}{2\pi \imath}\Big] \,
  e^{-\<\psi_*,\,D_0\psi\>_0
       -\cV_0(\psi_*,\psi)
       +\mu_0 \<\psi_*,\,\psi\>_0
       +\cR_0(\psi^*,\psi) 
       +\cE_0(\psi^*,\psi) }\chi_0(\psi)
\end{align*}
where 
\begin{itemize}[leftmargin=*, topsep=2pt, itemsep=2pt, parsep=0pt]
\item
 $\cZ_\In=\cZ_\th e^{-\th\mu}$
\item
$
\cX_0 = \big( \bbbz \times \bbbz^3 \big) / 
            \big(\sfrac{1}{\th k T}\bbbz \times L_\sp \bbbz^3 \big)
$
\item
$
D_0=\bbbone - e^{-h_0} -e^{-h_0} \partial_0
\,$ with $h_0 = \th h$ and $\partial_0$ the forward time derivative.
%%%%%%%%%%%%%%%%%%%%%%%%%%%%%%
\item
 There is a real--valued kernel $\bV_0(x_1,x_2,x_3,x_4)$
on $\big((\bbbz/\sfrac{1}{\th kT}\bbbz) \times \bbbz^3 \big)^4$ 
that is invariant under $x_1\leftrightarrow x_3$ and under 
$x_2\leftrightarrow x_4$ and under the symmetry group $\fS$,
and there is a constant $K_v$, depending only on $\th$, $\mass$, $c_v$,
$K_\mu$ and $\cH$, such that
\begin{itemize}[leftmargin=*, topsep=2pt, itemsep=2pt, parsep=0pt]
\item
$
\cV_0(\psi_*,\psi)=\half\int_{\cX_0^4} dx_1 \cdots dx_4\  
                  V_0(x_1,x_2,x_3,x_4)\,  
                  \psi_*(x_1) \psi(x_2)\psi_*(x_3) \psi(x_4)
$ where $V_0$ is the spatial periodization of $\bV_0$,
\item
$\|\bV_0\|_{\frac{2}{3}\mass}\le K_v  \fv$  and
${\|\bV_0\|}_0\ge \sfrac{1}{K_v}  \fv$
\item
 $\Big|\half\int dx_2\, dx_3\, dx_4\ V_0(x_1,\cdots,x_4)
  - \th\int_X d\bx\ v(\bZ,\bx)\Big|\le K_v\fv^{2-2e_\R-4e_\r}$
\item
$
\big\|\bV_0-\de_{x_{1,0},x_{3,0}}\,\de_{x_{2,0},x_{4,0}}\,
         \de_{x_{1,0},x_{2,0}-1} \,\bv_\th(\bx_1,\bx_2,\bx_3,\bx_4)
\big\|_{\frac{2}{3}\mass}\le K_v\fv^{2-2e_\R-4e_\r} 
$
with 
\begin{align*}
&\bv_\th(\bx_1,\cdots,\bx_4)\\
&\hskip0.15in=
\sum_{\bx,\by\in \bbbz^3}\hskip-4pt
{\tst\int_0^\th} dt\  e^{-t\bh}(\bx,\bx_1)\,e^{-(\th-t)\bh}(\bx,\bx_2)\ 
           \bv(\bx,\by)\  e^{-t\bh}(\by,\bx_3)\,e^{-(\th-t)\bh}(\by,\bx_4)\\
&\hskip0.2in +  \sum_{\bx,\by\in \bbbz^3}\hskip-4pt
{\tst\int_0^\th} dt\  e^{-t\bh}(\bx,\bx_3)\,e^{-(\th-t)\bh}(\bx,\bx_2)\ 
          \bv(\bx,\by)\ e^{-t\bh}(\by,\bx_1)\,e^{-(\th-t)\bh}(\by,\bx_4)
\end{align*}
\end{itemize}
%%%%%%%%%%%%%%%%%%%%%%%%%%%%%%%
\item
 $\big|\mu_0 -\big(1-e^{-\th\mu}\big)\big|
   \le K\th \fv^{\mass\log\sfrac{1}{\fv}}$
%%%%%%%%%%%%%%%%%%%%%%%%%%%%%%
\item
 $\cR_0(\psi_*,\psi) = \tilde\cR_\In\big(
                 (\psi_*, \{\partial_\nu\psi_*\}), (\psi,\{\partial_\nu\psi\})
                  \big)+ \cR_0^{(6)}(\psi_*,\psi)\ $\\
where $\tilde\cR_\In((\psi_*, \{\psi_{*\nu}\}), (\psi,\{\psi_\nu\})$ 
is an $\fS$ invariant, particle--number preserving function with real valued 
kernels that    
\begin{itemize}[leftmargin=*, topsep=2pt, itemsep=2pt, parsep=0pt]
\item is of degree two in the fields, with either one $\psi_{(*)}$
field and one $\psi_{(*)0}$ field or two $\psi_{(*)\nu}$ fields, with both
having $1\le \nu\le 3$ and
\item obeys the bound 
$
\|\tilde \cR_\In\|_{\mass}\le C_r \fv^{\mass\log\sfrac{1}{\fv}}
$
with a constant $C_r$ that depends only on $\mass$, $K_\mu$ and $K$.
\end{itemize}
and $\cR_0^{(6)}(\psi_*,\psi)$ 
is an $\fS$ invariant, particle--number preserving function with real valued kernel
\begin{itemize}[leftmargin=*, topsep=2pt, itemsep=2pt, parsep=0pt]
\item
that thas degree three both in $\psi_*$ and $\psi$, and
\item
 fulfils the estimate
$
\| \cR_0^{(6)}\|_{2\mass}\le e^{-3\th\mu+6\mass} K\, (\fv\th)^{2-2e_\r}
$
\end{itemize}
%%%%%%%%%%%%%%%%%%%%%%%%%%%%%%
\item
 $\cE_0(\psi_*,\psi)$ 
is an $\fS$ invariant, particle--number preserving function with real valued kernels that
\begin{itemize}[leftmargin=*, topsep=2pt, itemsep=2pt, parsep=0pt]
\item 
is of degree at least four both in $\psi_*$ and in $\psi$, 
and 
\item
 has norm with mass $2\mass$ and weight $2e^{\th\mu/2-\mass}\big(\sfrac{1}{\th\fv}\big)^{e_\R+e_\r}$
at most
$
K\, (\fv\th)^{2-6e_\R-8e_\r}
$.
\end{itemize}
%%%%%%%%%%%%%%%%%%%%%%%%%%%%%%%%%%%%%%%%

\item
the ``small field cut off function'' $\,\chi_0(\psi)\,$ is one if 
\begin{itemize}[leftmargin=*, topsep=2pt, itemsep=2pt, parsep=0pt]
\item $|\psi(x)|
   \le e^{\th\mu/2}\big(\sfrac{1}{\th\fv}\big)^{e_\R+e_\r}$ 
       for all $x\in \cX_0$
and
\item $|\partial_\nu\psi(x)|
 \le  e^{\th\mu/2}\big(\sfrac{1}{\th}\big)^{e_{\RP}}
                \big(\sfrac{1}{\th\fv}\big)^{e_\r}$ 
for all  $1\le\nu\le 3$ and all $x \in\cX_0$ and
\item $|\partial_0\psi(x)|
 \le  e^{\th\mu/2}\big(\sfrac{1}{\th\fv}\big)^{e_\r}$ 
   for all $x\in \cX_0$
\end{itemize}
and is zero otherwise.
\end{itemize}
\end{proposition}

\begin{proof} We start by defining a field $\psi$ on the lattice 
$\cX_0$ by
\begin{equation*}
\psi(x_0,\bx) = e^{\th\mu/2}\al_{\th x_0}(\bx)
\end{equation*}
Making a change of variables from $\al_\tau(\bx)$ to $\psi(x)$
converts the integral on the left hand side to
\begin{align*}
 \cZ_\th^{|\cX_0|} e^{-\th\mu|\cX_0|}\int \Big[ \smprod_{x\in \cX_0} 
        \sfrac{ d\psi(x)^\ast\wedge d\psi(x)}{2\pi \imath}\Big] \,
  e^{\cF(\psi^*,\psi) }\chi_0(\psi)
\end{align*}
where
\begin{align}
\cF(\psi_*,\psi)
  &=\hskip-10pt\sum_{\tau\in\th\bbbz\cap(0,\frac{1}{kT}]} \hskip-17pt
    \big\{\!-\< \al_{*\tau},\, \al_\tau \>_X \!
   +\< \al_{*\tau-\th},\,j(\th) \al_\tau \>_X  
    +V_\th(\al_{*\tau-\th},\al_\tau)
%\\\noalign{\vskip-0.15in} &\hskip3in
    +\cD_\th(\al_{*\tau-\th},\al_\tau)\big\}\nonumber\\
%%%
  &=-e^{-\th\mu}\<\psi_*,\psi\>_{\cX_0}
   +\sum_{x_0\in\bbbz\cap(0,\frac{1}{\th kT}]}  \Big\{
          \<\psi_*(x_0-1,\,\cdot\,)\ ,\ 
             \big(e^{-\th h}\psi\big)(x_0,\,\cdot\,)\>_X\nonumber\\ &\hskip2in
      + e^{-2\th\mu}V_\th\big(\psi_*(x_0-1,\,\cdot\,),\psi(x_0,\,\cdot\,)\big)
            \nonumber\\ &\hskip2in
    +\cD_\th\big(e^{-\th\mu/2}\psi_*(x_0-1,\,\cdot\,),
                     e^{-\th\mu/2}\psi(x_0,\,\cdot\,)\big)\Big\}\nonumber\\
%%%
  &=-\sum_{x\in\cX_0} \psi_*(x)\, 
       \big([\bbbone - e^{-\th h} -e^{-\th h} \partial_0]\psi\big)(x)
       +(1-e^{-\th\mu})\<\psi_*,\psi\>_{\cX_0} \nonumber\\
   &\hskip0.5in
   -\cV_\In(\psi_*,\psi) 
    +\cR_\In'(\psi_*,\psi)+\cE_\In(\psi_*,\psi)
\label{eqnSZfirstA}
\end{align}
with
\begin{align*}
\cV_\In(\psi_*,\psi) 
&= -\sum_{x_0\in\bbbz\cap(0,\frac{1}{\th kT}]} 
e^{-2\th\mu}\,V_\th\big(\psi_*(x_0-1,\,\cdot\,),\psi(x_0,\,\cdot\,)\big)
\\ %%%%%%%%%%%%%%%%%
\cR_\In'(\psi_*,\psi)
&=  \sum_{x_0\in\bbbz\cap(0,\frac{1}{\th kT}]}
    e^{-\th\mu}\,\cR_\th\big(\psi_*(x_0-1,\,\cdot\,),\psi(x_0,\,\cdot\,)\big) 
\\  %%%%%%%%%%%%%%%%%%
\cE_\In(\psi_*,\psi)
&=  \sum_{x_0\in\bbbz\cap(0,\frac{1}{\th kT}]} 
\cE_\th\big(e^{-\th\mu/2}\psi_*(x_0-1,\,\cdot\,),
                     e^{-\th\mu/2}\psi(x_0,\,\cdot\,)\big)
\end{align*}
All of $\cV_\In$, $\cR_\In'$, $\cE_\In$ are invariant under $\fS$
and have real--valued kernels.

Observe that
\begin{align*}
\cV_\In(\psi_*,\psi)
&= e^{-2\th\mu}\hskip-15pt
   \sum_{\atop{x_0\in\bbbz\cap(0,\frac{1}{\th kT}]}
              {\atop{\bx,\by\in X}{\bx_1,\bx_2,\bx_3,\bx_4\in X}}
        }\hskip-15pt
{\tst\int_0^\th} dt\hskip5pt j(t)(\bx,\bx_1)\psi_*(x_0-1,\bx_1) \ 
               j(\th-t)(\bx,\bx_2)\psi(x_0,\bx_2) 
               \\ \noalign{\vskip-0.35in}&\hskip1.3in
   v(\bx,\by)\   j(t)(\by,\bx_3)\psi_*(x_0-1,\bx_3) \ 
               j(\th-t)(\by,\bx_4)\psi(x_0,\bx_4) \\\noalign{\vskip0.05in}
&=\half\int_{\cX_0^4}dx_1\cdots dx_4\ V_\In(x_1,\cdots,x_4)\ \psi_*(x_1)\psi(x_2)
               \psi_*(x_3)\psi(x_4)
\end{align*}
where $V_\In$ is the spatial periodization of
\begin{equation*}
\begin{split}
\bV_\In(x_1,\!\cdots\!,x_4)\!
%&= \de_{x_{1,0},x_{3,0}}\,\de_{x_{2,0},x_{4,0}}\,\de_{x_{1,0},x_{2,0}-1}\\
%&\hskip0.1in \bigg\{ \hskip-1pt\sum_{\bx,\by\in \bbbz^3}\hskip-4pt
%{\tst\int_0^\th} dt\  e^{-t\bh}(\bx,\bx_1)\,e^{-(\th-t)\bh}(\bx,\bx_2)\ \bv(\bx,\by)\               e^{-t\bh}(\by,\bx_3)\,e^{-(\th-t)\bh}(\by,\bx_4)\\
%&\hskip0.1in \hskip-2pt +\hskip-5pt\sum_{\bx,\by\in \bbbz^3}\hskip-4pt
%{\tst\int_0^\th} dt\  e^{-t\bh}(\bx,\bx_3)\,e^{-(\th-t)\bh}(\bx,\bx_2)\ \bv(\bx,\by)\               e^{-t\bh}(\by,\bx_1)\,e^{-(\th-t)\bh}(\by,\bx_4)\bigg\}\\
&=\de_{x_{1,0},x_{3,0}}\,\de_{x_{2,0},x_{4,0}}\,
         \de_{x_{1,0},x_{2,0}-1} \,\bv_\th(\bx_1,\bx_2,\bx_3,\bx_4)
\end{split}
\end{equation*}
As in \cite[Lemma 3.21]{UV}, 
\begin{equation}\label{eqnSZvinupper}
\|\bV_\In\|_{5\mass}\le 2 \th\, e^{2K_j\th}e^{5\mass} 
                      \tn \bv\tn
\end{equation}
By translation invariance
\begin{align*}
&\half\int dx_2\, dx_3\, dx_4\ V_\In(x_1,\cdots,x_4)\\
&\hskip1in=\Big[{\tst\int_X} d\bx\ v(\bZ,\bx)\Big]
       \int_0^\th dt\ \Big[{\tst\int_X} d\bx\ e^{-t h}(\bZ,\bx)\Big]^2
                    \Big[{\tst\int_X} d\bx\ e^{-(\th-t) h}(\bZ,\bx)\Big]^2\\
&\hskip1in=\th \int_X d\bx\ v(\bZ,\bx)
\end{align*}
since, using $\hat h$ to denote the Fourier transform of $ h$, 
\begin{equation*}
\int_X d\bx\ e^{-\tau h}(\bZ,\bx) =e^{-\tau\hat h (\bZ)}=1
\end{equation*}
Similarly
\begin{equation}\label{eqnSZvinlower}
\begin{split}
{\|\bV_\In\|}_0
&\ge \int dx_2\, dx_3\, dx_4\ \bV_\In(x_1,\cdots,x_4)
=2\th \int_{\bbbz^3} d\bx\ \bv(\bZ,\bx)
=2\th \int_X d\bx\ v(\bZ,\bx)\\
&=2\th\sfrac{\<1, v 1\>_{L^2(X)} }    {\<1,1\>_{L^2(X)}} \\
&\ge \sfrac{\th c_v}{2}\fv
\end{split}
\end{equation}

We now move onto a discussion of $\cR_\In'$. By the bound on $\cR_\th$
following \eqref{eqnSZspa}
\begin{align*}
\|\cR_\In'\|_{2\mass}
&\le e^{2\mass} e^{-\th\mu} \|\cR_\th\|_{2\mass}
\le e^{2\mass} e^{-\th\mu}\sfrac{(\th\fv)^{2(e_\R+e_\r)}}{4}\ 
        K\th\fv^{\mass\log\sfrac{1}{\fv}}
\end{align*}
By \cite[Lemma \lemLlocalize.c]{PAR2},
%[\MS, Lemma \lemMXlocalize.b]
\begin{align*}
\cR_\In'(\psi_*,\psi)
&=\bigg[\sum_{x_0\in\bbbz\cap(0,\frac{1}{\th kT}]} e^{-\th\mu}\,\cR_\th\big(\psi_*(x_0,\,\cdot\,),\psi(x_0,\,\cdot\,)\big)\bigg]\\
&\hskip1in
+\bigg[\sum_{x_0\in\bbbz\cap(0,\frac{1}{\th kT}]} e^{-\th\mu}\,\cR_\th\big(\psi_*(x_0,\,\cdot\,),(\partial_0\psi)(x_0,\,\cdot\,)\big)\bigg]
\\
&=\De\mu\int_{\cX_0}dx\ \psi_*(x)\psi(x) 
    \ +\ \tilde\cR_\In\big(
                 (\psi_*, \{\partial_\nu\psi_*\}), (\psi,\{\partial_\nu\psi\})
                  \big)
\end{align*}
with a real number $\De\mu$ obeying 
             $|\De\mu|\le K\th\fv^{\mass\log\frac{1}{\fv}}$ 
and a function $\tilde\cR_\In$
that has the properties specified in the statement of the proposition.
(The contribution with one time derivative and one space derivative 
that is allowed by \cite[Lemma \lemLlocalize.c]{PAR2} vanishes in this case 
since the time arguments of the two fields in  
$\cR_\th\big(\psi_*(x_0,\,\cdot\,),\psi(x_0,\,\cdot\,)\big)$ are always equal.)

Next, we discuss $\cE_\In$. If the part of $\cE_\th$ that is homogeneous
of degree $2n$ has kernel $E_{\th,n}(\bx_1,\cdots,\bx_n;\by_1,\cdots,\by_n)$,
then the part of $\cE_\In$  that is homogeneous of degree $2n$ has kernel
\begin{align*}
E_{\In,n}(x_1,\cdots,x_n;y_1,\cdots,y_n)
&=E_{\th,n}(\bx_1,\cdots,\bx_n;\by_1,\cdots,\by_n) e^{-n\th\mu}\\ &\hskip1in
\Big[\smprod_{i=1}^n \de{x_{1,0},x_{i,0}}\Big]
\Big[\smprod_{i=1}^n \de{x_{i,0},y_{i,0}-1}\Big]
\end{align*}
Let $\cT$ be a shortest tree on $X$ having $\bx_1$, $\cdots$, $\bx_n$, $\by_1$,
$\cdots$, $\by_n$ among its vertices. Then the tree on $\cX_0$
whose edges are
\begin{itemize}[leftmargin=*, topsep=2pt, itemsep=2pt, parsep=0pt]
\item
 $\{(x_{1,0},\bz)\,,\,(x_{1,0},\bz')\}$ if $\{\bz,\bz'\}$ is an edge
of $\cT$
\item
 and an edge from $(x_{1,0},\by_i)$ to $y=(x_{1,0}+1,\by_i)$ for 
each $1\le i\le n$
\end{itemize}
has $x_1$, $\cdots$, $x_n$, $y_1$, $\cdots$, $y_n$ among its vertices
so that
\begin{equation*}
\tau\big(x_1,\cdots,x_n,y_1,\cdots,y_n)
\le  \tau\big(\bx_1,\cdots,\bx_n,\by_1,\cdots,\by_n) +n
\end{equation*}
Consequently,
\begin{equation*}
\|E_{\In,n}\|_{2\mass}\le \|E_{\th,n}\|_{2\mass} e^{-n\th\mu} e^{2\mass n}
\end{equation*}
and the norm of $\cE_\In$ with mass $2\mass$ and weight $\ka_\In$ is bounded
by the norm of $\cE_\th$ with mass $2\mass$ and weight $\ka$ provided
$\ka_\In e^{-\th\mu/2+\mass}\le\ka$. So the norm,
with mass $2\mass$ and weight $\ka_\In$, of $\cE_\In$, obeys the bound
quoted for $\cE_\th$ (but with mass $2\mass$ and weight $\ka$)
following \eqref{eqnSZspa}, if we choose 
$\ka_\In = e^{\th\mu/2-\mass}\ka
     =e^{\th\mu/2-\mass}2\big(\sfrac{1}{\th\fv}\big)^{e_\R+e_\r}.$

Denote by $\cE_4$ and $\cR_0^{(6)}$ the two monomials in $\cE_\In$ that are of 
degree four and six, respectively, and set
$$
\mu_0=\big(1-e^{-\th\mu}\big) +\De\mu\qquad
\cV_0=\cV_\In-\cE_4\qquad
\cE_0=\cE_\In-\cE_4-\cR_0^{(6)}
$$
and
$$
\cR_0(\psi_*,\psi) = \tilde\cR_\In\big(
                 (\psi_*, \{\partial_\nu\psi_*\}), (\psi,\{\partial_\nu\psi\})
                  \big)+ \cR_0^{(6)}(\psi_*,\psi)
$$
Obviously
\begin{equation*}
-\cV_\In+\cR_\In+\cE_\In=-\cV_0+\cR_0+\cE_0
\end{equation*}
and $\cR_0$ and $\cE_0$ have the desired properties.
Except for the definition and properties of $\bV_0$,
the Proposition now follows from \eqref{eqnSZfirstA} and the discussion above.

Set
\begin{align*}
F=\Big\{\ x_1,x_2,x_3,x_4\in (\bbbz/\sfrac{1}{\th kT}\bbbz) \times \bbbz^3
   \ \Big|\ 
        &\sfrac{L_\sp}{2}< x_{i,j}-x_{1,j} <\sfrac{L_\sp}{2} \\
   \noalign{\vskip-0.05in}
        &\text{\ for all $i=2,3,4$ and $j=1,2,3$}
   \ \Big\}
\end{align*}
Here $x_{i,j}$ is the $j^{\rm th}$ (spatial) coordinate of $x_i$.
 Set,
for $x_1,x_2,x_3,x_4\in (\bbbz/\sfrac{1}{\th kT}\bbbz) \times \bbbz^3$
\begin{equation*}
\bE_4(x_1,x_2,x_3,x_4)
=\begin{cases}E_4\big([x_1],\cdots,[x_4]\big)
      &\text{if $(x_1,x_2,x_3,x_4)\in F$}\\
           0 & \text{otherwise}
      \end{cases}
\end{equation*}
where $E_4$ is the kernel of $\cE_4$.
Then $E_4$ is the spatial periodization of $\bE_4$. 
Define
\begin{equation*}
\bV_0 = \bV_\In - \text{symmetrization of }\bE_4
\end{equation*}
It remains only to
prove that $\|\bE_4\|_{\frac{2}{3}\mass}\le \|E_4\|_{2\mass}$.
The desired bounds on $\bV_0$ will then follow from \eqref{eqnSZvinupper}
and \eqref{eqnSZvinlower}.

Denote by $\tilde\tau$ the tree length in 
$(\bbbz/\sfrac{1}{\th kT}\bbbz) \times \bbbz^3$. See 
Definition \ref{defHTkernelnorm}. If we have $(x_1,x_2,x_3,x_4)\in F$, 
then
\begin{align*}
\tilde\tau(x_1,x_2,x_3,x_4)
&\le |x_2-x_1| + |x_3-x_1| + |x_4-x_1| 
\le 3\tau([x_1],[x_2],[x_3],[x_4])
\end{align*}
It follows that $\|\bE_4\|_{\frac{2}{3}\mass}\le \|E_4\|_{2\mass}$.
\end{proof}

In the setting of Proposition \ref{propSZprepforblockspin},
the $\mu_*$ of \eqref{eqnINTmustardef} has a particularly simple form.
We thank Martin Lohmann for pointing this out.

\begin{lemma}\label{lemSZexplicitmustar}
We have
\begin{align*}
\mu_*
&= 2\th  \int_{\bbbr^3 /2\pi\bbbz^3}  \sfrac{d^3\bk}{(2\pi)^3} 
\sfrac{\hat\bv(\bZ)+\hat \bv(\bk) }{ e^{\be\hat\bh(\bk)}-1}
\ +\ O(\fv^{2-2e_\R-4e_\r} )
\end{align*}
where $\hat\bh(\bk)$ and $\hat \bv(\bk)$ are the Fourier transforms 
of $\bh(\bx,\bZ)$ and $\bv(\bx,\bZ)$ and $\be=\sfrac{1}{kT}$.
\end{lemma}
\begin{proof}
As a preliminary calculation, we evaluate
\begin{equation*}
\begin{split}
&\sum_{\bx_2,\bx_3\in\bbbz^3}
          e^{-t\bh}(\bx,\bx_2)\bD_0^{-1}(x_3,x_2)e^{-(\th-t)\bh}(\by,\bx_3)\\
&\hskip1in=\big(e^{-(\th-t)\bh} \bD_0^{-1} e^{-t\bh} \big)
                          \big((x_{3,0},\by)\,,\,(x_{2,0},\bx)\big)\\
&\hskip1in=\big(\bD_0^{-1} e^{-\th\bh}\big)
                          \big((x_{3,0},\by)\,,\,(x_{2,0},\bx)\big)\\
%&\hskip1in=\big(
%         (\bbbone - e^{-\th\bh} -e^{-\th\bh} \partial_0)^{-1} e^{-\th\bh}\big)
%                         \big((x_{3,0},\by)\,,\,(x_{2,0},\bx)\big)\\
&\hskip1in=\big(e^{\th\bh} -\bbbone - \partial_0\big)^{-1}
                          \big((x_{3,0},\by)\,,\,(x_{2,0},\bx)\big)
\end{split}
\end{equation*}
where
\begin{itemize}[leftmargin=*, topsep=2pt, itemsep=2pt, parsep=0pt]
\item
 we have used that $\bh$ is a symmetric operator
\item
 we are thinking of $e^{-\tau\bh}(\bx,\by)$ as being tensored with
an identity operator in the temporal arguments $x_0,y_0$ 
\item
 we have used that $e^{-\tau \bh}$ and  $\bD_0^{-1}$ are both
translation invariant operators and hence commute with each other
\item
the operator $\big(e^{\th\bh} -\bbbone - \partial_0\big)^{-1}$
is the inverse of $e^{\th\bh} -\bbbone - \partial_0=e^{\th\bh}\bD_0$
acting on the space 
$L^2\big((\bbbz/\sfrac{1}{\th kT}\bbbz) \times \bbbz^3\big)$.
\end{itemize}
Since
$$
\sum_{\by'\in\bbbz^3} e^{-\tau\bh}_X(\bx',\by')=e^{-\tau\hat\bh (\bZ)}=1
$$
we have (recalling the definition of $\bv_\th$ from 
Proposition \ref{propSZprepforblockspin}
\begin{align*}
&\sum_{\bx_1,\bx_2,\bx_3\in\bbbz^3} \bv_\th(\bZ,\bx_1,\bx_2,\bx_3)\ \bD_0^{-1}(x_3,x_2) \\
&\hskip0.2in
=  \int_0^\th dt\sum_{\bx,\by\in\bbbz^3}
        v(\bx,\by)\,\Big\{
         e^{-t\bh}(\bx,\bZ)
           \Big[\big(e^{\th\bh}\!-\!\bbbone\!-\! \partial_0\big)^{-1}
                          \big((x_{3,0},\by)\,,\,(x_{2,0},\by)\big)\Big]
         \\&\hskip1.7in +
        e^{-t\bh}(\by,\bZ)
       \Big[\big(e^{\th\bh}\!-\!\bbbone\!-\! \partial_0\big)^{-1}
                          \big((x_{3,0},\by)\,,\,(x_{2,0},\bx)\big)\Big]
          \Big\}
      \\
&\hskip0.2in
= \th\sum_{\bx\in\bbbz^3} v(\bZ,\bx)\,\Big\{
         \big(e^{\th\bh} -\bbbone - \partial_0\big)^{-1}\!
                   \big((x_{3,0},\bZ)\,,\,(x_{2,0},\bZ)\big)
\\&\hskip1.7in
      +\big(e^{\th\bh} \!-\!\bbbone \!-\! \partial_0\big)^{-1}
                          \big((x_{3,0},\bx)\,,\,(x_{2,0},\bZ)\big)\Big\}
\end{align*}
Recall from Proposition \ref{propSZprepforblockspin} that
\begin{equation*}
\big\|\bV_0-\de_{x_{1,0},x_{3,0}}\,\de_{x_{2,0},x_{4,0}}\,
         \de_{x_{1,0},x_{2,0}-1} \,\bv_\th(\bx_1,\bx_2,\bx_3,\bx_4)
\big\|_{\frac{2}{3}\mass}\le K_v\fv^{2-2e_\R-4e_\r} 
\end{equation*}
As $\bD_0^{-1}(x_3,x_2)$ is bounded, we have
\begin{equation}\label{eqnSZmustarformB}
\begin{split}
\mu_*
&= 2\th\sum_{\bx\in\bbbz^3} v(\bZ,\bx)\Big\{
              \big(e^{\th\bh} -\bbbone - \partial_0\big)^{-1}\!
                   \big((1,\bZ)\,,\,0\big)%\\&\hskip2.2in
      +         \big(e^{\th\bh} - \bbbone - \partial_0\big)^{-1}
                          \big((1,\bx)\,,\,0\big)\Big\}\\
&\hskip3.7in+O(\fv^{2-2e_\R-4e_\r} )
\end{split}
\end{equation}

On the other hand
\begin{equation}\label{eqnSZmustarformC}
\begin{split}
\big(e^{\th\bh} -\bbbone - \partial_0\big)^{-1}\big((1,\bx),0\big) 
&= \int_{\bbbr^3 /2\pi\bbbz^3}  \sfrac{d^3\bk}{(2\pi)^3} 
   e^{ i\bk \bx}e^{-\th\bh(\bk)}  \
 \sfrac{1}{L_\tp}\!\!\!\!\!\!
   \sum_{k_0\in \frac{2\pi}{L_\tp} \bbbz/2\pi \bbbz}  
\sfrac{e^{ik_0 } }{1-  e^{ik_0}e^{-\hat\bh_0(\bk)} }  
\\
&=  \int_{\bbbr^3 /2\pi\bbbz^3}  \sfrac{d^3\bk}{(2\pi)^3}  
  e^{ i\bk \bx}e^{-\th\bh(\bk)}\
\sfrac{ e^{-(L_\tp-1)\hat\bh_0(\bk)} }{1- e^{-L_\tp\hat\bh_0(\bk)}} 
\\
&=  \int_{\bbbr^3 /2\pi\bbbz^3}  \sfrac{d^3\bk}{(2\pi)^3}  e^{ i\bk \bx}\
\sfrac{ e^{-\be\hat\bh(\bk)} }{1- e^{-\be\hat\bh(\bk)}} 
\end{split}
\end{equation}
Here, we applied Lemma \ref{lemSZprimsum}, below, with $p=L_\tp$,
$\ze = e^{\frac{2\pi i}{L_\tp}}$ and $w= e^{-\hat\bh_0(\bk)}$,
to the $k_0$ sum. The Lemma now follows by combining \eqref{eqnSZmustarformB}
and \eqref{eqnSZmustarformC}.

\end{proof}

\begin{corollary}\label{corSZstepzeroConditions}
With $\bH=\th\cH$, the data of Proposition \ref{propSZprepforblockspin} fulfill
the conditions of \S \ref{sectINTstartPoint} provided $\fv$ is small enough 
(depending on $\eps$ and $\th$), $\mass\ge 3m $ and
\begin{equation*}
2  \int_{\bbbr^3 /2\pi\bbbz^3}  \sfrac{d^3\bk}{(2\pi)^3} 
\sfrac{\hat\bv(\bZ)+\hat \bv(\bk) }{ e^{\be\hat\bh(\bk)}-1}
    + \fv^{\sfrac{4}{3}-2\eps} 
<  \mu < \fv^{\sfrac{8}{9}+2\eps} 
\end{equation*}
and 
$
\sfrac{1}{3}-\eps < e_\R < \sfrac{1}{3} -\sfrac{5}{6}\eps$
and
$
e_\r < \sfrac{\eps}{2}
$. 
Observe that
\begin{equation*}
\lim_{\be\rightarrow\infty} \int_{\bbbr^3 /2\pi\bbbz^3}  \sfrac{d^3\bk}{(2\pi)^3} 
\sfrac{\hat\bv(\bZ)+\hat \bv(\bk) }{ e^{\be\hat\bh(\bk)}-1}
=0
\end{equation*}
\end{corollary}
\begin{proof} By the definition of $\fv_0$ in \S \ref{sectINTstartPoint} and
Proposition \ref{propSZprepforblockspin},
\begin{equation*}
\sfrac{2}{K_v}\fv
\le 2{\|\bV_0\|}_0
\le\fv_0
=2\|\bV_0\|_{2m}
\le 2\|\bV_0\|_{\frac{2}{3}\mass}
\le 2 K_v\fv
\end{equation*}
The condition on $\mu_0$ in \S \ref{sectINTstartPoint} is satisfied since
$\th\mu= \mu_0+O(\mu^2) +O(\fv^{\mass\log\frac{1}{\fv}})$.

\end{proof}

\begin{lemma}\label{lemSZprimsum}
Let $\ze$ be a primitive $p^{th}$ root of unity and $w\in \bbbc$ 
not be a $p^{th}$ root of unity. Then
\begin{equation*}
\sfrac{1}{p}\smsum_{k=0}^{p-1} \sfrac{\ze^k}{1- w\ze^k} =  \sfrac{ w^{p-1}}{1-w^p}
\end{equation*}
\end{lemma}
\begin{proof} First consider the case $0<w<1$.
Expanding the geometric series and interchanging sums
\begin{equation*}
\sfrac{1}{p}\smsum_{k=0}^{p-1} \sfrac{\ze^k}{1- w\ze^k} 
 = \smsum_{n=0}^\infty \sfrac{1}{p} \smsum_{k=0}^{p-1} w^n \ze^{(n+1)k}
\end{equation*}
Now
\begin{equation*}
\sfrac{1}{p} \smsum_{k=0}^{p-1}  \ze^{(n+1)k}
= \begin{cases} 0 & \text{if $n+1$ is not an integer multiple of $p$}\\
             \noalign{\vskip0.05in}
         1 & \text{if $n= mp-1 $ for some integer $m$}
  \end{cases}
\end{equation*}
If $n\ge 0$, the integer $m$ above has to be at least one. Therefore
\begin{equation*}
\sfrac{1}{p}\smsum_{k=0}^{p-1} \sfrac{\ze^k}{1- w\ze^k} 
 = \smsum_{m=1}^\infty   w^{mp-1}
=  \sfrac{ w^{p-1}}{1-w^p}
\end{equation*}
The claim now follows by analytic continuation in $w$.
\end{proof}

\newpage
%%%%%%%%%%%%%%%%%%%%%%%%%%%%%%%%%%%%%%
\bibliographystyle{plain}
\bibliography{refs}
%%%%%%%%%%%%%%%%%%%%%%%%%%%%%%%%%%%%%

%\printIssueCount
\end{document}